\documentclass[lettersize,journal]{IEEEtran}
\usepackage{color}
\usepackage{amsmath,amsthm,amssymb,amsfonts}

\allowdisplaybreaks[4]

\usepackage[T1]{fontenc}
\usepackage{graphicx}
\usepackage{longtable}
\usepackage{float}
\usepackage{hyperref}

\usepackage{bm}

\usepackage[bottom]{footmisc}

\usepackage[mathscr]{eucal}

\usepackage[caption=false,font=normalsize,labelfont=sf,textfont=sf]{subfig}

\usepackage{multirow}

\usepackage{cite}

\usepackage{epsfig}

\usepackage{soul}

\usepackage[lined,boxed,ruled]{algorithm2e}

\newcommand{\barjmath}{\bar{\jmath}}
\newtheorem{theorem}{Theorem}
\newtheorem{lemma}{Lemma}

\newtheorem{corol}{Corollary}

\newcommand{\rank}[1]{\mathrm{rank} \left(  #1 \right)  }
\newcommand{\funf}[1][\ba,\bA]{\mathbf{f}\left( #1\right)}

\newcommand{\Exp}{{\mathbb{E}}}
\newcommand{\Expect}[2]{\Exp_{#1}\left\lbrace #2 \right\rbrace}

\newcommand{\braces}[1]{\left\lbrace #1\right\rbrace}

\newcommand{\Dkl}[2]{D_{\textrm{KL}}\left\lbrace #1 : #2\right\rbrace}

\newcommand{\setposi}[1]{\mathcal{Z}_{#1}^+}

\newcommand{\setnnega}[1]{\mathcal{Z}_{#1}}

\newcommand{\tr}[1]{\mathrm{tr}\left\lbrace #1\right\rbrace }
\newcommand{\diag}[1]{\mathrm{diag}\left\lbrace #1\right\rbrace }
\newcommand{\Diag}[1]{\mathrm{Diag}\left\lbrace #1\right\rbrace }

\newcommand{\logtwo}[1]{\log_{2}\left(#1\right)}

\newcommand{\intd}{\mathrm{d}}
\newcommand{\intdx}[1]{\intd #1}

\newcommand{\argmin}[1]{\mathop{\arg\min}\limits_{#1}}

\newcommand{\toinf}[1]{#1 \to\infty}

\newcommand{\liminfty}[1]{\lim_{\toinf{#1}}}

\newcommand{\mtxvec}[1]{\mathrm{vec}\left\lbrace #1\right\rbrace }

\newcommand{\vecc}{\mathrm{vec}}

\newcommand{\equaa}{\mathop{=}^{(\textrm{a})}}
\newcommand{\equab}{\mathop{=}^{(\textrm{b})}}
\newcommand{\equac}{\mathop{=}^{(\textrm{c})}}
\newcommand{\equad}{\mathop{=}^{(\textrm{d})}}
\newcommand{\equae}{\mathop{=}^{(\textrm{e})}}
\newcommand{\equaf}{\mathop{=}^{(\textrm{f})}}

\newcommand{\expb}[1]{\exp \left\lbrace  #1 \right\rbrace}

\newcommand{\ba}{\mathbf{a}}
\newcommand{\bb}{\mathbf{b}}
\newcommand{\bc}{\mathbf{c}}
\newcommand{\bd}{\mathbf{d}}
\newcommand{\be}{\mathbf{e}}

\newcommand{\bg}{\mathbf{g}}
\newcommand{\bh}{\mathbf{h}}

\newcommand{\bp}{\mathbf{p}}
\newcommand{\bq}{\mathbf{q}}
\newcommand{\br}{\mathbf{r}}
\newcommand{\bs}{\mathbf{s}}
\newcommand{\bt}{\mathbf{t}}
\newcommand{\bu}{\mathbf{u}}
\newcommand{\bv}{\mathbf{v}}

\newcommand{\bx}{\mathbf{x}}
\newcommand{\by}{\mathbf{y}}
\newcommand{\bz}{\mathbf{z}}

\newcommand{\bA}{\mathbf{A}}
\newcommand{\bB}{\mathbf{B}}

\newcommand{\bD}{\mathbf{D}}
\newcommand{\bE}{\mathbf{E}}
\newcommand{\bF}{\mathbf{F}}
\newcommand{\bG}{\mathbf{G}}
\newcommand{\bH}{\mathbf{H}}
\newcommand{\bI}{\mathbf{I}}
\newcommand{\bJ}{\mathbf{J}}
\newcommand{\bK}{\mathbf{K}}
\newcommand{\bL}{\mathbf{L}}
\newcommand{\bM}{\mathbf{M}}

\newcommand{\bO}{\mathbf{O}}
\newcommand{\bP}{\mathbf{P}}
\newcommand{\bQ}{\mathbf{Q}}
\newcommand{\bR}{\mathbf{R}}
\newcommand{\bS}{\mathbf{S}}
\newcommand{\bT}{\mathbf{T}}
\newcommand{\bU}{\mathbf{U}}
\newcommand{\bV}{\mathbf{V}}

\newcommand{\bX}{\mathbf{X}}
\newcommand{\bY}{\mathbf{Y}}
\newcommand{\bZ}{\mathbf{Z}}

\newcommand{\bbC}{\mathbb{C}}
\newcommand{\bbR}{\mathbb{R}}

\newcommand{\bzero}{\mathbf{0}}
\newcommand{\bone}{\mathbf{1}}

\newcommand{\bSigma}{{\boldsymbol\Sigma}}

\newcommand{\bLambda}{{\boldsymbol\Lambda}}

\newcommand{\bOmega}{{\boldsymbol\Omega}}
\newcommand{\bomega}{{\boldsymbol\omega}}

\newcommand{\bPi}{{\boldsymbol\Pi}}

\newcommand{\btheta}{{\boldsymbol\theta}}

\newcommand{\bgamma}{{\boldsymbol\gamma}}
\newcommand{\bmu}{{\boldsymbol\mu}}

\newcommand{\bxi}{{\boldsymbol\xi}}
\newcommand{\bvartheta}{{\boldsymbol\vartheta}}

\newcommand{\bbeta}{{\boldsymbol\eta}}

\newcommand{\bbnu}{{\boldsymbol\nu}}

\newcommand{\dnnot}[2]{#1_{\textrm{#2}}}

\newcommand{\normmm}[1]{{\left\vert\kern-0.25ex\left\vert\kern-0.25ex\left\vert #1 
		\right\vert\kern-0.25ex\right\vert\kern-0.25ex\right\vert}}

\newcommand{\norm}[1]{\lVert #1 \rVert}

\definecolor{myback}{RGB}{204,232,207}

\newcommand{\tildebvartheta}{\bvartheta}
\newcommand{\tildebtheta}{\btheta}
\newcommand{\tildebnu}{\bbnu}

\newcommand{\overlinebvartheta}{\bvartheta}
\newcommand{\overlinebtheta}{\btheta}
\newcommand{\overlinebnu}{\bbnu}

\usepackage{makecell}

\begin{document}

\title{Efficient Information Geometry Approach for Massive MIMO-OFDM Channel Estimation}	
\author{
	Jiyuan~Yang,~\IEEEmembership{Member,~IEEE,}
	~Yan~Chen,
	~Mingrui~Fan,
    ~An-An~Lu,~\IEEEmembership{Member,~IEEE,}
    ~Wen~Zhong,
    ~Xiqi~Gao,~\IEEEmembership{Fellow,~IEEE,}
    ~Xiaohu~You,~\IEEEmembership{Fellow,~IEEE,}
    ~Xiang-Gen~Xia,~\IEEEmembership{Fellow,~IEEE,}
    and~Dirk~Slock,~\IEEEmembership{Fellow,~IEEE}

    \thanks{A short version has been accepted in The 2023 IEEE 98th Vehicular Technology Conference (VTC2023-Fall) \cite{10333557}. Compared to the short version, we provide detailed proofs as well as analyses of the main results in this paper.}

}

\maketitle
\begin{abstract}
	We investigate the channel estimation for massive multiple-input multiple-output orthogonal frequency division
	multiplexing (MIMO-OFDM) systems.
	We revisit the information geometry approach (IGA) for massive MIMO-OFDM channel estimation.
	By using the constant magnitude property of the entries of the measurement
	matrix, we find that the second-order natural parameters of the distributions on all the auxiliary manifolds are equivalent to each other, and the first-order natural parameters are asymptotically equivalent to each other at the fixed point.
	Motivated by these results, we simplify the process of IGA
	and propose an efficient IGA (EIGA) for massive MIMO-OFDM channel estimation, which allows efficient implementation with fast Fourier transformation (FFT).
    We then establish a sufficient condition of its convergence and accordingly find a range of the damping factor for the convergence. 
    We show that this range of damping factor is sufficiently wide by using the specific properties of the measurement matrices.
	Further, we prove that at the fixed point, the \textsl{a posteriori} mean obtained by EIGA is asymptotically optimal.
	Simulations confirm that EIGA can
	achieve the optimal performance with low complexity in a limited number of iterations.
\end{abstract}

\section{Introduction}
Massive multiple-input multiple-output (MIMO) combined with orthogonal frequency division multiplexing (OFDM) can provide tremendous gains in both capacity and energy efficiency for communication systems. 
As a high-priority option, massive MIMO-OFDM has become a key enabling technique for $5$G systems and will play a critical role in future $6$G systems with the antenna number scale further increased \cite{8626085,10054381}.
To fully reap the various benefits of massive MIMO-OFDM, the accurate acquisition of the channel state information (CSI) is essential.
Pilot-aided channel estimation is the common channel estimation approach for practical systems, where the transmitter periodically sends the pilots, and the receiver estimates the CSI with the received pilot signal.
Given the received pilot signal, the task of channel estimation is to obtain the \textsl{a posteriori} information of the channel parameters.
With the Gaussian prior, the \textsl{a posteriori} distribution of the channel parameters is also Gaussian, of which the \textsl{a posteriori} information is determined by the mean vector and the covariance matrix.
Nonetheless, the large dimension of the channel matrix in massive MIMO-OFDM systems poses a great challenge in the acquisition of the \textsl{a posteriori} mean and covariance. 
The calculation of the optimal estimators, e.g., MMSE estimator, is usually unaffordable due to the large dimension matrix inverse operation.

In the past years, many works have been devoted to the channel estimation for massive MIMO-OFDM systems \cite{8298537,8425578,6940305,8678465,Xiaofengliu}.
Among them, Bayesian inference approaches, e.g., message passing, Bethe free energy minimization and etc, have attracted much attention due to their reliable performance and relatively low computational complexity.
One common solution in Bayesian inference is to calculate the marginals (or the approximation of marginals) of the \textsl{a posteriori} distribution, from which  the \textsl{a posteriori} mean and variance are obtained.
\cite{8425578} proposes an algorithm for downlink channel estimation in massive MIMO systems via turbo orthogonal approximate message passing. 
Combining the variational expectation maximization and generalized approximate message passing, \cite{8678465} proposes a super-resolution channel estimation algorithm for massive MIMO.
In \cite{Xiaofengliu}, a hybrid message passing algorithm is  proposed for massive MIMO-OFDM channel estimation based on Bethe free energy minimization.

Pioneered by Rao\cite{rao}, and later formally developed by Cencov\cite{cencov} and Amari\cite{amari}, information geometry has found a wide range of applications.
For Bayesian inference, Amari et al. \cite{srbpig} reveal the intrinsic geometrical structure of the space defined by the parameters of the \textsl{a posteriori}  probability density function (PDF) by regarding the parametric space as a differentiable manifold with a Riemannian structure.
With the information geometry theory,
the geometric insight of some conventional Bayesian inference approaches, e.g., belief propagation (BP) \cite{BP}, are shown, and some optimization methods, e.g., the concave-convex procedure (CCCP) \cite{CCCP}, are also applied to calculate the marginals of the \textsl{a posteriori} distribution.
In addition to the distinct intuition provided by the geometric perspective, information geometry also provides a unified framework where different sets of PDFs are considered to be endowed with the structure of differential geometry, which allows to construct a distance between two parametrized distributions. And it is shown that this distance is invariant to non-singular transformation of the parameters \cite{amari}.
Since the distance is based on the Fisher information matrix, the results derived from information geometry are tightly linked with fundamental results in estimation theory, such as the celebrated Cram\'{e}r-Rao lower bound.
Due to these advantages, information geometry has recently been applied to many other problems such as 
verification of dynamic models in power systems \cite{7895195} and direction of arrival estimation \cite{8686110}.

Recently, we have introduced the information geometry approach (IGA) to the massive MIMO-OFDM channel estimation \cite{IGA}.
We first provide the space-frequency (SF) beam based channel model for massive MIMO-OFDM system.
By allowing the fine factors to be greater that $1$, the SF beam based channel model can accurately characterize the channels in massive MIMO-OFDM systems.
The channel estimation is then formulated as obtaining the \textsl{a posteriori} information of the beam domain channel.
By introducing the information geometry theory, we solve this problem through calculating the approximations for the marginals of the \textsl{a posteriori} distribution.
Specifically, we turn the calculation of the approximations of the marginals into an iterative projection process by treating the set of Gaussian distributions with different constraints as different types of manifolds. 
Through the fixed point analysis, we improve the stability of IGA by introducing the damped updating and show that IGA can obtain accurate \textsl{a posteriori} mean at its fixed point.

In this paper, we first revisit the proposed IGA.
Based on the constant magnitude property of the entries of the measurement matrix in the massive MIMO-OFDM channel estimation, we reveal that at each iteration of IGA, the second-order natural parameters of the distributions on all the auxiliary manifolds are equivalent to each other, and at the fixed point of IGA, the first-order natural parameters of the distributions on all the auxiliary manifolds are asymptotically equivalent to each other.
These two results motivate us to replace the original natural parameters with a common natural parameter.
On this basis, we simplify the iteration of IGA and propose an efficient IGA (EIGA) for massive MIMO-OFDM channel estimation.
With the fast Fourier transform (FFT), we provide a low complexity implementation of EIGA.
We then analyze the convergence of the proposed EIGA.
We show that given a damping factor in a specific range, EIGA is guaranteed to converge.
We determine the range of the damping factor that guarantees the convergence of EIGA through the properties of the measurement matrices.
At last, we show that at the fixed point, the \textsl{a posteriori} mean obtained by EIGA is asymptotically optimal.

The rest of this paper is organized as follows. The system configuration and channel model are presented in Section \uppercase\expandafter{\romannumeral2}. 
We revisit IGA  and reveal two new results in Section \uppercase\expandafter{\romannumeral3}. 
EIGA for massive MIMO-OFDM channel estimation  is proposed in Section \uppercase\expandafter{\romannumeral4}.
Convergence and fixed point analysis are given in Section \uppercase\expandafter{\romannumeral5}
Simulation results are provided in Section \uppercase\expandafter{\romannumeral6}. 
The conclusion is drawn in Section \uppercase\expandafter{\romannumeral7}.

Notations: We adopt the following notations in this paper. 
Upper (lower) case boldface letters denote matrices
(column vectors). 
We use $\lceil x \rceil$ to denote the largest integer not larger than $x$. 
The superscripts $\left( \cdot \right)^*$, $\left( \cdot \right)^T$ and $\left( \cdot \right)^H$ denote the conjugate, transpose and conjugate-transpose operator, respectively. $\Diag{\bx}$ denotes the diagonal matrix with $\bx$ along its main diagonal and $\diag{\bX}$ denotes a vector consisting of the diagonal components of $\bX$.  
We use $\left[ \bA \right]_{:,i}$ to denote the $i$-th row of the matrix $\bA$, where the component indices start with $1$. 
$\odot$ and $\otimes$ denote the Hadamard product and  Kronecker product, respectively. 
Define $\setnnega{N} \triangleq \braces{0,1,\ldots,N}$ and $\setposi{N} \triangleq \braces{1,2,\ldots,N}$.
$\ba < b$ means that each component in vector $\ba$ is smaller than the scalar $b$.
$\ba < \bc$ means that each
component in vector $\ba$ is smaller than the component in the corresponding position in vector $\bc$.
$\norm{\bx}_0$ and $\norm{\bx}$ denote the $\ell_0$-norm and $\ell_2$-norm of $\bx$, respectively.
$p_G\left( \bh;\bmu,\bSigma \right)$ denotes the PDF of a complex Gaussian distribution $\mathcal{CN}\left( \bmu,\bSigma \right)$ for vector $\bh$ of complex random variables.
$\ba < b$ means that all the components of vector $\ba$ are smaller than scalar $b$.
$\ba < \bc$ means that each
component of vector $\ba$ is smaller than the corresponding component of vector $\bc$.

\section{System Model and Problem Statement}
In this section, we first present the configuration of the massive MIMO-OFDM system and the space-frequency beam based statistical channel model.
Then, we formulate the channel estimation as a standard Bayesian inference problem.

\subsection{System Configuration and Channel Model}
We consider a typical massive MIMO-OFDM system working in time division duplexing (TDD) mode with one base station (BS) serving $K$ single-antenna users within a cell, where the BS comprises a uniform planar array (UPA) of $N_r = N_{r,v}\times N_{r,h}$ antennas, and $N_{r,v}$ and $N_{r,h}$ are the numbers of the antennas at each vertical column and horizontal row, respectively. 
Due to channel reciprocity, channel state information can be obtained from uplink (UL) training, and then used for UL signal detection and downlink (DL) precoding.
Hence, our focus is on UL channel estimation. Standard OFDM modulation with $N_c$ subcarriers is applied, where the cyclic prefix (CP) is $N_g$. 
$N_p$ training  subcarriers are employed, and the set of them are denoted as $\mathcal{N}_p = \braces{N_1, N_1+1, \cdots, N_2}$, where $N_1$ and $N_2$ are the start
and end indices of the training subcarriers, respectively.
Assume that the channel is quasi-static, then, during each OFDM symbol, the SF domain received signal $\bY\in\bbC^{N_r\times N_p}$ for training at the BS  can be expressed as \cite{channelaqyou,IGA,Xiaofengliu} 
\begin{equation}\label{equ:maMIMO rece signal}
	\bY = \sum_{k=1}^{K}\bG_k\bP_k + \bZ,
\end{equation}
where $\bG_k \in \bbC^{ N_r \times N_p }$ is the SF domain channel of user $k$, $\bP_k  = \Diag{\bp_k} \in\bbC^{N_p\times N_p}$ is the pilot signal of user $k$, $\bp_k$ is the pilot sequence of user $k$,
and $\bZ$ is the noise matrix whose components are independent and identically distributed complex Gaussian random variables with zero mean and variance $\sigma_z^2$.

Suppose that the antenna spacings of each row and each column of the UPA are one-half wavelength, respectively.
Define the directional cosines as $u \triangleq \sin\theta$ and $v \triangleq \cos\theta\sin\phi$, where $\theta, \phi\in\left[-\pi/2,\pi/2\right]$ are the vertical and the horizontal angles of arrival (AoA) at the BS, respectively. 
Then, the space steering vectors can be expressed as \cite{8298537,2Dlu} 
\begin{subequations}
	\begin{equation*}
		\bv\left( u,v \right) = \bv_v\left( u \right) \otimes \bv_h\left( v \right) \in \bbC^{N_r\times 1},
	\end{equation*}
	\begin{equation*}
	 \bv_v\left(u\right) = \left[ p\left(1\right), \ p\left(2\right),  \ \cdots, \ p\left(N_{r,v}\right) \right]^T \in \bbC^{N_{r,v} \times 1}, 
	\end{equation*}
	\begin{equation*}
	 \bv_h\left(v\right) = \left[ q\left(1\right), \ q\left(2\right), \ \cdots, \ q\left(N_{r,h}\right) \right]^T \in \bbC^{N_{r,h} \times 1},
	\end{equation*}
\end{subequations}
where $p\left(n\right) = \exp\left\lbrace -\barjmath \pi\left(n-1\right)u\right\rbrace$ and $q\left(n\right) = \exp\left\lbrace  -\barjmath \pi\left( n - 1 \right)v\right\rbrace$. Denote the delay of the multipaths of the channel as $\tau$ \cite{IGA,channelaqyou,2Dlu}. Then, the frequency steering vector is given by \cite{IGA}, 
\begin{equation*}\label{equ:frequnce steering vector}
 \bu\left( \tau \right) = \left[ r\left(N_1\right),\  \cdots, \ r\left( N_2 \right) \right]^T \in \bbC^{N_p \times 1},
\end{equation*}
where $r\left(n\right) = \exp\left\lbrace -\barjmath 2\pi \Delta_f n \tau\right\rbrace$ and $\Delta_f$ is the subcarrier interval. 
Define the matrices containing the sampled space steering vectors and the sampled frequency steering vectors as
\begin{equation}\label{equ:def of V}
	\bV \triangleq \bV_v \otimes \bV_h \in \bbC^{N_r \times N_vN_h},
\end{equation} 
\begin{equation}\label{equ:def of F}
	\bF \triangleq \left[ \bu\left( \tau_{1} \right),  \ \bu\left( \tau_{2} \right),  \ \cdots, \ \bu\left( \tau_{N_\tau} \right)\right] \in \bbC^{N_p \times N_\tau},
\end{equation}
where 
\begin{equation*}
	\bV_v \triangleq \left[ \bv_v\left( u_{1}\right), \ \bv_v\left( u_{2}\right), \ \cdots, \ \bv_v\left( u_{N_v} \right)  \right] \in \bbC^{N_{r,v} \times N_v},
\end{equation*}
\begin{equation*}
	\bV_h \triangleq \left[ \bv_h\left( v_{1}\right), \ \bv_h\left( v_{2}\right), \ \cdots, \ \bv_h\left( v_{N_h} \right)  \right] \in \bbC^{N_{r,h} \times N_{h}}.
\end{equation*}
$u_i$, $v_j$ and $\tau_\ell$ above are the sampled directional cosines and delays, 
which are defined as follows: 
\begin{subequations}
	\begin{equation*}
		u_i \triangleq \frac{2\left(i-1\right)-N_v}{N_v}, i\in \mathcal{Z}_{N_v}^+,
	\end{equation*}
\begin{equation*}
	v_j = \frac{2\left(j-1\right)-N_h}{N_h}, j\in \mathcal{Z}_{N_h}^+,
\end{equation*}
\begin{equation*}
	\tau_\ell  = \frac{\left( \ell-1 \right)N_f}{N_{\tau }N_p\Delta_f }, \ell \in \mathcal{Z}_{N_\tau}^+,
\end{equation*}
\end{subequations}
$N_v \triangleq F_vN_{r,v}$, $N_h \triangleq F_hN_{r,h}$,
$N_\tau \triangleq F_\tau N_f$ and $N_f = \left\lceil  {N_p\dnnot{N}{g}}/{N_\mathrm{c}} \right\rceil$.
$F_v$, $F_h$ and $F_\tau$ above are called the fine (oversampling) factors (FFs).
$N_v$, $N_h$ and $N_\tau$ are the numbers of sampled directional cosines and sampled delays, respectively.
Larger FFs lead to more sampled directional cosines and delays, which is necessary for accurately modeling the SF channel in massive MIMO-OFDM systems \cite{IGA}.
When $N_v$, $N_h$ and $N_\tau$ are sufficiently large, the SF domain channel $\bG_k$ can be expressed as \cite{IGA,2Dlu,channelaqyou}
\begin{equation}\label{equ:SF beam channel}
	\bG_k = \bV\bH_k\bF^T, k\in\setposi{K},
\end{equation}
where $\bH_k \in \bbC^{F_vF_hN_r \times F_\tau N_f}$ is the SF beam domain channel matrix of user $k$, and the components in $\bH_k$ follow the independent complex Gaussian distributions with zero mean and possibly different variances. 
We denote the power matrix of beam domain channel as 
\begin{equation}
	 \bOmega_k = \Exp\braces{\bH_k\odot \bH_k^*}, k\in\setposi{K}.
\end{equation}
Due to the channel sparsity, most of the components in $\bOmega_k$ are (close to) zero and the non-zero components usually gather in clusters, where each cluster corresponds to a physical scatterer.  
Meanwhile, compared to the SF domain channel matrix $\bG_{k}$, the power matrix $\bOmega_k$ maintains unchanged within a much longer period \cite{7042346,channelaqyou}.
The channel power matrices $\braces{\bOmega_k}_{k=1}^K$ can be obtained by methods such as \cite{2Dlu,9373011}.
In the rest of this paper, we assume that $\braces{\bOmega_k}_{k=1}^K$ are known at the BS. 

\subsection{Problem Statement}\label{sec:Probelm Statement}
The goal of channel estimation is to obtain the \textsl{a posteriori}  information of  the SF domain channel $\bG_k, k\in\setposi{K}$ when the received signal $\bY$ is given. 
Since the \textsl{a posteriori}  information of $\bG_k$ can be
calculated from that of the SF beam domain channel matrix $\bH_k$ through (\ref{equ:SF beam channel}), we focus on the estimation of $\bH_k, k\in\setposi{K}$.
Substituting (\ref{equ:SF beam channel}) into (\ref{equ:maMIMO rece signal}), we can obtain
\begin{equation}\label{equ:maMIMO rece signal 2}
	\bY = \bV\bH\bM + \bZ,
\end{equation}
where $\bV$ and $\bZ$ are the same as above,
$\bH = \left[ \bH_1, \ \bH_2, \  \cdots, \ \bH_K \right] \in \bbC^{ F_aN_r\times KF_\tau N_f }$, $F_a\triangleq F_v\times F_h$ and $\bM = \left[ \bP_1\bF,\  \bP_2\bF,\ \cdots,\ \bP_K\bF \right]^{T} \in \bbC^{KF_\tau N_f \times N_p}$.
After vectorizing (\ref{equ:maMIMO rece signal 2}), and removing the components of $\mathrm{vec}\braces{\bH}$ with zero variance and the corresponding columns in $\bM^T\otimes \bV$, we can obtain 
\begin{equation}\label{equ:maMIMO rece signal 4}
	\by = \bA\bh + \bz,
\end{equation}
where $\bA \in\bbC^{N\times M}$ is a deterministic matrix extracted from $\bM^T\otimes \bV$, 
$N = N_rN_p$,
$M$ is the number of components in $\bH$ with non-zero variance,
$\by$ and $\bz$ are the vectorizations of $\bY$ and $\bZ$, respectively,
$\bh \in \bbC^{M}$ is a Gaussian random vector extracted from $\mathrm{vec}\braces{\bH}$.
In \eqref{equ:maMIMO rece signal 4}$, \bh \sim\mathcal{CN}\left(\mathbf{0},\bD\right)$ with diagonal and positive definite $\bD$ and $\bz \sim \mathcal{CN}\left(\mathbf{0},\sigma_z^2\bI\right)$.
Assume that $\bh$ and $\bz$ are independent with each other. 
Then, given the observation $\by$, the \textsl{a posteriori} distribution of $\bh$ is Gaussian, and we have 
\begin{equation}
	\begin{split}
	   	p\left(\bh|\by\right) &= p_G\left(\bh;\tilde{\bmu},\tilde{\bSigma}\right) \propto p\left(\bh\right)p\left(\by|\bh\right)	\\
	   	& \propto\exp\braces{-\bh^H\bD^{-1}\bh -  \frac{\norm{\by - \bA\bh}^2}{\sigma_z^2}}.
	\end{split}
\end{equation}
The \textsl{a posteriori} mean and covariance matrix of $\bh$ are given by \cite{kay1993fundamentals}
\begin{subequations}\label{equ:post information}
	\begin{equation}\label{equ:post mean}
		\tilde{\bmu} = \bD\left( \bA^H\bA\bD + \sigma_z^2\bI  \right)^{-1}\bA^H\by,
	\end{equation}
	\begin{equation}\label{equ:post covariance}
		\tilde{\bSigma} = \left( \bD^{-1} + \frac{1}{\sigma_z^2}\bA^H\bA \right)^{-1},
	\end{equation}
\end{subequations}
respectively.
The \textsl{a posteriori} mean $\tilde{\bmu}$ is also the MMSE estimate of $\bh$ \cite{kay1993fundamentals}. 
In this work, our goal is to calculate the approximate marginals of the \textsl{a posteriori} PDF $p\left(\bh|\by\right)$, where the marginals are denoted as $p\left(h_i|\by\right), i\in\setposi{M}$.
Then, the \textsl{a posteriori} mean and variance of $\bh$ can be obtained.

\section{Revisiting IGA}
In this section, we introduce the IGA for massive MIMO-OFDM channel estimation, for which more details can be found in \cite{IGA,srbpig}. 
Then, two new properties of the IGA are obtained, which motivate us to simplify the IGA.

\subsection{IGA}
Given \eqref{equ:maMIMO rece signal 4}, $p\left(\bh|\by\right)$ can be further expressed as \cite{IGA}
\begin{equation}\label{equ:a posteriori pdf 1}
		p\left(\bh|\by\right) \propto \expb{-\bh^H\bD^{-1}\bh -\sum_{n=1}^{N}\frac{\left|y_n - \bgamma_n^H\bh\right|^2}{\sigma_z^2}},
\end{equation}
where $y_n$ is the $n$-th component of $\by$, and
\begin{equation}\label{equ:gamma_n}
	\bgamma_n = \left[ \bA^H \right]_{:,n} = \left[ {a}_{n1}^* \ \cdots \  {a}_{nM}^* \right]^T \in \bbC^{M\times 1}.
\end{equation}
We define a vector function as $\funf[\ba,\bb] \triangleq \left[ \ba^T, \ \bb^T \right]^T\in \bbC^{2M\times 1}$,
where $\ba, \bb \in \bbC^{M\times 1}$. 
Let $\bd \triangleq \funf[\mathbf{0},\diag{-\bD^{-1}}]$ and $\bt \triangleq \funf[\bh,\left(\bh\odot\bh^*\right)]$.
Then, \eqref{equ:a posteriori pdf 1} can be rewritten as 
\begin{equation}\label{equ:a posteriori pdf 2}
			p\left(\bh|\by\right) = \expb{\bd \circ \bt + \sum_{n=1}^{N}c_n\left(\bh\right) - \psi_q},
\end{equation}
where $\circ$ is an operator of two vectors with the same dimension, and $	\ba \circ \bb \triangleq \frac{1}{2} \left(\bb^H\ba + \ba^H\bb\right)$,
$\psi_q$ is the normalization factor and  
\begin{equation}\label{equ:c_n}
	c_n\left(\bh\right) \!=\! \frac{1}{\sigma_z^2}\left(  -\bh^H{\bgamma_n\bgamma_n^H}\bh + {y_n}\bh^H\bgamma_n + y_n^*\bgamma_n^H\bh\right).
\end{equation}
In \eqref{equ:a posteriori pdf 2}, $\bt$ only contains the  statistics of single random variables, i.e., $h_i$ and $\left|h_i\right|^2, i\in \mathcal{Z}_M^+$, and all the interactions (cross terms), $h_i{h}_j^*, i,j \in \mathcal{Z}_M^+$, are included in the terms $c_n\left(\bh\right), n \in \mathcal{Z}_N^+$.  IGA aims to approximate $\sum_{n=1}^{N}c_n\left(\bh\right)$  as $\bvartheta_0 \circ \bt$,
where $\bvartheta_0 = \funf[\btheta_0,\bbnu_0]$, $\btheta_0 \in \bbC^{M \times 1}$ and $\bbnu_0 \in \bbR^{M\times 1}$.
Then, we can obtain
\begin{equation}
	p\left(\bh|\by\right) \approx p_0\left(\bh;\bvartheta_0\right) = \expb{\left(\bd + \bvartheta_0\right) \circ \bt - \psi_0 },
\end{equation}
where $\psi_0$ is the normalization factor.
The marginals of $p_0\left(\bh;\bvartheta_0\right)$ can be calculated easily since it contains no interactions. 
To obtain $\bvartheta_{0}$, IGA constructs three types of manifolds and computes the approximation for each $c_n\left(\bh\right)$ in an iterative manner, which is 
denoted as $\bxi_{n}\circ\bt$.
At last, $\bvartheta_0 = \sum_{n=1}^N\bxi_n$ is used as the parameter of $p_0\left(\bh;\bvartheta_0\right)$.
The three types of manifolds are the original manifold (OM), the objective manifold (OBM) and the auxiliary manifold (AM), respectively.
The OM is defined as the set of PDFs of $M$ dimensional complex Gaussian random vectors,
\begin{equation}
	\!\!\!\mathcal{M}_{or} = \left\lbrace p\left(\bh\right) \!=\! p_G\left(\bh;\bmu,\bSigma\right), \bmu\in\bbC^{M\times 1}, \bSigma\in \mathbb{H}_+^M \right\rbrace, 
\end{equation}
where  $\mathbb{H}_+^M$ is the set of $M$ dimensional positive definite matrices.  
The OBM is defined as
\begin{equation}\label{equ:definition of M0}
		\mathcal{M}_0 = \left\lbrace p_0\left(\bh;\bvartheta_0\right) = \expb{\left(\bd+\bvartheta_0\right)\circ\bt - \psi_0\left(\bvartheta_0\right)}  \right\rbrace, 
\end{equation} 
where $\bvartheta_0 = \funf[\btheta_0,\bbnu_0]$ with $\btheta_0 \in \bbC^{M\times 1}$ and $\bbnu_0 \in \bbR^{M\times 1}$, and the free energy (normalization factor) $\psi_0\left(\bvartheta_0\right)$ is given by \cite[Equation (40a)]{IGA}.
We refer to $\bvartheta_{0}$, $\btheta_{0}$ and $\bbnu_{0}$ as the natural parameter (NP), the first-order natural parameter (FONP) and the second-order natural parameter (SONP) of $p_0$.
Finally, $N$ AMs are defined, where the $n$-th AM is defined as
\begin{subequations}\label{equ:definition of Mn}
	\begin{equation}
		\mathcal{M}_n = \braces{p_n\left(\bh;\bvartheta_n\right)}, n\in\mathcal{Z}_N^+,
	\end{equation}
     \begin{equation}
     \!\!\!	p_n\left(\bh;\bvartheta_n\right) \!=\! \expb{\left(\bd + \bvartheta_n\right)\circ \bt + c_n\left(\bh\right) - \psi_n\left(\bvartheta_n\right)},
     \end{equation}
\end{subequations}
where $\bvartheta_n = \funf[\btheta_n,\bbnu_n]$, $\btheta_n \in \bbC^{M\times 1}$ and $\bbnu_n \in \bbR^{M\times 1}$ are referred to as the NP, the FONP and the SONP of $p_n$,
and the free energy $\psi_n\left(\bvartheta_n\right)$ is given by \cite[Equation (40b)]{IGA}.
The distributions in the OBM and AMs are all $M$ dimensional complex Gaussian distributions. 
We have $p_n\left(\bh;\bvartheta_{n}\right) = p_G\left(\bh;\bmu_n,\bSigma_{n}\right), n\in\setnnega{N}$, where
\begin{subequations}\label{equ:mu_0 and Sigma_0}
	\begin{equation}\label{equ:mu_0}
		\bmu_0\left(\bvartheta_{0}\right) = \frac{1}{2}\bSigma_{0}\left(\bvartheta_{0}\right)\btheta_{0},
	\end{equation}
	\begin{equation}\label{equ:Sigma_0}
		\bSigma_{0}\left(\bvartheta_{0}\right) = \left( \bD^{-1}-\Diag{\bbnu_0} \right)^{-1},
	\end{equation}
\end{subequations}
\begin{subequations}\label{equ:mu_n and Sigma_n}
	\begin{equation}\label{equ:mu_n}
		\bmu_n\left(\bvartheta_{n}\right) = \bSigma_{n}\left(\bvartheta_{n}\right)\left( \frac{y_n}{\sigma_z^2}\bgamma_n + \frac{1}{2}\btheta_{n} \right),
	\end{equation}
	\begin{equation}\label{equ:Sigma_n}
		\bSigma_{n}\left(\bvartheta_{n}\right) = \bLambda_n - \frac{1}{\beta_n}\bLambda_n\bgamma_n\bgamma_n^H\bLambda_n,
	\end{equation}
    \begin{equation}\label{equ:Lambda_n}
    	\bLambda_n = \left( \bD^{-1} - \Diag{\bbnu_n}\right)^{-1},
    \end{equation}
    \begin{equation}\label{equ:beta_n}
    	\beta_n = \sigma_z^2 + \bgamma_n^H\bLambda_n\bgamma_n, n \in \mathcal{Z}_N^{+}.
    \end{equation}
\end{subequations}
Write $\bmu_n$ and $\bSigma_{n}$ as functions w.r.t. $\bvartheta_n, n\in\setnnega{N}$, since we will frequently use the relationship between the parameters and means and covariances in the following.

$p_n\left(\bh;\bvartheta_n\right)$ in \eqref{equ:definition of Mn} only contains single interaction item $c_n\left(\bh\right)$, and all others, i.e., $\sum_{n'\neq n}c_{n'}\left(\bh\right)$ are replaced as $\bvartheta_n\circ \bt$. 
Suppose that the NP $\bvartheta_n$ is given, the approximation of $c_n\left(\bh\right)$ is then obtained through $m$-projecting $p_n\left(\bh;\bvartheta_n\right)$ onto the OBM. 
Specifically, $m$-projecting $p_n\left(\bh;\bvartheta_n\right)$ onto the OBM is equivalent to finding the point on the OBM minimizing the following K-L divergence,
\begin{equation}
	\bvartheta_{0n} = \argmin{\bvartheta_0} \Dkl{p_n\left(\bh;\bvartheta_n\right)}{p_0\left(\bh;\bvartheta_0\right)},
\end{equation}
where
\begin{equation}
	\Dkl{p_n\left(\bh;\bvartheta_n\right)}{p_0\left(\bh;\bvartheta_0\right)} = \Exp_{p_n}\braces{\ln\frac{p_n\left(\bh;\bvartheta_n\right)}{p_0\left(\bh;\bvartheta_0\right)}}.
\end{equation}
$\bvartheta_{0n}\! =\! \funf[\btheta_{0n},\bbnu_{0n}]$, $n\in \setposi{N}$, is then given by
\begin{subequations}\label{equ:vartheta_0n}
	\begin{align}\label{equ:theta_0n}
			\btheta_{0n} = &\left[ \bI - \frac{1}{\beta_n}\bLambda_n \bI \odot \left( \bgamma_n\bgamma_n^H \right) \right]^{-1}\nonumber\\
			&\times \left( \frac{2y_n - \bgamma_n^H\bLambda_n\btheta_{n}}{\beta_n}\bgamma_n + \btheta_n \right), 
	\end{align}
	\begin{equation}\label{equ:nu_0n}
		\bbnu_{0n} = \diag{\bD^{-1}-\left[ \bLambda_n - \frac{1 }{\beta_n}\bLambda_n^2 \bI \odot \left( \bgamma_n\bgamma_n^H\right)  \right]^{-1}}, 
	\end{equation}
\end{subequations}
where $\bLambda_n$ and $\beta_n$ are given by \eqref{equ:Lambda_n} and \eqref{equ:beta_n}, respectively.
We now discuss an important property of the $m$-projection.
Given $p_n\left(\bh;\bvartheta_{n}\right)$ and its $m$-projection on the OBM $p_0\left(\bh;\bvartheta_{0n}\right), n\in \setposi{N}$, 
the expectations of $\bt$ w.r.t. $p_n\left(\bh;\bvartheta_{n}\right)$ and $p_0\left(\bh;\bvartheta_{0n}\right)$ are the same \cite{IGA,srbpig}, i.e., 
\begin{equation}\label{equ:mp invariant}
	 \int \bt p_n\left(\bh;\bvartheta_{n}\right)\intdx{\bh}= \int \bt p_0\left(\bh;\bvartheta_{0n}\right)\intdx{\bh}, n\in \setposi{N}.
\end{equation}
This is equivalent to 
\begin{equation}\label{equ:m-p invariant 2}
	\bbeta_n\left(\bvartheta_{n}\right) = \bbeta_{0n}, n\in\setposi{N},
\end{equation}
where 
\begin{subequations}
		\begin{equation*}
		\bbeta_n\left(\bvartheta_{n}\right) \!\triangleq\! \left[ \bmu_n^T\left(\bvartheta_n\right), \ \textrm{diag}^T\braces{\bSigma_{n}\left(\bvartheta_n\right)} \right]^T \!\!\in\! \bbC^{2M \times 1}.
	\end{equation*}
	\begin{equation*}
		\bbeta_{0n} \triangleq \left[ \bmu_{0}^T\left(\bvartheta_{0n}\right), \ \textrm{diag}^T\braces{\bSigma_{0}\left(\bvartheta_{0n}\right)} \right]^T \in \bbC^{2M \times 1}.
	\end{equation*}
\end{subequations}
We will use this property in the analysis of the fixed point of EIGA.

Now, let us express the $m$-projection $p_0\left(\bh;\bvartheta_{0n}\right)$ in the following way:
\begin{equation}\label{equ:explain of m-projection}
	\begin{split}
		p_0\left(\bh;\bvartheta_{0n}\right) &= \expb{\left(\bd + \bvartheta_{0n}\right)\circ\bt - \psi_0}\\
		&=\expb{\left( \bd + \bvartheta_n + \bxi_n \right)\circ\bt -\psi_0 }.
	\end{split}
\end{equation}
The NP $\bvartheta_{0n}$ of $p_0\left(\bh;\bvartheta_{0n}\right)$ is regarded as the sum of the NP $\bvartheta_{n}$ of $p_n\left(\bh;\bvartheta_{n}\right)$ and an extra item that is denoted as $\bxi_n$. Comparing the last line of \eqref{equ:explain of m-projection} and $p_n$ in \eqref{equ:definition of Mn}, we can find that $c_n\left(\bh\right)$ in $p_n$ is replaced by $\bxi_n\circ \bt$ in $p_0\left(\bh;\bvartheta_{0n}\right)$. Thus, we regard $\bxi_n\circ \bt$ as an approximate of $c_n\left(\bh\right)$ and calculate $\bxi_n$ as
\begin{equation}\label{equ:xi_n}
	\bxi_n = \bvartheta_{0n} - \bvartheta_{n}, n\in\mathcal{Z}_N^+.
\end{equation}
We then calculate $\bvartheta_0$ as $\bvartheta_0 = \sum_{n=1}^N\bxi_n$ and consider $p_0\left(\bh;\bvartheta_0\right)$ as an approximation of $p\left(\bh|\by\right)$. 

Now, we summarize the complete process of IGA.
Note that IGA proceeds in an iterative manner since the NPs of $\braces{p_n\left(\bh;\bvartheta_{n}\right)}_{n=1}^N$
 are unknown at the beginning. 
Specifically, we first initialize $\bvartheta_{n}, n\in \setnnega{N}$. 
We then calculate $\bvartheta_{0n}$ as \eqref{equ:vartheta_0n} and $\bxi_{n}$ as \eqref{equ:xi_n}. 
The NP of $p_n, n\in \setposi{N}$, is then updated as
$\bvartheta_{n} = \sum_{n' \neq n}\bxi_{n'}$ since  $\bvartheta_{n} \circ \bt$ replaces $\sum_{n' \neq n}c_{n'}\left(\bh\right)$ in $p_n$ and each interaction item $c_n\left( \bh \right)$ is approximated as $\bxi_{n} \circ \bt$ after the $m$-projection. 
The NP of $p_0$ is updated as $\bvartheta_{0} = \sum_{n=1}^{N}\bxi_{n}$.
Then, repeat the $m$-projections, calculate the approximation items and the updates until convergence.
In practice, the NPs of $\braces{p_n\left(\bh;\bvartheta_{n}\right)}_{n=0}^N$ are typically updated with a damping factor, i.e.,
\begin{subequations}\label{equ:update of NP in IGA}
	\begin{equation}\label{equ:update of vartheta_n in IGA}
		\bvartheta_{n}\left(t+1\right) = d\sum\nolimits_{n'\neq n}\bxi_{n'}\left(t\right) + \left(1-d\right)\bvartheta_{n}\left(t\right), 
	\end{equation}
	\begin{equation}\label{equ:update of vartheta_0 in IGA}
		\bvartheta_0\left(t+1\right) = d\sum\nolimits_{n=1}^N\bxi_{n}\left(t\right) + \left(1-d\right)\bvartheta_{0}\left(t\right),
	\end{equation}
\end{subequations}
where $n \in \setposi{N}$ in \eqref{equ:update of vartheta_n in IGA}.
The damped updating of the NPs could improve the convergence of IGA.

Next,
we introduce two conditions of the fixed point of IGA. 
When converged, denote the fixed points of the parameters in IGA as $\bxi_{n}^{\star}, n\in \setposi{N}$,
$\bvartheta_{n}^{\star} = \funf[\btheta_n^\star,\bbnu_n^\star], n\in \setposi{N}$, $\bvartheta_{0n}^{\star}, n\in\setposi{N}$, and $\bvartheta_{0}^\star = \funf[\btheta_0^\star,\bbnu_0^\star]$. 
By solving the fixed point equation of IGA, we can obtain \cite{IGA} 
\begin{equation*}
	\bvartheta_{0}^{\star} = \bvartheta_{0n}^{\star} = \frac{1}{N-1}\sum_{n=1}^{N}\bvartheta_{n}^{\star}.
\end{equation*}
Define 
\begin{equation}
	\bbeta_{0}\left(\bvartheta_{0}\right) \triangleq \left[ \bmu_0^T\left(\bvartheta_{0}\right), \ \mathrm{diag}^T\{ \bSigma_{0}\left(\bvartheta_{0}\right) \} \right]^T \in \bbC^{2M \times 1},
\end{equation} 
$\bbeta_{0}^{\star} \triangleq \bbeta_{0}\left(\bvartheta_{0}^{\star}\right)$, $\bbeta_{0n}^{\star} \triangleq \bbeta_{0}\left(\bvartheta_{0n}^{\star}\right), n\in\setposi{N}$, and $\bbeta_n^{\star} \triangleq \bbeta_{n}\left(\bvartheta_{n}^{\star}\right), n\in\setposi{N}$. 
We can obtain 
\begin{equation}\label{equ:m-condition2}
	\bbeta_{0}^{\star} \equaa \bbeta_{0n}^{\star} \equab \bbeta_n^{\star}, n\in\setposi{N},
\end{equation}
where $\left(\textrm{a}\right)$ comes from $\bvartheta_{0}^{\star} = \bvartheta_{0n}^{\star}, n\in\setposi{N}$, and $\left(\textrm{b}\right)$ comes from that $p_0\left(\bh;\bvartheta_{0n}^{\star}\right)$ is the $m$-projection of $p_n\left(\bh;\bvartheta_{n}^{\star}\right)$, $n\in\setposi{N}$, on the OBM and thus \eqref{equ:mp invariant} holds.
In summary, the two conditions are
\begin{equation}\label{equ:m and e conditions}
	\begin{cases}
		m\textrm{-condition:}\ \bbeta_{0}^{\star} = \bbeta_n^{\star}, n\in\setposi{N},\\
	    \ e \textrm{-condition:}\     	\bvartheta_{0}^{\star} = \frac{1}{N-1}\sum_{n=1}^{N}\bvartheta_{n}^{\star}.
	\end{cases}
\end{equation}

\subsection{New Results}

In practice, the pilot sequences with constant magnitude property are preferred for massive MIMO-OFDM systems \cite{channelaqyou,6940305,2Dlu}.
In this case, the measurement matrix $\bA$ in the received signal model \eqref{equ:maMIMO rece signal 4} have the constant magnitude entry property, i.e., $\left|a_{i,j}\right| = \left| a_{m,n} \right|, \forall i,j,m,n$, where $a_{i,j}$ is the $\left(i,j\right)$-th element of $\bA$.
Under this condition, the iteration of IGA shows two new properties. 
Unless specified, we assume that the components of the pilot sequences and thus the measurement matrix entries have unit magnitude in the rest of this paper.
\begin{theorem}\label{the:same nu}
	If the matrix $\bA$ in \eqref{equ:maMIMO rece signal 4} has constant magnitude entry property, then at each iteration of IGA, the SONPs of both $p_n, n\in\setposi{N}$, and its $m$-projection on the OBM are independent of $n$, i.e.,
	\begin{subequations}
		\begin{equation}
			\bbnu_{n}\left(t\right) = \bbnu_{n'}\left(t\right),
		\end{equation}
		\begin{equation}
			\bbnu_{0n}\left(t\right) = \bbnu_{0n'}\left(t\right), n, n' \in\setposi{N},
		\end{equation}
	\end{subequations}
	when the initializations of the SONPs of $\braces{p_n}_{n=1}^N$ are the same.
	Furthermore, if the initializations of the SONPs of $p_0$ and $p_n, n\in\setposi{N}$, satisfy $\bbnu_0\left(0\right), \bbnu_n\left(0\right) \le 0$, then their fixed points  satisfy $\bbnu_0^{\star}, \bbnu_n^{\star} < 0, n\in\setposi{N}$.
\end{theorem}
\begin{proof}
	See in Appendix \ref{proof:same nu}.
\end{proof}
Define the arithmetic mean of the SONPs of $\braces{p_n}_{n=1}^N$ as ${\bbnu} \triangleq \frac{1}{N}\sum_{n=1}^N\bbnu_{n}$.
From the above theorem, $\bbnu_{n}, n\in\setposi{N}$, in IGA can be replaced by ${\bbnu}$ in each iteration, and the two iteration modes are equivalent to each other when $\bA$ has constant magnitude entry property. 
Motivated by this observation, we find that a similar property is satisfied between the FONPs of $\braces{p_n}_{n=1}^N$ in IGA.

For an $M\times M$ positive definite diagonal matrix $\bD$,
define 
\begin{equation*}
	\norm{\btheta}_{\bD} \triangleq  \sqrt{\btheta^H\bD\btheta},
\end{equation*}
where $\btheta\in\bbC^{M\times 1}$. 
Since $\bD$ is positive definite diagonal, we have $\norm{\btheta}_{\bD} = \norm{\bD^{\frac{1}{2}}\btheta}$.
And $\norm{\cdot}_{\bD}$ is a weighted norm on $\bbC^{M\times 1}$.
Then, we have the following result.
\begin{theorem}\label{the:same theta_n}
	In IGA, the fixed points of all the FONPs of $\braces{p_n}_{n=1}^N$ are asymptotically equal to $\frac{N-1}{N}$ times the fixed point of the FONP of $p_0$, i.e.,
	\begin{equation}\label{equ:FONPs in IGA}
		\lim\limits_{N\to \infty}
		\frac{1}{NM}\sum_{n=1}^{N}\Exp\braces{\lVert\btheta_{n}^{\star} - \frac{N-1}{N}\btheta_{0}^{\star} \rVert^2_{\bD}} = 0,
	\end{equation}
	where $M/N = \alpha > 0$ is a constant.
\end{theorem}
\begin{proof}
	See in Appendix \ref{{proof:same theta_n}}.
\end{proof}
Theorem \ref{the:same theta_n} illustrates that as $N$ and $M$ tend to infinity, the average error between each component in the fixed point of the FONP of $p_n, n\in \setposi{N}$, and each component in the fixed point of the FONP of $p_0$ is asymptotically equal to zero.
In massive MIMO-OFDM channel estimation, $N$ is usually quite large. 
When the number of users is large, $M$ can be also large even though the channel sparsity exists. In this case, the fixed point of the FONP of 
$p_n, n\in\setposi{N}$, tends to be equal to each other,  and the value can be obtained directly from the $e$-condition in \eqref{equ:m and e conditions}.


\section{Efficient IGA}
In this section, we simplify the iteration of IGA and propose EIGA by replacing the original NPs of $\braces{p_n}_{n=1}^N$ with a common NP.
Then, the efficient implementation with FFT of EIGA is provided.

\subsection{EIGA}


Define the arithmetic mean of the NPs of $\braces{p_n\left(\bh;\bvartheta_{n}\right)}_{n=1}^N$ as ${\bvartheta} \triangleq \left({1}/{N}\right)\sum\nolimits_{n=1}^N\bvartheta_{n}$.
We use
${\bvartheta}$ instead of $\bvartheta_{n}, n\in\setposi{N}$, to simplify the iteration of IGA.
This replacement allows more efficient implementation.

The input is the same as that of IGA.
At the initialization, we set the counter $t=0$ and choose the damping $d$, where $0<d\le 1$.
We shall see more explicit ranges of $d$ in the next section.
We initialize the NP for $p_0$ as $\bvartheta_{0}\left(0\right)$ and initialize the NP for $\braces{p_n}_{n=1}^N$ as $\bvartheta\left(0\right)$ while ensuring that $\bbnu_{0}\left(0\right), \bbnu\left(0\right) \le 0$.
We refer to $\bvartheta$ as the common NP of $\braces{p_n}_{n=1}^N$ (abbreviated as the common NP).
Given the common NP ${\bvartheta}\left(t\right) = \funf[{\btheta}\left(t\right),{\bbnu}\left(t\right)]$ at the $t$-th iteration, we $m$-project $p_n\left(\bh;{\bvartheta}\left(t\right)\right)$ onto the OBM and obtains the $m$-projection, denoted as $p_0\left(\bh;{\bvartheta}_{0n}\left(t\right)\right)$, where $n\in\setposi{N}$.
Substituting ${\bvartheta}\left(t\right) = \funf[{\btheta}\left(t\right),{\bbnu}\left(t\right)]$ into \eqref{equ:Lambda_n}, \eqref{equ:beta_n} and \eqref{equ:vartheta_0n}, i.e., replacing $\bvartheta_{n} = \funf[\btheta_n,\bbnu_n]$ with ${\bvartheta}\left(t\right) = \funf[{\btheta}\left(t\right),{\bbnu}\left(t\right)]$, and considering that $\bA$ is of constant magnitude entries, 
${\bvartheta}_{0n}\left(t\right) = \funf[{\btheta}_{0n}\left(t\right),\bbnu_{0n}\left(t\right)], n\in \setposi{N}$, is now given by
\begin{subequations}\label{equ:tilde vartheta_0n}
	 \begin{equation}\label{equ:tilde theta_0n}
		\begin{split}
		{\btheta}_{0n}\left(t\right) = &\left( \bI - \frac{1}{\beta\left(\bbnu\left(t\right)\right)}\bLambda\left(\bbnu\left(t\right)\right) \right)^{-1}\\
		&\times\left( \frac{2y_n - \bgamma_n^H\bLambda\left(\bbnu\left(t\right)\right){\btheta}\left(t\right)}{\beta\left(\bbnu\left(t\right)\right)}\bgamma_n + {\btheta}\left(t\right) \right),
		\end{split}
     \end{equation}
    \begin{equation}\label{equ:tilde nu}
	\bbnu_{0n}\left(t\right) = \diag{\bD^{-1}-\left( \bLambda\left(\bbnu\left(t\right)\right) - \frac{1}{\beta\left(\bbnu\left(t\right)\right)}\bLambda^2\left(t\right) \right)^{-1}},
    \end{equation}
	\begin{equation}\label{equ:Lambda}
	\bLambda\left(\bbnu\left(t\right)\right) = \left( \bD^{-1} - \Diag{{\bbnu}\left(t\right)}\right)^{-1},
    \end{equation}
    \begin{equation}\label{equ:beta}
	\beta\left(\bbnu\left(t\right)\right) = \tilde{\sigma}_z^2 + \tr{\bLambda\left(\bbnu\left(t\right)\right)}.
    \end{equation}
\end{subequations}
Note that in \eqref{equ:beta} $\sigma_z^2$ is replaced with $\tilde{\sigma}_z^2$.
We refer to $\tilde{\sigma}_z^2$ as the virtual noise variance.
This is a common technique has been used to improve the performance in iterative Bayesian inference methods\cite{8074806,6998861}, since they do not necessarily have the best performance when the exact $\sigma_z^2$ is used.
In the next section, we will give a closed-form expression of $\tilde{\sigma}_z^2$. 
Its calculation is simple, yet we will show that it could improve the estimation performance.

From \eqref{equ:tilde nu}, we can find that the SONP of the $m$-projection is independent of $n$. 
Thus, we can obtain $\bbnu_{0n}\left(t\right) = \bbnu_{0n'}\left(t\right), n,n'\in \setposi{N}$.
We now present the updatings of the parameters.
Since we replace $\bvartheta_{n}\left(t\right), n\in\setposi{N}$, with ${\bvartheta}\left(t\right)$, the approximation item $\bxi_{n}\left(t\right)$ can be re-expressed as
\begin{equation}\label{equ:new xi_n}
	\bxi_{n}\left(t\right) = {\bvartheta}_{0n}\left(t\right) - {\bvartheta}\left(t\right), n\in\setposi{N}.
\end{equation}
Then, from \eqref{equ:update of vartheta_n in IGA} , $\braces{\bvartheta_{n}\left(t+1\right)}_{n=1}^N$ can be obtained.
To update the common NP $\bvartheta$,
we calculate $\bvartheta\left(t+1\right)$ as the arithmetic mean of $\braces{\bvartheta_{n}\left(t+1\right)}_{n=1}^N$,
\begin{align}\label{equ:update of overlinevartheta in RIGA}
	{\bvartheta}\left(t+1\right) & = \frac{1}{N}\sum_{n=1}^{N}\bvartheta_{n}\left(t+1\right) \nonumber \\ 
&\equaa \frac{d}{N}\sum_{n=1}^{N}\sum_{n'= 1}^N\left(\bxi_{n'}\left(t\right) - \bxi_{n}\left(t\right)\right) \!+\! \frac{1-d}{N}\sum_{n=1}^N\bvartheta_{n}\left(t\right) \nonumber \\ 
& \equab \frac{d\left(N-1\right)}{N}\sum_{n=1}^{N}\bxi_{n}\left(t\right) + \left(1-d\right){\bvartheta}\left(t\right)\\ 
&\equac \frac{d\left(N-1\right)}{N}\sum_{n=1}^{N}{\bvartheta}_{0n}\left(t\right) + \left(1-dN\right){\bvartheta}\left(t\right),  \nonumber
\end{align} 
where $\left(\textrm{a}\right)$ comes from \eqref{equ:update of vartheta_n in IGA}, $\left( \textrm{b}\right)$ comes from that if $\bvartheta$ is updated as above, then at each iteration, $\bvartheta\left(t\right) = \frac{1}{N}\sum_{n=1}^{N}\bvartheta_{n}\left(t\right)$ can be obtained, and $\left( \textrm{c}\right)$ comes from \eqref{equ:new xi_n}.  
From \eqref{equ:update of vartheta_0 in IGA} ,
the update of the NP of $p_0\left(\bh;\bvartheta_{0}\right)$ can be modified as
\begin{align}\label{equ:update of vartheta_0 in RIGA 1}
		\!\!\!\bvartheta_{0}\left(t+1\right) 
		= d\sum_{n=1}^{N}{\bvartheta}_{0n}\left(t\right) -dN{\bvartheta}\left(t\right) + \left(1-d\right)\bvartheta_{0}\left(t\right).
\end{align}

We now discuss  the update of $\bvartheta_0$ in \eqref{equ:update of vartheta_0 in RIGA 1}, which is derived directly from the non-damping version of  \eqref{equ:update of overlinevartheta in RIGA} and \eqref{equ:update of vartheta_0 in RIGA 1}. 
Setting $d=1$ in \eqref{equ:update of overlinevartheta in RIGA} and \eqref{equ:update of vartheta_0 in RIGA 1}, and after some calculation, we can obtain 
\begin{equation*}\label{equ:e-condtion in iteration}
	\left(N-1\right)\bvartheta_0\left(t+1\right) = N{\bvartheta}\left(t+1\right).
\end{equation*}
Then, when $0<d<1$,
if we constrain $\left(N-1\right)\bvartheta_0\left(0\right) = N{\bvartheta}\left(0\right)$ at the initialization,  at each iteration of \eqref{equ:update of overlinevartheta in RIGA} and \eqref{equ:update of vartheta_0 in RIGA 1}, we still have $\left(N-1\right)\bvartheta_0\left(t\right) = N{\bvartheta}\left(t\right), \forall t$. 
In summary, when the initialization satisfies $\left(N-1\right)\bvartheta_0\left(0\right) = N{\bvartheta}\left(0\right)$, the update of the NPs can be summarized as follows: calculate $\bvartheta\left(t+1\right)$ as in the last equation of \eqref{equ:update of overlinevartheta in RIGA}, and calculate $\bvartheta_{0}\left(t+1\right)$ as
	\begin{equation}\label{equ:update of vartheta_0 in RIGA2}
		\bvartheta_{0}\left(t+1\right) = \frac{N}{N-1}\bvartheta\left(t+1\right).
	\end{equation}	
Moreover, the detailed expression of 	
\begin{equation*}
	\bvartheta\left(t+1\right) = \funf[\btheta\left(t+1\right),\bbnu\left(t+1\right)]
\end{equation*}
can be expressed as follows:
\begin{subequations}\label{equ:update of SIGA}
	\begin{equation}\label{equ:update of nu_0}
		\bbnu\left(t+1\right) = \tilde{\bg}\left(\bbnu\left(t\right)\right) \triangleq d \bg\left(\bbnu\left(t\right)\right) + \left(1-d\right)\bbnu\left(t\right),
	\end{equation}
	\begin{equation}\label{equ:function g text}
		\begin{split}
			\bg\left( \bbnu\left(t\right) \right) 
			\!=\! \!\left( 1\! - \! N  \right) \! \diag{\! \left(  \beta\left(\bbnu\left(t\right)\right)\bI - \bLambda\left(\bbnu\left(t\right)\right) \right)^{-1}  },
		\end{split}
	\end{equation}
\end{subequations}
\begin{subequations}\label{equ:intermideate variable}
	\begin{equation}\label{equ:update of theta_0}
		\btheta\left(t+1\right) 
		= \tilde{\bB}\left(\bbnu\left(t\right)\right)\btheta\left(t\right) + \bb\left(\bbnu\left(t\right)\right),
	\end{equation}
	\begin{equation}\label{equ:tilde B}
		\tilde{\bB}\left(\bbnu\left(t\right)\right) \triangleq d\bB\left(\bbnu\left(t\right)\right) + \left(1-d\right)\bI,
	\end{equation}
	\begin{equation}\label{equ:B}
		\begin{split}
			\bB\left(\bbnu\left(t\right)\right) = &\frac{N-1}{\beta\left(\bbnu\left(t\right)\right)}\left( \bI - \frac{1}{\beta\left(\bbnu\left(t\right)\right)}\bLambda\left(\bbnu\left(t\right)\right) \right)^{-1} \\
			&\ \ \times\left(\bI - \frac{1}{N}\bA^H\bA\right)\bLambda\left(\bbnu\left(t\right)\right),	
		\end{split}
	\end{equation}
	\begin{equation}
		\bb\left(\bbnu\left(t\right)\right) = \frac{2d\left(N-1\right)}{N\beta\left(\bbnu\left(t\right)\right)}\left(\bI - \frac{1}{\beta\left(\bbnu\left(t\right)\right)}\bLambda\left(\bbnu\left(t\right)\right)\right)^{-1}\bA^H\by,
	\end{equation}
\end{subequations}
where \eqref{equ:update of SIGA} is the iterating system of $\bbnu$, \eqref{equ:intermideate variable} is the iterating system of $\btheta$, and
the derivations are provided in Appendix \ref{proof:calculation of tilde_bvartheta_s}.
All the above matrices that need to be inverted are also shown to be invertible at each iteration in Appendix \ref{proof:calculation of tilde_bvartheta_s}, which guarantees that \eqref{equ:update of SIGA} and \eqref{equ:intermideate variable} are valid.
$\tilde{\bB}$ and $\bB$ are two matrix functions with $\bbnu\left(t\right)$ being the variable, i.e., $\tilde{\bB},\bB:\bbR^{M} \to \bbC^{M\times M}$,
and $\bb$, $\tilde{\bg}$ and $\bg$ are three vector functions with $\bbnu\left(t\right)$ being the variable, i.e., $\bb:\bbR^{M} \to \bbC^{M}$, 
and $ \tilde{\bg}, \bg:\bbR^{M} \to \bbR^{M}$. 
In \eqref{equ:update of SIGA} and \eqref{equ:intermideate variable}, the common NP $\bvartheta\left(t+1\right)$ is directly calculated without the step for calculating the approximation item $\bxi_{n}\left(t\right)$.
From \eqref{equ:update of vartheta_0 in RIGA2}, we can see that the NP of $p_0$ in each iteration relies on the common NP.
Therefore, its updating in the iteration process is not necessary.
We only need to calculate the NP  of $p_0$ with the resulting common NP from the iteration process.
We refer the above approach as EIGA and summarize it in Algorithm \ref{Alg:RIGA}.
The initialization of $\bbnu$ will be discussed in detail in Sec. \ref{sec:convergence} , and the range guarantees the convergence of EIGA.
\begin{algorithm}[htbp]
	\SetAlgoNoLine 
	\caption{EIGA}
	\label{Alg:RIGA}
	
	\KwIn{The covariance $\bD$ of the priori distribution $p\left(\bh\right)$, the received signal $\by$, the  noise power ${\sigma}_z^2$ and the maximal iteration number $t_{\mathrm{max}}$.}
	
	\textbf{Initialization:} set $t=0$, calculate the virtual noise variance $\tilde{\sigma}_z^2$ as $\tilde{\sigma}_z^2 = f\left(\sigma_z^2\right)$, where $f\left(\cdot\right)$ is given by \eqref{equ:function f beta},
	set damping $d$, where $0< d\le 1$, initialize the common NP as $\bvartheta\left(0\right) = $ $ \funf[\btheta\left(0\right),\bbnu\left(0\right)]$ and ensure $-\frac{N-1}{\tilde{\sigma}_z^2}\le\bbnu\left(0\right)\le0$; 
	
	\Repeat{\rm{Convergence or $t > t_{\mathrm{max}}$}}{
		1. Update $\bvartheta = \funf[\btheta,\bbnu]$ as \eqref{equ:update of SIGA} and \eqref{equ:intermideate variable}, where $\bLambda\left(\bbnu\left(t\right)\right)$ and $\beta\left(\bbnu\left(t\right)\right)$ are given by \eqref{equ:Lambda} and \eqref{equ:beta}, respectively;\\
		2. $t = t+1$;}
	
	\KwOut{\rm{Calculate the NP of $p_0\left(\bh;\bvartheta_{0}\right)$ as $\bvartheta_{0} = \frac{N}{N-1}\bvartheta\left(t\right)$.	The mean and variance of the approximate marginal, $p\left( h_i|\by\right)$, $i\in \setposi{M}$, are given by the $i$-th component of $\bmu_0$ and $\diag{\bSigma_{0}}$, respectively, where $\bmu_0$ and $\bSigma_{0}$ are calculated by \eqref{equ:mu_0} and \eqref{equ:Sigma_0}, respectively. }}
\end{algorithm}


\subsection{Efficient Implementation}\label{sec:efficient implementation of SIGA}
The computational complexity of each iteration of EIGA mainly comes from the two matrix-vector multiplications by $\bA$ and $\bA^H$ in \eqref{equ:intermideate variable}.
In this subsection, we focus on \eqref{equ:intermideate variable} and present an efficient implementation.
We assume that the adjustable phase shift pilots (APSPs) \cite{channelaqyou} are adopted as the training signal, which is an extension of the conventional phase shift orthogonal pilots in LTE and 5G NR \cite{dahlman20134g,dahlman20205g}.
Note that any other pilot sequences with constant magnitude can be adopted. 
We set the transmit power of the training signal for each user to $1$.
Then, the APSP for the user $k$ is set to be $\bP_k = \Diag{\br\left(n_k\right)}\bP$, where
\begin{equation*}
	\br\left(n_k\right) = \left[\exp\braces{-\barjmath 2\pi\frac{n_kN_1}{F_\tau N_p}}, \cdots, \exp\braces{-\barjmath 2\pi\frac{n_kN_2}{F_\tau N_p}} \right]^T ,
\end{equation*} 
$n_k \in \braces{0, 1, \cdots, F_\tau N_p-1}$ is the phase shift scheduled for the user $k$, and $\bP = \Diag{\bp}$ is the  basic pilot satisfying $\bP\bP^H = \bI$.  
Given the channel power matrix $\bOmega_k, k\in \setposi{K}$, we can use \cite[Algorithm 1]{channelaqyou} to determine the value of $n_k$ and thus $\bP_k, k\in \setposi{K}$.
Define a partial DFT matrix of $F\tau N_p$ points as 
\begin{equation}\label{equ:def of Fd}
	\bF_d \triangleq \left[\br\left(0\right), \ \br\left(1\right), \ \cdots, \ \br\left( F_\tau N_p - 1 \right)   \right]\in \bbC^{N_p \times F_\tau N_p}
\end{equation}
and a permutation matrix as
\begin{equation}
	\bPi_{n_k} \triangleq \left[  
	\begin{matrix}
		\bO \ &\bI_{F_\tau N_p - n_k} \\ 
		\bI_{n_k} \ & \bO
	\end{matrix}
	\right] \in \bbC^{F_\tau N_p \times F_\tau N_p}.
\end{equation}
Substituting $\bP_k$ and \eqref{equ:SF beam channel} into \eqref{equ:maMIMO rece signal}, we can obtain 
\begin{equation*}
	\bY = \bV\bH_a\bF_d^T\bP + \bZ,
\end{equation*}
where 
\begin{equation*}
	\bH_a = \sum_{k=1}^{K}\bH_k^{\textrm{e}}\bPi_{n_k},
\end{equation*}
$\bH_k^{\textrm{e}} = \left[ \bH_k, \ \bO \right] \in \bbC^{F_aN_r \times F_\tau N_p}, k\in \setposi{K}$, is the extended beam domain channel matrix for the user $k$. 
Define $\bOmega_a \triangleq  \sum_{k=1}^{K}\bOmega_k^{\textrm{e}}\bPi_{n_k}$ with $\bOmega_k^{\textrm{e}} \triangleq \left[ \bOmega_k, \ \bO \right] \in \bbC^{F_aN_r \times F_\tau N_p}$. 
It is not difficult to check that $\bOmega_a$ is the power matrix of $\bH_a$.
Then, we can obtain 
\begin{equation*}\label{equ:rece signal 1}
	\by_p = \mathrm{vec} \braces{\bY\bP^H} = \tilde{\bA}_p\tilde{\bh}_a + \bz_p,
\end{equation*}
where 
\begin{equation}\label{equ:def of tilde Ap}
	\tilde{\bA}_p = \bF_d \otimes \bV \in \bbC^{N\times F_aF_\tau N},
\end{equation}
$\tilde{\bh}_a \in \bbC^{F_a F_\tau N \times 1}$ is the vectorization of $\bH_a$, and $\bz_p \in \bbC^{N\times 1}$ is the  vectorization of $\bZ\bP^H$.
Since $\bP^H$ is unitary, we can readily obtain that $\bz_p \sim \mathcal{CN}\left(\mathbf{0},\sigma_z^2\bI\right)$.
Define the number of non-zero components in $\bomega_a \triangleq \mathrm{vec}\braces{\bOmega_a}$ as $M_a \triangleq \norm{\bomega_a}_0$ and the indexes of non-zero components in $\bomega_a$ as $\mathcal{Q} \triangleq \braces{q_1, q_2, \cdots, q_{M_a}}$, where $1\le q_i\le F_a F_\tau N$.
The extraction matrix is defined as 
\begin{equation}
\bE_p \triangleq \left[\be_{q_1}, \ \be_{q_2},  \cdots, \be_{q_{M_a}}\right] \in \bbC^{F_a F_\tau N \times M_a},
\end{equation}
where $\be_i\in\bbC^{\tilde{M}\times 1}, i\in\mathcal{P}$ is the $i$-th column of the $\tilde{M}$ dimensional identity matrix.
Then, $\by_p$ can be re-expressed as 
\begin{equation}\label{equ:rece signal 2}
	\by_p = \bA_p\bh_a + \bz_p,
\end{equation}
where 
\begin{equation}\label{equ:A in APSP}
	\bA_p  = \tilde{\bA}_p \bE_p \in \bbC^{N \times M_a},
\end{equation}
$\bh_a = \bE_p^T\tilde{\bh}_a \in \bbC^{M_a\times 1}$, $\bh_a \sim \mathcal{CN}\left(\mathbf{0}, \bD_a\right)$ and $\bD_a \triangleq \Diag{\bE_p^T\bomega_a}$ is positive definite diagonal. 
In this case, at each iteration of EIGA, \eqref{equ:intermideate variable}
can be rewritten as (we omit the counter $t$ on the right-side of the equation for convenience)
	\begin{equation}\label{equ:theta t+1}
		\begin{split}
	    	\btheta\left(t+1\right) = &\frac{2}{\beta}\bJ_p\bA_p^{H}\by_p - \frac{1}{\beta}\bJ_p\bA_p^{H}\bA_p\bLambda\btheta\left(t\right) \\
	    	 &+ \left[ N\bJ_p + \left(1-dN\right)\bI \right]\btheta\left(t\right),
		\end{split}
	\end{equation}
where $\bJ_p = \frac{d\left(N-1\right)}{N}\left( \bI - \frac{1}{\beta}\bLambda \right)^{-1}$.
Since both $\bJ_p$ and $\bLambda$ are diagonal, the complexity in \eqref{equ:theta t+1} mainly comes from $\bA_p^H\by_p$, $\bA_p^H\bs$ and $\bA_p\bu$, where $\bs = \bA_p\bLambda\btheta\left(t\right) \in \bbC^{N \times 1}$ and $\bu = \bLambda\btheta\left(t\right) \in \bbC^{M_a \times 1}$. 
For $\bA_p\bu$, we have 
\begin{equation*}
	\bA_p\bu = \tilde{\bA}_p\tilde{\bu} = \mathrm{vec}\braces{ \bV\tilde{\bU}\bF_d^T },
\end{equation*}
where $\tilde{\bu} = \bE_p\bu \in \bbC^{F_aF_\tau N\times 1}$, $\tilde{\bU} \in \bbC^{F_aN_r \times F_\tau N_p}$ and $\vecc\braces{\tilde{\bU}} = \tilde{\bu}$.
Then, $\bV\tilde{\bU}\bF_d^T$ can be calculated by FFT since $\bV$ is the Kronecker product of two partial DFT matrices and $\bF_d$ is a partial DFT matrix. The complexity of the efficient implementation of $\bA_p\bu$ is $\mathcal{O}\left(C\right)$, where 
\begin{equation}\label{equ:C}
	\begin{split}
			C = &N\Big[ F_aF_\tau \logtwo{F_vN_{r,v}} + F_hF_\tau \logtwo{F_hN_{r,h}} \\
		&+ F_\tau \logtwo{F_\tau N_p} \Big].
	\end{split}
\end{equation}
For the calculation of $\bA_p^H\bs$, we have that
\begin{equation*}
	\bA_p^H\bs = \bE_p^T\tilde{\bA}_p^H\bs = \bE_p^T\vecc\braces{\bV^H\bS\bF_d^*},
\end{equation*}
where $\bS \in \bbC^{N_r\times N_p}$ and $\vecc\braces{\bS} = \bs$.
We first compute $\bS' \triangleq \bS\bF_d^* \in\bbC^{N_r\times F_\tau N_p}$ and then $\bV^H\bS'$. 
Both of the above two calculations can be implemented through inverse FFT (IFFT).
Then, $\bE_p^T\tilde{\bA}_p^H\bs$ is equivalent to extracting the components from $\tilde{\bA}_p^H\bs$ with the indexes determined by $\mathcal{Q}$.
The complexity of the efficient implementation of $\tilde{\bA}_p^H\bs$ is $\mathcal{O}\left({C}\right)$, too.
As for the calculation of $\bA_p^H\by_p$, since it is the same at each iteration, we only need to calculate it once.
The calculation of $\bA_p^H\by_p$ and the corresponding complexities are the same as that of  $\bA_p^H\bs$ in one iteration.  


%
%

\section{Convergence and Fixed point}
In this section, we give the convergence and fixed point analyses of EIGA.
We prove that with a sufficient small damping, EIGA is guaranteed to converge.
We determine a wider range of damping through the specific properties of the measurement matrix.
Then, we show that at the fixed point, the \textsl{a posteriori} mean obtained by EIGA is asymptotically optimal. 
The calculation of the virtual noise variance $\tilde{\sigma}_z^2$ will be also presented.

\subsection{Convergence}\label{sec:convergence}
We begin with following lemma related to the range of $\bbnu$.
\begin{lemma}\label{prop:0}
	Given a finite initialization  $\bvartheta\left(0\right) = \funf[\btheta\left(0\right),\bbnu\left(0\right)]$ with $-\frac{N-1}{\tilde{\sigma}_z^2}\mathbf{1}\le  \bbnu\left(0\right) \le  \bzero$, then at each iteration,  $\bvartheta\left(t\right) = \funf[\btheta\left(t\right),\bbnu\left(t\right)]$ satisfies: $\btheta\left(t\right)$ and $\bbnu\left(t\right)$ are finite, and $\bbnu\left(t\right) \le  \bzero$.
	Specifically, we have $\bbnu\left(0\right) \le \bzero$ and $\bbnu\left(t\right) < \bzero, t \ge 1$.
\end{lemma}
\begin{proof}
	See in Appendix \ref{proof:calculation of tilde_bvartheta_s}.
\end{proof}
We refer matrix $\tilde{\bB}$ in \eqref{equ:tilde B} as the iterating system matrix of $\btheta$, which is determined by the common SONP $\bbnu$ and the measurement matrix $\bA$ at each iteration.
Combining \eqref{equ:update of nu_0} and \eqref{equ:update of theta_0}, we can see that $\bbnu\left(t+1\right)$ only depends on $\bbnu\left(t\right)$ and does not depend on $\btheta\left(t\right)$, 
while $\btheta\left(t+1\right)$ depends on both $\btheta\left(t\right)$ and $\bbnu\left(t\right)$. 
This shows that the iterating system of $\bbnu$ is separated from that of $\btheta$, and hence, the convergence of $\bbnu\left(t\right)$ can be checked individually.
We then consider the convergence of $\bbnu\left(t\right)$.
To this end, we first present the following lemma about the function $\tilde{\bg}\left(\bbnu\right)$ defined in \eqref{equ:update of nu_0}.

\begin{lemma}\label{lemma:function tilde{g}}
	Given $ \bbnu \le \bzero$, $\tilde{\bg}\left(\bbnu\right)$ satisfies the following two properties.\\
	1. Monotonicity: If $\bbnu < \bbnu' \le \bzero$, then $\tilde{\bg}\left( \bbnu\right) < \tilde{\bg}\left(\bbnu'\right)$.\\
	2. Scalability: Given a positive constant $0<\alpha<1$, we have $\tilde{\bg}(\alpha\bbnu)<\alpha\tilde{\bg}(\bbnu)$.\\
	Moreover, if $\tilde{\bg}_{\textrm{min}} \le \bbnu \le \bzero$ with $\tilde{\bg}_{\textrm{min}} \triangleq -\frac{N-1}{\tilde{\sigma}_z^2}\mathbf{1} \in \bbR^M$, we have $\tilde{\bg}_{\textrm{min}} < \tilde{\bg}\left(\bbnu\right) < \bzero$.
\end{lemma}
\begin{proof}
	See in  Appendix \ref{proof:function tilde{g}}.
\end{proof}

Based on Lemma \ref{lemma:function tilde{g}}, we have the following theorem. 
\begin{theorem}\label{The-theta_1 conv}
	Given initialization $\bbnu\left(0\right)$ with $ \tilde{\bg}_{\textrm{min}}\le  \bbnu\left(0\right) \le  \bzero$, where $\tilde{\bg}_{\textrm{min}}$ is defined in Lemma \ref{lemma:function tilde{g}}, the sequence $\bbnu\left(t+1\right) = \tilde{\bg}\left( \bbnu\left(t\right) \right)$ converges to a finite fixed point $\bbnu^\star$, where $\tilde{\bg}_{\textrm{min}}< \bbnu^{\star} < \bzero$.
\end{theorem}
\begin{proof}
	See in Appendix \ref{proof:theta_1 conv}.
\end{proof}
From Theorem \ref{The-theta_1 conv}, we can find that $\bbnu\left(t\right)$ converges to a finite fixed point as long as its initialization satisfies $ \tilde{\bg}_{\textrm{min}}\le \bbnu\left(0\right) \le \bzero$, and this range can be quite large.
For example, in our simulations, $N = 46080$, and when the virtual noise variance $\tilde{\sigma}_z^2 = 1$,  we obtain $\tilde{\bg}_{\textrm{min}} = -46079\times \bone$. 
In this case, the range of the initialization of $\bbnu$ is quite large.
Theorem \ref{The-theta_1 conv} also shows that the convergence of $\bbnu\left(t\right)$ is not related to the damping factor $d$.
Yet we shall see that the convergence of $\btheta\left(t\right)$ is related to the choice of the damping factor later.
Define 
\begin{equation}\label{equ:Lambda0 star}
	\bLambda^\star = \left( \bD^{-1} - \Diag{\bbnu^\star} \right)^{-1},
\end{equation}
\begin{equation}\label{equ:beta0 star}
	\beta^\star = \tilde{\sigma}_z^2 + \tr{\bLambda^\star}.
\end{equation}
From Theorem \ref{The-theta_1 conv}, diagonal matrix $\bLambda^\star$ is positive definite and $\beta^\star > 0$.
From \eqref{equ:update of nu_0} and ${\bbnu^\star} = \tilde{\bg}\left(\bbnu^\star\right)$, we have $\bbnu^\star = \bg\left(\bbnu^\star\right)$.
Then, we can obtain the following relationship for $\bbnu^{\star}$
\begin{equation}\label{equ:relation of bbnu}
	\frac{N}{N-1}\bbnu^\star = \diag{\bD^{-1} - \left( \bLambda^\star - \frac{1}{\beta^\star}\left(\bLambda^\star\right)^2 \right)^{-1}},
\end{equation}
where the derivation is given in Appendix \ref{proof:Calculation 2}.
From \eqref{equ:relation of bbnu}, we have
\begin{equation}\label{equ:condition of nu}
	\bLambda^\star - \frac{1}{\beta^\star}\left(\bLambda^\star\right)^2 = \left(\bD^{-1} - \frac{N}{N-1}\Diag{\bbnu^\star}   \right)^{-1}.
\end{equation}

Define 
\begin{equation}\label{equ:tilde B 0}
	\tilde{\bB}^{\star} = \tilde{\bB}\left( \bbnu^{\star} \right) = d\bB^{\star} + \left(1-d\right)\bI,
\end{equation}
where $\bB^{\star} = {\bB}\left( \bbnu^{\star} \right)$ and 
$\bb^\star = \bb\left( \bbnu^{\star} \right)$.
From the definition, $\tilde{\bB}^\star$ is determined by the fixed point of the common SONP $\bbnu^\star$ and the measurement matrix $\bA$, which does not vary with iterations.
To avoid ambiguity,  the iterating system matrix refers to $\tilde{\bB}^\star$ in the rest of the paper, since the convergence condition for the iterating system of $\btheta\left(t\right)$ given in the following lemma depends only on the spectral radius of $\tilde{\bB}^\star$.
\begin{lemma}\label{The-synchronous updates} 
	Given a finite initialization $\btheta\left(0\right) \in \bbC^{M \times 1}$ and $\bbnu\left(0\right)$ with $-\frac{N-1}{\tilde{\sigma}_z^2}\mathbf{1}\le   \bbnu\left(0\right) \le  \bzero$.
	Then,
	$\btheta\left(t\right)$ in \eqref{equ:intermideate variable} converges to its fixed point  if the spectral radius of $\tilde{\bB}^\star$ is less than $1$, i.e., $\rho\left(\tilde{\bB}^{\star}\right) < 1$, with $\rho\left(\tilde{\bB^\star}\right) = \max\braces{\left|\lambda\right|: \lambda \  \textrm{is an eigenvalue of } \tilde{\bB^\star}}$.  
\end{lemma} 
\begin{proof}
	See in Appendix \ref{Proof-The-syn}.
\end{proof}

From Lemma \ref{The-synchronous updates}, we see that when $\bbnu$ converges and the spectral radius of the iterating system matrix in \eqref{equ:tilde B 0}, i.e., $\tilde{\bB}^{\star}$, is less than $1$, $\btheta$ converges. 
We next give an analysis of the eigenvalue distribution of $\tilde{\bB}^\star$ and a theoretical explanation for the improved convergence of $\btheta$ under damped updating.
The key point in the next is to analyze the eigenvalues of $\tilde{\bB}^\star$.
We begin with the eigenvalues  of $\bB^\star$ in \eqref{equ:tilde B 0} since from \eqref{equ:tilde B 0}, the eigenvalues of $\tilde{\bB}^{\star}$ can be directly obtained from those of $\bB^\star$.
As mentioned above, 
when $\bbnu$ converges to $\bbnu^{\star}$, from \eqref{equ:condition of nu} and \eqref{equ:tilde B 0}, it is not difficult to obtain 
\begin{equation}\label{equ:tilde B star 1}
	\begin{split}
		\bB^\star &= \frac{N-1}{\beta^{\star}}\left( \bI - \frac{1}{\beta^{\star}}\bLambda^\star \right)^{-1}
		\left(\bI -\frac{1}{N}\bA^H\bA\right)\bLambda^\star \\
		&=	 \left( \bI - \frac{1}{N}\bD^{-1}\bLambda^{\star} \right)\left(N\bI - \bA^H\bA\right)\left( \frac{1}{\beta^{\star}}\bLambda^{\star} \right).
	\end{split}
\end{equation}
We can find that ${\bB}^\star$ is the product of three matrices. 
The first matrix of the right hand side of \eqref{equ:tilde B star 1} satisfies the following property.
\begin{lemma}\label{prop:1}
	$\bI - \frac{1}{N}\bD^{-1}\bLambda^{\star}$ is diagonal with diagonal components all  positive and smaller than $1$. 
\end{lemma}
\begin{proof}
	From \eqref{equ:Lambda0 star}, we can obtain that $\bzero<\diag{\bD^{-1}\bLambda^{\star}} < \mathbf{1}$, which implies that the diagonal of $\bI - \frac{1}{N}\bD^{-1}\bLambda^{\star}$ is positive and smaller than $1$.
	This completes the proof.
\end{proof}

Since all the three matrices in the product in \eqref{equ:tilde B star 1} are Hermitian, we have \cite[Exercise below Theorem 5.6.9]{horn2012matrix}
\begin{equation}\label{equ:range of eigen of B}
	\begin{split}
		&\rho\left(\bB^{\star}\right) \\
		\le &\rho\left( \bI - \frac{1}{N}\bD^{-1}\bLambda^{\star} \right)\rho\left(N\bI - \bA^H\bA\right)\rho\left( \frac{1}{\beta^{\star}}\bLambda^{\star} \right).
	\end{split}
\end{equation}
From Lemma \ref{prop:1}, we can obtain that
\begin{equation}\label{equ:aux1 in text}
	\rho\left( \bI - \frac{1}{N}\bD^{-1}\bLambda^{\star}  \right) <1.
\end{equation} 
%
%
\begin{lemma}\label{Lemma-spectral radius}
	The spectral radius of $\bLambda^*$ satisfies 
	\begin{equation}\label{equ:aux in lemma5}
		\rho\left( \bLambda^\star \right) < \frac{\beta^\star}{N}.
	\end{equation}
\end{lemma}
\begin{proof}
	See in Appendix \ref{{Proof-Lemma-spec}}.
\end{proof}

We next show some properties of the eigenvalues  of $\bB^{\star}$ and $\tilde{\bB}^\star$.
\begin{lemma}\label{lemma:eigens of B star}
	Denote the eigenvalues of $\bB^{\star}$ as $\lambda_{B,i}, i\in \setposi{M}$.
	Then,  $\braces{\lambda_{B,i}}_{i=1}^M$ are all real and
	\begin{equation}\label{equ:range of v_B}
		-\frac{\rho\left(N\bI - \bA^H\bA\right)}{N} < \lambda_{B,i} <1.
	\end{equation}
\end{lemma}
\begin{proof}
	See in Appendix \ref{proof:eigens of B star}.
\end{proof}

Denote the eigenvalues of the iterating system matrix $\tilde{\bB}^{\star}$ in \eqref{equ:tilde B 0} as $\tilde{\lambda}_i, i \in \setposi{M}$.
From \eqref{equ:tilde B 0}, we have $\tilde{\lambda}_i = d\lambda_{B,i} + 1-d, i \in \setposi{M}$.
We then have the following lemma.
\begin{lemma}\label{lemma:eigenvalues of tilde B}
	The eigenvalues of $\tilde{\bB}^\star$ are all real and satisfy
	\begin{equation}
		1-d\left( 1+ \frac{\rho\left(N\bI - \bA^H\bA\right)}{N} \right) < \tilde{\lambda}_i <1.
	\end{equation}
\end{lemma}
\begin{proof}
	This is a direct result from Lemma \ref{lemma:eigens of B star}.
\end{proof}
From the above lemma, we can find that the eigenvalues of $\tilde{\bB}^\star$ are smaller than $1$, and their lower bound depends on the measurement matrix $\bA$ and the damping $d$.
Combining Lemmas \ref{lemma:eigens of B star} and \ref{lemma:eigenvalues of tilde B}, we have the following theorem.
\begin{theorem}\label{corol:rang of v'}
	Given a finite initialization $\btheta\left(0\right) \in \bbC^{M \times 1}$ and $\bbnu\left(0\right)$ with $-\frac{N-1}{\tilde{\sigma}_z^2}\mathbf{1}\le   \bbnu\left(0\right) \le  \bzero$.
	Then,
	$\btheta\left(t\right)$ in \eqref{equ:intermideate variable} converges to its fixed point  if the damping factor satisfies
	\begin{equation}\label{equ:range of d 0}
		d < \frac{2}{ 1+\frac{\rho\left(N\bI - \bA^H\bA\right)}{N}  }.
	\end{equation}
\end{theorem}
\begin{proof}
	This is a direct result from Lemmas \ref{The-synchronous updates} and \ref{lemma:eigenvalues of tilde B}.
\end{proof}

From Theorem \ref{corol:rang of v'}, we can find that EIGA will always converge with a sufficiently small damping factor  and the range of $d$ is mainly determined by $\rho\left(N\bI - \bA^H\bA\right)$. 
The spectral radius $\rho\left(N\bI - \bA^H\bA\right)$ depends on the measurement matrix $\bA$.
We next discuss the range of $\rho\left(N\bI - \bA^H\bA\right)$ in the worst case and give the range of damping factor accordingly.
The range of $\rho\left(N\bI - \bA^H\bA\right)$ and the corresponding range of damping factor in massive MIMO-OFDM channel estimation will be discussed later in this section.
\begin{theorem}\label{lemma:radius of NI-A^HA}
	The spectral radius of $N\bI - \bA^H\bA$ satisfies 
	\begin{equation}
		\rho\left(N\bI- \bA^H\bA \right)\le NM -N.
	\end{equation}
	If $\rank{\bA} = 1$, then $\rho\left(N\bI- \bA^H\bA \right) = NM-N$.
\end{theorem}
\begin{proof}
	See in Appendix \ref{proof:radius of NI-A^HA}.
\end{proof}
\begin{corol}\label{corol:1}
	Given a finite initialization $\btheta\left(0\right) \in \bbC^{M \times 1}$ and $\bbnu\left(0\right)$ with $-\frac{N-1}{\tilde{\sigma}_z^2}\mathbf{1}\le   \bbnu\left(0\right) \le  \bzero$.
	Then,
	$\btheta\left(t\right)$ in \eqref{equ:intermideate variable} converges to its fixed point  if the damping factor satisfies $d < \frac{2}{M}$.
\end{corol}
\begin{proof}
	It is a direct result from Theorems \ref{corol:rang of v'} and \ref{lemma:radius of NI-A^HA}.
\end{proof}

From Corollary \ref{corol:1}, we can find that in the worst case, if $d< \frac{2}{M}$, then EIGA converges.

Now, let us discuss the range of $\rho\left(N\bI - \bA^H\bA\right)$ in massive MIMO-OFDM channel estimation, where the range of $d$ will be expanded.
We first consider the case when general pilot sequences with constant magnitude
property are adopted.
In this case, $\bA$ is defined in \eqref{equ:maMIMO rece signal 4}.
From the definitions in \eqref{equ:def of V} and \eqref{equ:def of F}, we can obtain that $\bV_v\in\bbC^{N_{r,v} \times F_vN_{r,v}}$
and $\bV_h\in\bbC^{N_{r,h} \times F_hN_{r,h}}$ are partial DFT matrices, i.e., $	\bV_v = \tilde{\bI}_{N_{r,v}\times F_vN_{r,v}}\tilde{\bV}_v$ and $\bV_h = \tilde{\bI}_{N_{r,h}\times F_hN_{r,h}}\tilde{\bV}_h$,
$\tilde{\bV}_v$ and $\tilde{\bV}_h$ are $F_vN_{r,v}$ and $F_hN_{r,h}$ dimensional DFT matrices, respectively, $\tilde{\bI}_{N\times FN}$ is a matrix containing the first $N$ rows of the $FN$ dimensional identity matrix,  $F_v$ and $F_h$ are two fine (oversampling) factors.
$\bF$ can be re-expressed as
\begin{equation*}
	\bF  \triangleq \tilde{\bI}_{N_p\times F_\tau N_p}\tilde{\bF}\tilde{\bI}_{F_\tau N_p\times F_\tau N_f} 	\in \bbC^{N_p \times N_\tau N_f},
\end{equation*}
$\tilde{\bF}$ is the $F_\tau N_p$ dimensional DFT matrix, $\tilde{\bI}_{F_\tau N_p\times F_\tau N_f}$ is a matrix containing the first $F_\tau N_f$ columns of the $F_\tau N_p$ dimensional identity matrix, 
i.e., $\bF$ is the matrix obtained by $\tilde{\bF}$ after row extraction and column extraction.
Similarly, $\bF_d$ in \eqref{equ:def of Fd} can be re-expressed as
\begin{equation}
	\bF_d \triangleq \tilde{\bI}_{N_p\times F_\tau N_p}\tilde{\bF},
\end{equation}
and we have 
\begin{equation}
	\bF = \bF_d\tilde{\bI}_{F_\tau N_p\times F_\tau N_f}.
\end{equation}
From the definitions above, we can obtain that $	\bV_v\bV_v^H = F_vN_{r,v}\bI$, $	\bV_h\bV_h^H = F_hN_{r,h}\bI$, and $	\bF_d\bF_d^H = F_\tau N_p\bI$.
We then have the following theorem.
\begin{theorem}\label{lemma:SR in mMIMO-CE 1}
	For matrix $\bA$ in \eqref{equ:maMIMO rece signal 4}, we have, 
	\begin{equation}
		\rho\left(N\bI - \bA^H\bA\right) \le\left(  KF_vF_hF_\tau-1\right)  N.
	\end{equation}
	In this case, if 
	\begin{equation}
		d < \frac{2}{KF_vF_hF_\tau},	
	\end{equation}
	EIGA converges.
\end{theorem}
\begin{proof}
	See in Appendix \ref{proof:SR in mMIMO-CE 1}.
\end{proof}

In our simulations, $K = 48$, $M = 29277$, and $F_v= F_h= F_\tau=2$,
when general pilot sequences with constant magnitude
property are adopted, $d < 0.0052$ is sufficient to ensure the convergence of EIGA.
Note that this range is much larger than the worst case  $ d< \frac{2}{M} = 6.8\times 10^{-5}$ in Corollary \ref{corol:1}.
We finally consider the special case, where the adjustable phase shift pilots (APSPs) are used.
In this case, $\bA$ is equal to $\bA_p$ defined in \eqref{equ:A in APSP}.
And we have the following theorem.
\begin{theorem}\label{lemma:SR in mMIMO-CE 2}
	For $\bA$ in \eqref{equ:A in APSP}, we have, 
	\begin{equation}
		\rho\left(N\bI - \bA^H\bA\right) \le  \left(F_vF_hF_\tau -1\right)N.
	\end{equation}
	In this case, if 
	\begin{equation}
		d < \frac{2}{F_vF_hF_\tau},
	\end{equation}
	EIGA converges.
\end{theorem}
\begin{proof}
	See in Appendix \ref{proof:SR in mMIMO-CE 2}.
\end{proof}
For the case with $F_v = F_h = F_\tau = 2$,
	$d < 0.25$ is sufficient for EIGA to converge.

\subsection{Fixed Point}\label{sec:fixed point of simplified IGA}

In this subsection, we present the analysis for the fixed point of EIGA. 
The discussion on how to determine the value of $\tilde{\sigma}_z^2$ will also be presented.
Denote the fixed points of the common NP  and the NP of the $m$-projection in EIGA  as ${\bvartheta}^{\star} = \funf[{\btheta}^{\star},{\bbnu}^{\star}]$,  and ${\bvartheta}_{0n}^{\star} = \funf[{\btheta}_{0n}^{\star},\bbnu_{0n}^{\star}], n\in \setposi{N}$, respectively.
Denote the NP of $p_0$ at the fixed point of EIGA as
\begin{equation*}
	\bvartheta_{0}^{\star} = \funf[\btheta_0^{\star}, \bbnu_0^\star] \triangleq \frac{N}{\left({N-1}\right)}\bvartheta^{\star}.
\end{equation*}
Substituting  ${\bvartheta}\left(t+1\right) = {\bvartheta}\left(t\right) = {\bvartheta}^{\star}$ and $\bvartheta_{0n}\left(t\right) = {\bvartheta}_{0n}^{\star}$ into the last equation of
\eqref{equ:update of overlinevartheta in RIGA} (the step $1$ of Algorithm \ref{Alg:RIGA}), 
we can obtain the fixed point
\begin{equation}\label{equ:m-condition of RIGA1}
	\bvartheta_0^{\star}  = \frac{N}{N-1}\overlinebvartheta^{\star} =  \frac{1}{N}\sum_{n=1}^{N}\tildebvartheta_{0n}^{\star}.
\end{equation}
Comparing the first equation in \eqref{equ:m-condition of RIGA1} with $e$-condition in \eqref{equ:m and e conditions}, we can find that the fixed point of EIGA satisfies an alternative version of the $e$-condition since $\overlinebvartheta$ is calculated as the arithmetic mean of $\bvartheta_n, n\in\setposi{N}$. Then, from the second equation in \eqref{equ:m-condition of RIGA1}, we can obtain
\begin{subequations}
	\begin{equation}\label{equ:m-condition of RIGA1 new}
		\btheta_{0}^{\star} = \frac{1}{N}\sum_{n=1}^{N}\btheta_{0n}^{\star}
	\end{equation}
	\begin{equation}\label{equ:m-condition of RIGA2}
		\bbnu_{0}^{\star} = \frac{1}{N}\sum_{n=1}^{N}\bbnu_{0n}^{\star} = \bbnu_{0n}^{\star}, n\in \setposi{N},
	\end{equation}
\end{subequations}
where \eqref{equ:m-condition of RIGA2} comes from $\bbnu_{0n}\left(t\right) =\bbnu_{0n'}\left(t\right) , n,n'\in \setposi{N}, \forall t$.
Denote the means and covariance matrices of $p_0\left(\bh;\bvartheta_{0}^{\star}\right)$, $p_0\left(\bh;\tildebvartheta_{0n}^{\star}\right), n\in\setposi{N}$, and $p_n\left(\bh;\overlinebvartheta^{\star}\right), n\in\setposi{N}$, as 
\begin{subequations}\label{equ:definition of fixed point of RIGA1}
	\begin{equation}\label{equ:fixed point of mu0 and Sigma0 in RIGA}
		\bmu_0^{\star} = \bmu_0\left(\bvartheta_{0}^{\star}\right),\ \bSigma_{0}^{\star} = \bSigma_{0}\left(\bvartheta_{0}^{\star}\right),
	\end{equation}
	\begin{equation}
		\bmu_{0n}^{\star} = \bmu_0\left(\tildebvartheta_{0n}^{\star}\right),\ \bSigma_{0n}^{\star} = \bSigma_{0}\left( \tildebvartheta_{0n}^{\star} \right), n\in\setposi{N},
	\end{equation}
	\begin{equation}\label{equ:fixed point of mun and Sigman in RIGA}
		\bmu_n^{\star} = \bmu_{n}\left(\overlinebvartheta^{\star}\right),\ \bSigma_{n}^{\star} = \bSigma_{n}\left(\overlinebvartheta^{\star}\right), n\in\setposi{N},
	\end{equation}
\end{subequations}
where functions $\bmu_n\left(\cdot\right)$ and $\bSigma_n\left(\cdot\right), n\in\setnnega{N}$, are given by \eqref{equ:mu_0 and Sigma_0} and \eqref{equ:mu_n and Sigma_n}, respectively.
Then, we have the following lemma.

\begin{lemma}\label{lem:m and e-condition of RIGA}
	At the fixed point of EIGA, the mean of $p_0\left(\bh;\bvartheta_{0}^{\star}\right)$ on the OBM is equal to the arithmetic mean of the means of $p_1\left(\bh;\overlinebvartheta^{\star}\right), p_2\left(\bh;\overlinebvartheta^{\star}\right), \cdots, p_N\left(\bh;\overlinebvartheta^{\star}\right)$, on the AMs. 
	Meanwhile, the variance of $p_0\left(\bh;\bvartheta_{0}^{\star}\right)$ is equal to the variance of $p_n\left(\bh;\overlinebvartheta^{\star}\right), n\in\setposi{N}$, i.e., 
		\begin{subequations}
				\begin{equation}
					\bmu_0^{\star} = \frac{1}{N}\sum_{n=1}^{N}\bmu_n^{\star},
					\end{equation}
			    \begin{equation}\label{equ:Sigma equation of the fixed point in RIGA}
					\diag{\bSigma_{0}^{\star}} = \diag{\bSigma_{n}^{\star}}, n\in\setposi{N}.
				    \end{equation}
			\end{subequations}
\end{lemma}
\begin{proof}
	See in Appendix \ref{proof:m and e-condition of RIGA}.
\end{proof}

From Lemma \ref{lem:m and e-condition of RIGA}, the two conditions of the fixed point of EIGA are summarized as
\begin{equation}\label{equ:m and e-conditions of RIGA}
	\begin{cases}
		m\textrm{-condition:}\ \bbeta_0\left(\bvartheta_0^{\star}\right) = \frac{1}{N}\sum_{n=1}^{N}\bbeta_n\left(\overlinebvartheta^{\star}\right), \\
		\ e\textrm{-condition:}\ \bvartheta_{0}^{\star} = \frac{N}{N-1}\overlinebvartheta^{\star},
	\end{cases}
\end{equation}
where $\bbeta_n\left(\cdot\right) = \left[ \bmu_{n}^T\left(\cdot\right) \ \mathrm{diag}^T\braces{\bSigma_{n}\left(\cdot\right)} \right]^T \in \bbC^{2M\times 1}$.

For the fixed point of EIGA in the asymptotic case, we first present the following theorem.
\begin{theorem}\label{the:the expression of mu_0^star in S-IGA}
	If the initialization of the SONP of the common NP in EIGA  satisfies $\overlinebnu\left(0\right) \le  0$, then, the fixed points of the SONPs of the NP of $p_0$ satisfy $\bbnu_0^{\star} <0$, and the fixed point of $\bmu_0$ defined in \eqref{equ:fixed point of mu0 and Sigma0 in RIGA} satisfies 
		\begin{equation}\label{equ:mu_0^star 1}
				\bmu_0^{\star} = \bD\left[ \bA^H\bA\left( \bD- \frac{1}{N}\bLambda^{\star} \right) + \beta^{\star}\bI  \right]^{-1}\bA^H\by,
			\end{equation}
	where 
	\begin{subequations}\label{equ:Lambda^star and beta^star}
		\begin{equation}\label{equ:Lambda^star first}
			\bLambda^{\star} \triangleq \left( \bD^{-1} - \Diag{\overlinebnu^{\star}} \right)^{-1},
		\end{equation}
		\begin{equation}\label{equ:beta^star first}
			\beta^{\star} \triangleq \tilde{\sigma}_z^2 + \tr{\bLambda^{\star}} > 0.
		\end{equation}
	\end{subequations}
\end{theorem}
\begin{proof}
	See in Appendix \ref{proof:the expression of mu_0^star in S-IGA}.
\end{proof}
Theorem \ref{the:the expression of mu_0^star in S-IGA} provides the expression of $\bmu_0^{\star}$ in EIGA.
We then show that $\bmu_{0}^{\star}$ above can be asymptotically optimal when $M < N$ and $N$ tends to infinity, where $M$ and $N$ are the numbers of the variables to be estimated and the observations, respectively.
We first define an injection as $f: \bbR^{+}\to\bbR$, 
\begin{equation}\label{equ:function f beta}
	f\left(x\right) = x - \tr{\left( \bD^{-1} + \frac{N-1}{x}\bI \right)^{-1}}, x>0.
\end{equation}
$f\left(x\right)$ plays an important role in the calculation of the virtual noise variance.
\begin{lemma}\label{lemma:bijection of f}
	If $M<N$, then $f\left(x\right)$ is a monotonically increasing function, and we have 
		\begin{equation*}
				f\left(x\right) = f\left(y\right) \iff x = y,
			\end{equation*} 
	where $x, y, f\left(x\right), f\left(y\right) >0$.
\end{lemma}
\begin{proof}
	See in Appendix \ref{proof:bijection of f}.
\end{proof}
Based on Lemma \ref{lemma:bijection of f}, we have the following theorem, which illustrates how we can determine the virtual noise variance.
\begin{theorem}\label{the:asymptotic value of mu_0 under case 2}
	When the initialization of the SONP of the common NP in EIGA satisfies $\bbnu\left(0\right) \le  0$ and $M <N$, the asymptotic values of $\bLambda^{\star}$ and $f\left(\beta^{\star}\right)$ as $N$ tends to infinity satisfy 
		\begin{subequations}
				\begin{equation}
						\lim\limits_{N\to \infty}\left[ \bLambda^{\star} \right]_{i,i} = 0, i\in \setposi{M},
					\end{equation}
			\begin{equation}
					\lim\limits_{N\to \infty}f\left(\beta^{\star}\right) = \tilde{\sigma}_z^2.
				\end{equation}
			\end{subequations}
	Then, if $\tilde{\sigma}_z^2 = f\left(\sigma_z^2\right)$, we can obtain $\lim\limits_{N\to\infty}\beta^{\star} = \sigma_z^2$ and 
	\begin{equation}\label{equ:Thm. equation}
		\lim\limits_{N\to \infty}\bmu_{0}^{\star} = \tilde{\bmu},
	\end{equation}
	where $\tilde{\bmu}$ is the \textsl{a posteriori} mean in \eqref{equ:post mean}.
	Moreover,  when $M$ is fixed, 
	we have $\lim\limits_{N\to\infty}f\left(\sigma_z^2\right) = \sigma_z^2$. 
	In this case, $\lim\limits_{N\to \infty}\bmu_{0}^{\star} = \tilde{\bmu}$  holds if either $\tilde{\sigma}_z^2 = \sigma_z^2$ or $\tilde{\sigma}_z^2 = f\left(\sigma_z^2\right)$ is satisfied.
\end{theorem}
\begin{proof}
	See in Appendix \ref{proof:asymptotic value of mu_0 under case 2}.
\end{proof}
Theorem \ref{the:asymptotic value of mu_0 under case 2}  provides the asymptotic values of $\bLambda^{\star}$ and $f\left(\beta^{\star}\right)$ when $N$ tends to infinity and $M < N$.
It also shows that if we set the virtual noise variance as the exact one, i.e., $\tilde{\sigma}^2 = \sigma_z^2$, then $\bmu_{0}^{\star}$ is asymptotically optimal as $N$ tends to infinity when $M$ is fixed.
Most importantly, $\bmu_{0}^{\star}$ is asymptotically optimal as $N$ tends to infinity if $\tilde{\sigma}_z^2$ is set to be $f\left(\sigma_z^2\right)$ and $M<N$. 
In massive MIMO-OFDM channel estimation, $M$ can be large when the number of users is large. In order to ensure that $M < N$ and guarantee asymptotically optimal performance of EIGA, an appropriate number of users can be chosen by using the statistical CSI of users in the BS.
Meanwhile, it can be checked that $0<f\left(\sigma_z^2\right) < \sigma_z^2$.

\section{Simulation Results}

In this section, we provide simulation results to illustrate
the complexity and performance of the proposed EIGA for massive MIMO-OFDM channel estimation. 
{The widely adopted QuaDRiGa \cite{quadriga} is used to generate the SF domain channel $\bG_{k}$ for each user.}
The simulation scenario is set as "3GPP\_38.901\_UMa", and main parameters for the simulations are summarized in Table \ref{tab:para}. 
\begin{table}[htbp]
	\centering
	\caption{Parameter Settings of the QuaDRiGa}\label{tab:para}
	\begin{tabular}{cc}
		\hline
		Parameter &Value \\
		\hline
		Number of BS antenna $N_{r,v}\times N_{r,h}$ & $8\times 16$ \\
		UT number $K$ & $48$ \\
		Center frequency $f_c$ & $4.8$GHz \\
		Number of training subcarriers $N_p$ & $360$ \\
		Subcarrier spacing $\Delta_f$ & $15$kHz \\
		Number of subcarriers $N_c$ & $2048$ \\
		CP length $\dnnot{N}{g}$ & $144$ \\		Mobile velocity of users &  $3-10$ kmph  \\
		\hline
	\end{tabular}
\end{table}
We locate the BS at $\left( 0,0,25 \right)$ and randomly generate the users in a $120^{\circ}$ sector with radius $r = 200$m around $(0, 0, 1.5)$. 
The $\textrm{SNR}$ is set as $\textrm{SNR} = \frac{1}{\sigma_z^2}$.
The APSPs are adopted as the pilot. 
We set the fine factors to $F_v = F_h = F_\tau = 2$ in all simulations,  which can achieve significant performance gain compared with the case with $F_v = F_h = F_\tau = 1$ as shown in \cite{IGA}.
It has also been shown that setting the fine factors to $2$ is sufficient to obtain good performance \cite{2Dlu,8067658,9967937}.
We adopt a standard Bayesian learning method proposed in \cite{8074806} to obtain the channel power matrix $\bOmega_k$ of each user from the generated SF domain channel $\bG_{k}$.
The number of total non-zero components in $\braces{\bOmega_k}_{k=1}^{48}$ is calculated as {$M_a = 29277$}, which is smaller than that of the observations $N = N_{r,v} \times N_{r,h} \times N_p = 46080$.
With a total of $48$ users, each user contains an average of $610$ variables to be estimated, i.e., the number of non-zero components in the channel power matrix of each user is $610$. 
This value is quite small when compared to the number of total components of the channel power matrix $\bOmega_k \in \bbC^{F_vF_hN_r\times F_\tau N_f}$ of each user, where $N_r = 128$  and $N_f = 26$ in our simulations.  
This coincides with the sparsity of the beam domain channel.
We use the normalized mean-squared error (NMSE) as the performance metric for the channel estimation,
\begin{equation*}
	\textrm{NMSE} = \frac{1}{KN_{\textrm{sam}}}\sum_{k=1}^{K}\sum_{n=1}^{N_{\textrm{sam}}}\frac{\lVert \bG_{k}^{\left(n\right)} - \hat{\bG}_{k}^{\left(n\right)} \rVert_F^2}{\lVert \bG_{k}^{\left(n\right)}  \rVert_F^2},
\end{equation*}
where $N_{\textrm{sam}}$ is the number of the channel samples, $\bG_{k}^{\left(n\right)}$ is the $n$-th channel sample of user $k$, $\hat{\bG}_{k}^{\left(n\right)}$ is the estimate of the $\bG_{k}^{\left(n\right)}$ and $\lVert \cdot \rVert_F$ is the F-norm. 
We set $N_{\textrm{sam}} = 1000$ in  our simulations.
Based on the received signal model \eqref{equ:rece signal 2},
we compare EIGA with the following algorithms. \\
\textbf{GAMP}: Generalized approximate message passing algorithm proposed in \cite{GAMP}.\\
\textbf{IGA}:  The original information geometry approach proposed in \cite{IGA}. \\
\textbf{MMSE}: The MMSE estimation of the beam domain
channels based on \eqref{equ:post mean}.



\subsection{Complexity}

The computational complexities of different  algorithms are summarized in Table \ref{tab:complex}, where $C$ is given by \eqref{equ:C}.
The actual computational complexity of different 
algorithms in our simulations are summarized in the Table \ref{tab:complex2}. We can find that the complexity of MMSE is the highest since a matrix-inversion is involved. 
On the other hand, owing to the utilization of the structure of $\bA_p$ in \eqref{equ:rece signal 2} and FFT, the complexity of EIGA is the
lowest among all the algorithms. Then, we combine the number of iterations to compare the overall computational complexities of EIGA. Taking the SNR $=10$dB as an example, from Fig. \ref{fig:10dB}, we can see that IGA and EIGA require about $200$ and $300$ iterations for convergence, respectively (GAMP requires around $600$ iterations). 
In this case, the overall computational complexity of EIGA is saved by $275$ times and $6.59\times10^{4}$ times compared to IGA and MMSE estimation, respectively.

\begin{table}[H]
	\centering
	\caption{COMPLEXITIES OF  ALGORITHMS}\label{tab:complex}
	\footnotesize
	\begin{tabular}{cc}
		\hline
		\textbf{Algorithm} & \textbf{Complexity} \\
		MMSE & $\mathcal{O}\left( M_a^3 + M_a^2N \right)$ \\
		{ GAMP$/$IGA (per iteration)} & $\mathcal{O}\left( NM_a \right)$ \\
		EIGA (per iteration) & $\mathcal{O}\left( {C} \right)$ \\
		\hline	
	\end{tabular}
\end{table}
\begin{table}[htbp]
	\centering
	\caption{ACTUAL COMPLEXITIES OF  ALGORITHMS}\label{tab:complex2}
	\begin{tabular}{cc}
		\hline
		\textbf{Algorithm} & \textbf{Complexity} \\
		MMSE  & $6.46\times 10^{13}$ \\
		{ GAMP/IGA (per-iteration) } & $1.35 \times 10^9$ \\
		EIGA (per-iteration)  & $3.27\times 10^6$ \\
		\hline	
	\end{tabular}
\end{table}


\subsection{Performance}

\begin{figure}[htbp]
	\centering
	\includegraphics[width=0.45\textwidth]{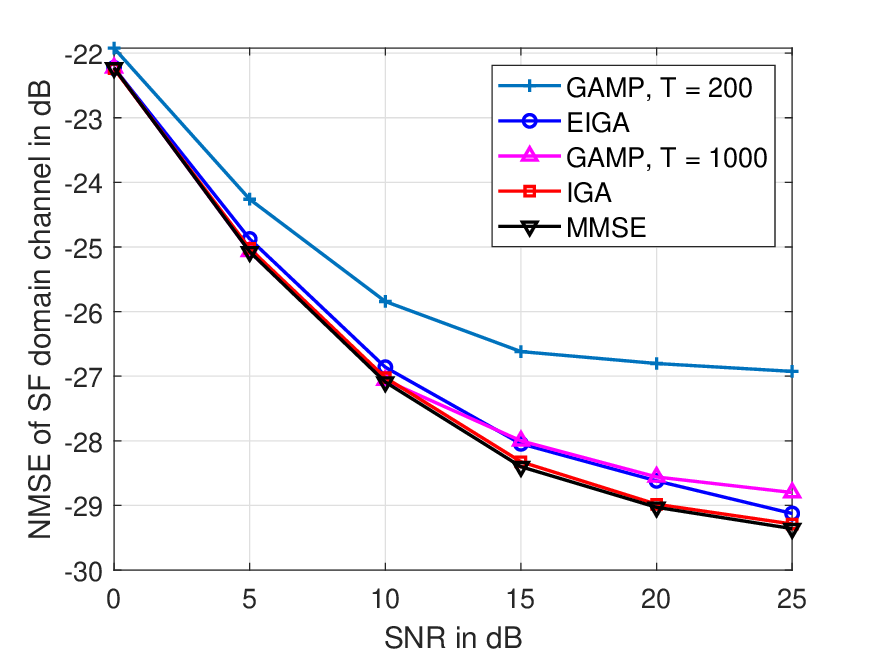}
	\caption{ NMSE performance of EIGA compared with GAMP, IGA and MMSE.}
	\label{fig:SNR}
\end{figure}

Fig. \ref{fig:SNR} shows the NMSE performance of EIGA channel estimation compared with GAMP, IGA and MMSE. 
The iteration numbers of EIGA and IGA are set as $200$. 
The  iteration number of GAMP is set as $200$ and $1000$. The damping factors of the iterative algorithms for different SNRs are summarized in Table \ref{tab:damping}.
\begin{table}[htbp]
	\centering
	\caption{Damping Factors}\label{tab:damping}
	\begin{tabular}{c|c|c|c|c|c|c}
		\hline
	 \multirow{2}*{Algorithm} & \multicolumn{6}{c}{SNR (dB)} \\
	\cline{2-7}
		        & 0  & 5 & 10 &15 &20 &25 \\
		\hline
	 GAMP   & 0.32     & 0.32 & 0.3 & 0.28 &0.28 &0.28 \\
		 IGA & 0.03 &0.03 &0.028 &0.025 &0.025 &0.025       \\
		 EIGA & 0.22 & 0.22 & 0.21 & 0.2 & 0.2  & 0.2\\
		\hline
	\end{tabular}
\end{table}
We can find that IGA with $200$ iterations and GAMP with $1000$ iterations can obtain almost the same NMSE performance as the MMSE estimation at all SNRs.  
The performance of EIGA can approach that of the MMSE estimation with a small gap.
The NMSE performance gain of EIGA compared to GAMP with $200$ iterations is about $1.3$dB when SNR is $20$dB.

\begin{figure}[htbp]
	\centering
	\includegraphics[width=0.45\textwidth]{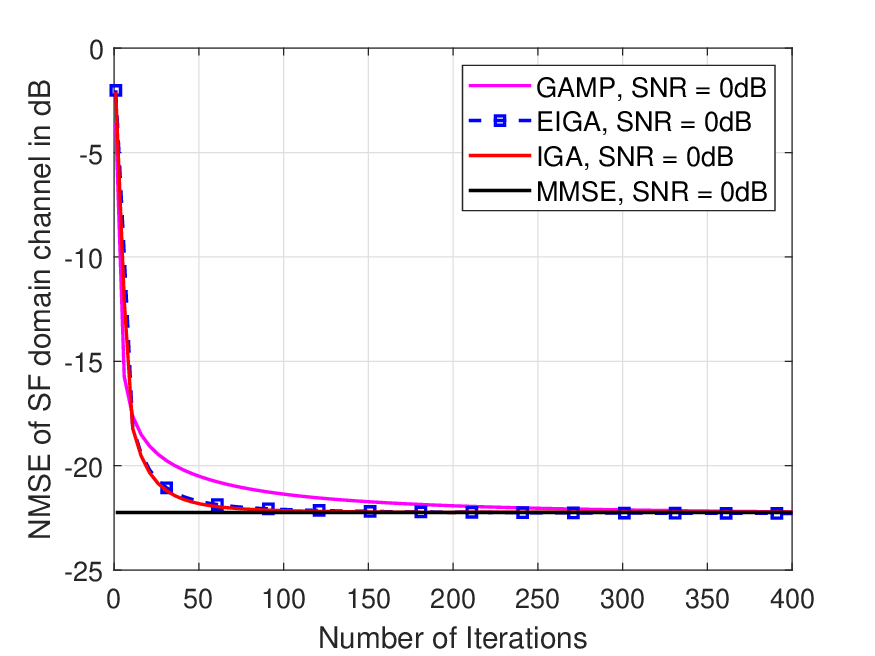}
	\caption{ Convergence performance of EIGA  compared with GAMP and IGA at SNR = $0$ dB.}
	\label{fig:0dB}
\end{figure}
\begin{figure}[htbp]
	\centering
	\includegraphics[width=0.45\textwidth]{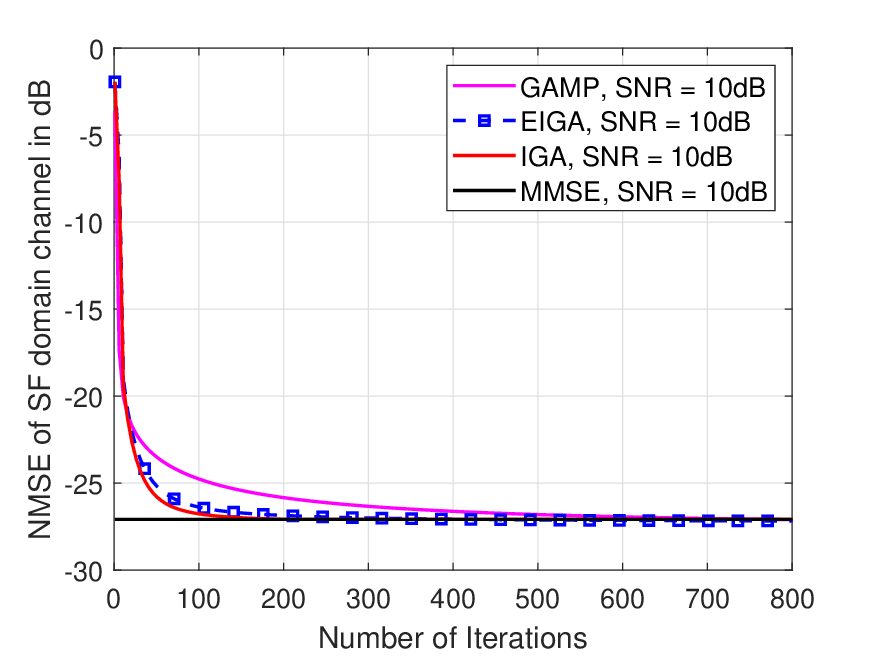}
	\caption{ Convergence performance of EIGA  compared with GAMP and IGA at SNR = $10$ dB.}
	\label{fig:10dB}
\end{figure}
\begin{figure}[htbp]
	\centering
	\includegraphics[width=0.45\textwidth]{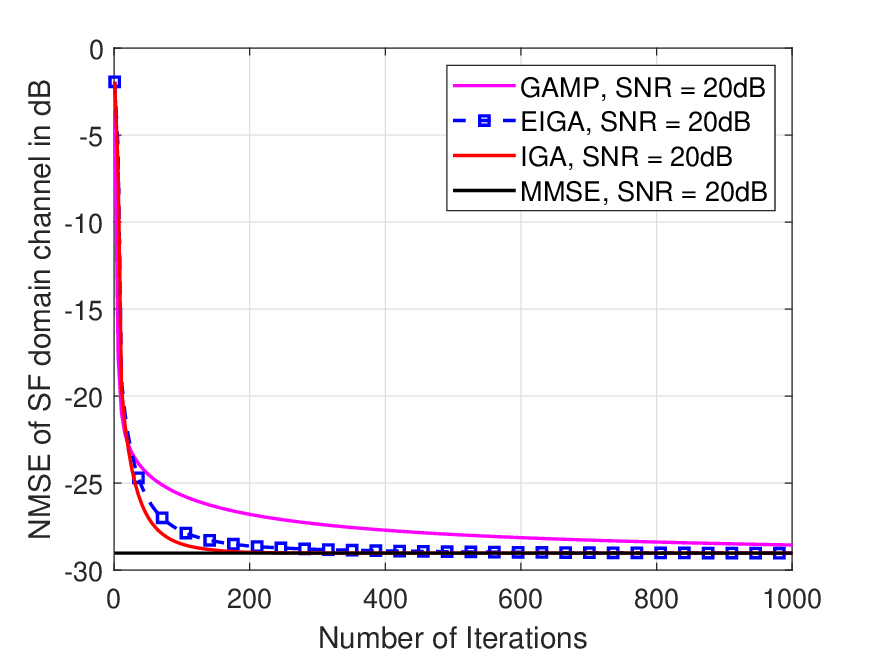}
	\caption{ Convergence performance of EIGA  compared with GAMP and IGA at SNR = $20$ dB.}
	\label{fig:20dB}
\end{figure}

Fig. \ref{fig:0dB} to Fig. \ref{fig:20dB} illustrate the convergence performance of EIGA compared with GAMP and IGA, where the SNR is set as $0$dB, $10$dB and $20$dB, respectively. 
In the case with SNR $=0$dB, EIGA and IGA require about $200$ and $120$ iterations, respectively, to converge and achieves the optimal solution as that by the MMSE estimation, while the GAMP needs around $300$ iterations to converge.
In the case with SNR $=10$dB, EIGA requires about $300$ iterations to converge, while IGA and GAMP converge in around  $200$ and $600$ iterations, respectively.
In the case with SNR $=20$dB, EIGA converges in about $300$ iterations, IGA requires about $200$  iterations to converge, while GAMP takes more than $1000$ iterations to converge. 
It can also be found that EIGA and IGA show similar convergence behavior,
while the computational complexity of EIGA is much lower than that of IGA.  
Compared with GAMP, EIGA converges with a faster rate. 
The EIGA along with the original IGA are developed based on the structure of the \textsl{a posteriori} distribution $p\left(\bh|\by\right)$ within the framework of information geometry theory.
As a result, we are able to resolve the statistical inference problem from an intrinsic and general standpoint.
This might be a significant factor in the improved convergence behavior of EIGA for massive MIMO-OFDM channel estimation.

\section{Conclusion}

In this paper, we have proposed the EIGA for channel estimation in massive MIMO-OFDM systems.
The original IGA is first revisited.
By using the constant magnitude property of the measurement matrix entries, we reveal that the FONPS of $\braces{p_n}_{n=1}^N$ on the AMs are asymptotically equal at the fixed point of IGA, and the SONPs of $\braces{p_n}_{n=1}^N$ on the AMs are equal to each other at each iteration.
Based on these results, we simplify its iteration by using the common NP to replace the original NPs of $\braces{p_n}_{n=1}^N$ on the AMs and propose the EIGA.
In EIGA, the common NP is the only parameter involved for the iteration.
A FFT-based fast implementation of EIGA is then provided.
Next, we present the convergence analysis for EIGA, where we discuss the ranges of damping that can guarantee the convergence of EIGA in general case and massive MIMO-OFDM channel estimation.
Compared to the general case, the range of damping in channel estimation is considerably wider.
Furthermore, we show that at its fixed point, the \textsl{a posteriori} mean obtained by EIGA is asymptotically optimal.
Simulation results verify that the proposed EIGA can obtain near optimal channel estimation performance with significantly reduced computational complexity compared with the existing algorithms.

\bibliographystyle{IEEEtran}  
\bibliography{IEEEabrv,reference} 




%
%
%

\appendices
\section{Proof of Theorem \ref{the:same nu}}\label{proof:same nu}
We use induction. With the same initialization, $\bbnu_{n}\left(0\right) = \bbnu_{n'}\left(0\right), n, n'\in \setposi{N}$. Assume that at iteration $t$, we have $\bbnu_n\left(t\right) = \bbnu_{n'}\left(t\right)$. From \eqref{equ:nu_0n}, 
\begin{subequations}
	\begin{equation}\label{equ:nu_0n in T1}
		\begin{split}
			\bbnu_{0n}\left(t\right) 
			\!\equaa\! \diag{\!\bD^{-1}-\left[ \bLambda_n\left(t\right) - \frac{1 }{\beta_n\left(t\right)}\bLambda_n^2\left(t\right) \right]^{-1}\!},
		\end{split}
	\end{equation}
	\begin{equation}\label{equ:Lambda_n in T1}
		\bLambda_n\left(t\right) = \left(\bD^{-1} - \Diag{\bbnu_n\left(t\right)}\right)^{-1},
	\end{equation}
	\begin{equation}\label{equ:beta_n in T1}
		\beta_n\left(t\right) = \sigma_z^2 + \bgamma_n^H\bLambda_n\left(t\right)\bgamma_n \equab \sigma_z^2 + \tr{\bLambda_n\left(t\right)}, 
	\end{equation}
\end{subequations}
where $\left(\textrm{a}\right)$ and $\left(\textrm{b}\right)$ come from that the magnitudes of the elements in $\bA$ are $1$.
$\bLambda_n\left(t\right) = \bLambda_{n'}\left(t\right), n, n'\in \setposi{N}$, can be immediately obtained since $\bbnu_n\left(t\right) = \bbnu_{n'}\left(t\right)$. 
Then, $\beta_n\left(t\right) = \beta_{n'}\left(t\right)$, can be obtained. Hence, we have $\bbnu_{0n}\left(t\right) = \bbnu_{0n'}\left(t\right)$. From \eqref{equ:update of NP in IGA}, $\bbnu_{n}\left(t+1\right), n\in\setposi{N}$, is calculated as 
\begin{equation*}\label{equ:update of nu_n in T1}
\bbnu_{n}\left(t+1\right) \!=\! d\sum\nolimits_{n'\neq n}\left( \bbnu_{0n'}\left(t\right) - \bbnu_{n'}\left(t\right) \right) + \left(1-d\right)\bbnu_{n}\left(t\right).
\end{equation*}
Since $\bbnu_{0n}\left(t\right) = \bbnu_{0n'}\left(t\right)$, $\bbnu_n\left(t\right) = \bbnu_{n'}\left(t\right), n, n'\in \setposi{N}$, we have $\bbnu_{n}\left(t+1\right) = \bbnu_{n'}\left(t+1\right), n, n'\in \setposi{N}$.

Assume that we have $\bbnu_0\left(t\right), \bbnu_n\left(t\right)<0, n\in\setposi{N}$, at the iteration $t$. Then, from \eqref{equ:nu_0n in T1}, we can obtain
\begin{align}\label{equ:nu_0n - nu_n}
	& \bbnu_{0n}\left(t\right)-\bbnu_{n}\left(t\right) \notag \\
	=& \diag{\bD^{-1}-\left[ \bLambda_n\left(t\right) - \frac{1 }{\beta_n\left(t\right)}\bLambda_n^2\left(t\right) \right]^{-1}} - \bbnu_n\left(t\right) \notag \\
	\equaa& \diag{\bLambda_n^{-1}\left(t\right) - \left[ \bLambda_n\left(t\right) - \frac{1 }{\beta_n\left(t\right)}\bLambda_n^2\left(t\right) \right]^{-1}}  \\
	=& \diag{ -\frac{1}{\beta_n\left(t\right)} \left( \bI - \frac{1}{\beta_n\left(t\right)}\bLambda_n\left(t\right)  \right)^{-1} }, n\in\setposi{N}, \notag
\end{align}
where $\left(\textrm{a}\right)$ comes from \eqref{equ:Lambda_n in T1}. 
From \eqref{equ:Lambda_n in T1}, we can obtain $\diag{\bLambda_n\left(t\right)} > 0$ since $\bD$ is positive definite diagonal and $\bbnu_n\left(t\right)<0$.  From \eqref{equ:beta_n in T1}, we then have $\beta_n\left(t\right) > 0$ and $\diag{\bLambda_n\left(t\right)} < \beta_n\left(t\right)$
since $\sigma_z^2 > 0$. Thus, $\diag{ \bI -  \frac{1}{\beta_n\left(t\right)}\bLambda_n\left(t\right) } > 0$ can be obtained. At last, we have $\bbnu_{0n}\left(t\right)-\bbnu_{n}\left(t\right) < 0, n\in\setposi{N}$. From \eqref{equ:update of NP in IGA}, $\bbnu_{n}, n\in\setposi{N}$, and $\bbnu_{0}$ are updated as described below \eqref{equ:beta_n in T1} and 
\begin{equation*}
	\bbnu_{0}\left(t+1\right) = d\sum_{n=1}^N\left( \bbnu_{0n}\left(t\right) - \bbnu_{n}\left(t\right) \right) + \left(1-d\right)\bbnu_{0}\left(t\right),
\end{equation*}
respectively. Combining $\bbnu_{0n}\left(t\right)-\bbnu_{n}\left(t\right) < 0, n\in\setposi{N}$, we have that $\bbnu_n\left(t+1\right)<0, n\in\setnnega{N}$.
From a similar process, it is not difficult to obtain that when $\bbnu_{0}\left(0\right), \bbnu_{n}\left(0\right) \le 0$, we have $\bbnu_{0}\left(1\right), \bbnu_{n}\left(1\right) < 0$.
This completes the proof.

\section{Proof of Theorem \ref{the:same theta_n}}\label{{proof:same theta_n}}
We first express $\btheta_{n}^{\star}, n\in\setposi{N}$, as
\begin{align}\label{equ:theta_n^star expression}
	&\btheta_{n}^{\star}\equaa 2\left( \bD^{-1}-\Diag{\bbnu_n^{\star}} + \frac{1}{\sigma_z^2}\bgamma_n\bgamma_n^H \right)\bmu_n^{\star} - \frac{2y_n}{\sigma_z^2}\bgamma_n \notag\\
	\equab& 2\left( \bD^{-1}- \frac{N-1}{N}\Diag{\bbnu_0^{\star}} + \frac{1}{\sigma_z^2}\bgamma_n\bgamma_n^H \right)\bmu_0^{\star} - \frac{2y_n}{\sigma_z^2}\bgamma_n \notag\\
	= &2 \left[ \! \frac{N-1}{N}\left( \bD^{-1} - \Diag{\bbnu_{0}^\star} \right) \!+\! \frac{1}{N}\bD^{-1} + \frac{1}{\sigma_z^2}\bgamma_n\bgamma_n^H \! \right] \bmu_{0}^\star  \notag \\
	&-   \frac{2y_n}{\sigma_z^2}\bgamma_n             \notag \\
	\equac&\frac{N-1}{N}\btheta_{0}^{\star} + 2\left( \frac{1}{N}\bD^{-1} + \frac{1}{\sigma_z^2}\bgamma_n\bgamma_n^H \right)\bmu_0^{\star} - \frac{2y_n}{\sigma_z^2}\bgamma_n, 
\end{align}
where $\left(\textrm{a}\right)$ comes from \eqref{equ:mu_n and Sigma_n} and Sherman-Morrison formula,
$\left(\textrm{b}\right)$ comes from the two conditions in \eqref{equ:m and e conditions} and $\left(\textrm{c}\right)$ comes from \eqref{equ:mu_0 and Sigma_0}. Combining the expression of $\bSigma_{n}$ in \eqref{equ:Sigma_n} and Sherman-Morrison formula, we can obtain
\begin{equation}
	\bSigma_n^{-1}\left(\bvartheta_{n}\right) = \bD^{-1}-\Diag{\bbnu_{n}} + \frac{1}{\sigma_z^2}\bgamma_n\bgamma_n^H, n\in\setposi{N}.
\end{equation}
From \eqref{equ:mu_n}, $\left(\textrm{a}\right)$ can be obtained. Then, from \eqref{equ:theta_n^star expression}, we have
\begin{align}\label{equ:inequality 1}
	&\frac{1}{NM}\sum_{n=1}^{N}\norm{\btheta_{n}^{\star} - \frac{N-1}{N}\btheta_{0}^{\star}}_{\bD}^2 \nonumber\\ 
	\overset{\left(\textrm{a}\right)}{=}& \frac{4}{NM}\sum_{n=1}^{N}\norm{ \frac{1}{N}\bD^{-1}\bmu_0^{\star} + \frac{\bgamma_n^H\bmu_0^{\star} - y_n}{\sigma_z^2}\bgamma_n  }_{\bD}^2 \nonumber\\ 
	\overset{\left(\textrm{b}\right)}{\le}& \frac{8}{NM}\sum_{n=1}^{N}\left( \norm{\frac{1}{N}\bD^{-1}\bmu_0^{\star}}_{\bD}^2 + \norm{ \frac{\bgamma_n^H\bmu_0^{\star} - y_n}{\sigma_z^2}\bgamma_n }_{\bD}^2  \right) \\
	\equac& \frac{8}{N^2M}\norm{\bD^{-1}\bmu_0^{\star}}_{\bD}^2 + \frac{8}{NM}\sum_{n=1}^{N}\left| \frac{\bgamma_n^H\bmu_0^{\star}-y_n}{\sigma_z^2} \right|^2\norm{\bgamma_n}_{\bD}^2 \nonumber\\
	\equad& \frac{8}{N^2M}\norm{\bD^{-1}\bmu_0^{\star}}_{\bD}^2 + \frac{8\tr{\bD}}{NM\sigma_z^4}\norm{\bA\bmu_0^{\star} - \by}^2,\nonumber
\end{align}
where $\left(\textrm{a}\right)$ and $\left(\textrm{c}\right)$ come from the homogeneity of the norm\cite[Definition 5.1.1]{horn2012matrix},
$\left(\textrm{b}\right)$ comes from 
\begin{equation*}
	\norm{\ba + \bb}_{\bD}^2 \le  2\left( \norm{\ba}_{\bD}^2 + \norm{\bb}_{\bD}^2 \right),
\end{equation*}
and $\left(\textrm{d}\right)$ comes from that $\bA$ is of constant magnitude entries and \eqref{equ:gamma_n}.
Define 
\begin{equation*}
	\bR_{yy} \triangleq \Exp\braces{\by\by^H}= \bA\bD\bA^H + \sigma_z^2\bI \in \bbC^{N\times N},
\end{equation*}
where $\by$ is defined in \eqref{equ:maMIMO rece signal 4}.
Then, $\bR_{yy}$ is positive definite. 
From the push-through identity, we have $\tilde{\bmu} = \bD\bA^H\bR_{yy}^{-1}\by$, where $\tilde{\bmu}$ is given by \eqref{equ:post mean}.
Meanwhile, it is shown that at the fixed point of IGA, $\bmu_0^{\star}$ is equal to the \textsl{a posteriori} mean $\tilde{\bmu}$ \cite[Theorem 2]{IGA}. 
Substituting $\bmu_{0}^{\star} = \tilde{\bmu} = \bD\bA^H\bR_{yy}^{-1}\by$ into the last equation of \eqref{equ:inequality 1}, we can obtain 
\begin{align}\label{equ:inequlity1}
	&\Exp\braces{\norm{\bD^{-1}\bmu_0^{\star}}_{\bD}^2} = \Expect{}{\norm{\bA^H\bR_{yy}^{-1}\by}_{\bD}^2} \nonumber \\
	=&\Exp\braces{\tr{\bR_{yy}^{-1}\bA\bD\bA^H\bR_{yy}^{-1}\by\by^H}} = \tr{\bR_{yy}^{-1}\bA\bD\bA^H} \nonumber \\
	\equaa&\tr{\bI-\sigma_z^2\bR_{yy}^{-1}} \overset{\left(\mathrm{b}\right)}{\le}N,
\end{align}
where $\left(\mathrm{a}\right)$ comes from the definition of $\bR_{yy}$ and $\left(\mathrm{b}\right)$ comes from $\bR_{yy} \succeq \sigma_z^2\bI $. 
Also,
\begin{align}\label{equ:inequlity2}
	&\Exp\braces{\norm{\bA\bmu_{0}^{\star} - \by}^2} = \sigma_z^4\Exp\braces{\norm{\bR_{yy}^{-1}\by}^2}\nonumber \\
	= &\sigma_z^4\tr{\bR_{yy}^{-1}} 
	\le \sigma_z^2\tr{\bI} = \sigma_z^2N.
\end{align}
Substituting \eqref{equ:inequlity1} and \eqref{equ:inequlity2} into \eqref{equ:inequality 1}, we can obtain
\begin{equation*}
	0\le \frac{1}{NM}\sum_{n=1}^{N}\Exp\braces{\norm{\btheta_{n}^{\star} - \frac{N-1}{N}\btheta_{0}^{\star}}_{\bD}^2}\le \frac{8}{NM} + \frac{8\tr{\bD}}{\sigma_z^2M}.
\end{equation*}
Since $\tr{\bD}$ and $\sigma_z^2$ are bounded, \eqref{equ:FONPs in IGA} can be obtained.
This completes the proof.

\section{Calculation of \texorpdfstring{$\bvartheta\left(t+1\right)$}{}} \label{proof:calculation of tilde_bvartheta_s}
Define 
\begin{equation*}
	\bvartheta_s\left(t\right) \triangleq \sum_{n=1}^N\bvartheta_{0n}\left(t\right) = \funf[\btheta_s\left(t\right),\bbnu_s\left(t\right)].
\end{equation*}
From \eqref{equ:tilde vartheta_0n}, we have \eqref{equ:calculation of tilde_betheta_s} and
\begin{figure*}
	\begin{align}\label{equ:calculation of tilde_betheta_s}
		&{\btheta}_s\left(t\right) = \sum_{n=1}^{N}{\btheta}_{0n}\left(t\right) 
		\equaa \bJ\left( \frac{2}{\beta\left(\bbnu\left(t\right)\right)}\sum_{n=1}^{N}\bgamma_n y_n \!-\! \frac{1}{\beta\left(\bbnu\left(t\right)\right)}\sum_{n=1}^{N}\bgamma_n\bgamma_n^H\bLambda\left(\bbnu\left(t\right)\right){\btheta}\left(t\right) + N{\btheta}\left(t\right) \right)\nonumber\\
		\equab &\bJ\left( \frac{2}{\beta\left(\bbnu\left(t\right)\right)}\bA^H\by - \frac{1}{\beta\left(\bbnu\left(t\right)\right)}\bA^H\bA\bLambda\left(\bbnu\left(t\right)\right){\btheta}\left(t\right) + N {\btheta}\left(t\right)  \right)
		= \bJ\left( \frac{1}{\beta\left(\bbnu\left(t\right)\right)}\bA^H\left(2\by - \bA\bLambda\left(\bbnu\left(t\right)\right){\btheta}\left(t\right) \right) + N{\btheta}\left(t\right) \right),
	\end{align}
	\hrule
\end{figure*}
\begin{equation}\label{equ:calculation of nu_s}
	\!\!\!{\bbnu}_s \!=\! \sum_{n=1}^{N}{\bbnu}\left(t\right) 
	\!\equac\! N\mathrm{diag}\!\left\lbrace\! \bD^{\!-1}\!\!\!-\!\left(\! \bLambda\!\left(t\right) \!-\! \frac{1}{\beta\left(\bbnu\left(t\right)\right)}\bLambda^2\!\left(t\right)\! \right)^{\!-1}\!\right\rbrace,
\end{equation}
where
\begin{equation*}
\bJ = \left( \bI - \frac{1}{\beta\left(\bbnu\left(t\right)\right)}\bLambda\left(\bbnu\left(t\right)\right) \right)^{-1}
\end{equation*}
$\bLambda\left(\bbnu\left(t\right)\right)$ and $\beta\left(\bbnu\left(t\right)\right)$ are given by \eqref{equ:Lambda} and \eqref{equ:beta}, respectively,
$\left(\textrm{a}\right)$ and $\left(\textrm{c}\right)$ come from \eqref{equ:tilde vartheta_0n}, and $\left(\textrm{b}\right)$ comes from \eqref{equ:gamma_n} and $\bA^H = \left[ \bgamma_1, \ \bgamma_2, \ \cdots, \ \bgamma_N \right]$. 
Then, from the update way of $\bvartheta$ in the last equation of \eqref{equ:update of overlinevartheta in RIGA}, we have
\begin{subequations}\label{equ:update of vartheta}
	\begin{equation}\label{equ:update of theta}
		\btheta\left(t+1\right) = \frac{d\left(N-1\right)}{N}\btheta_s\left(t\right) + \left(1-dN\right)\btheta\left(t\right),
	\end{equation}
	\begin{equation}\label{equ:update of nu}
		\bbnu\left(t+1\right) = \frac{d\left(N-1\right)}{N}\bbnu_s\left(t\right) + \left(1-dN\right)\bbnu\left(t\right).
	\end{equation}
\end{subequations}
Substituting \eqref{equ:calculation of tilde_betheta_s} and \eqref{equ:calculation of nu_s} into \eqref{equ:update of vartheta}, we can obtain \eqref{equ:tilde vartheta_s}.
\begin{figure*}
	\begin{subequations}\label{equ:tilde vartheta_s}
		\begin{equation}\label{equ:tilde theta_s}
			\btheta\left(t+1\right) = 
			\frac{d\left(N-1\right)}{N}\left( \bI - \frac{1}{\beta\left(\bbnu\left(t\right)\right)}\bLambda\left(\bbnu\left(t\right)\right) \right)^{-1}\left[ \frac{1}{\beta\left(\bbnu\left(t\right)\right)}\bA^H\left(2\by - \bA\bLambda\left(\bbnu\left(t\right)\right){\btheta}\left(t\right) \right) + N{\btheta}\left(t\right) \right] + \left(1-dN\right)\btheta\left(t\right)
		\end{equation}
		\begin{equation}\label{equ:tilde nu_s}
			\bbnu\left(t+1\right) = {d\left(N-1\right)}\diag{\bD^{-1}-\left( \bLambda\left(\bbnu\left(t\right)\right) - \frac{1}{\beta\left(\bbnu\left(t\right)\right)}\bLambda^2\left(\bbnu\left(t\right)\right)\right)^{-1}} + \left(1-dN\right)\bbnu\left(t\right)
		\end{equation}
	\end{subequations}
	\hrule
\end{figure*}

We now show that $\bbnu\left(t+1\right)$ in \eqref{equ:tilde nu_s} can be re-expressed as that in \eqref{equ:update of nu_0}.
From \eqref{equ:tilde nu_s},  \eqref{equ:aux0 in AppA} on the next page  is direct.
\begin{figure*}
	\begin{equation}\label{equ:aux0 in AppA}
		\bbnu\left(t+1\right) = \underbrace{d\left(N-1\right)\left( \diag{\bD^{-1}-\left( \bLambda\left(\bbnu\left(t\right)\right) - \frac{1}{\beta\left(\bbnu\left(t\right)\right)}\bLambda^2\left(\bbnu\left(t\right)\right)\right)^{-1}} - \bbnu\left(t\right) \right)}_{d\bg\left(\bbnu\left(t\right)\right)}  + \left(1-d\right)\bbnu\left(t\right)
	\end{equation}
	\hrule
\end{figure*}
Thus,  $\bg\left(\bbnu\left(t\right)\right)$ can be expressed as \eqref{equ:aux1 in AppA} on the next page, where $\left(\textrm{a}\right)$ comes from that \eqref{equ:Lambda}.
\begin{figure*}
	\begin{equation}\label{equ:aux1 in AppA}
		\begin{split}
			&\bg\left(\bbnu\left(t\right)\right) = \left(N-1\right) \diag{ \bD^{-1} - \left( \bLambda\left(\bbnu\left(t\right)\right) - \frac{1}{\beta\left(\bbnu\left(t\right)\right)}\bLambda^2\left(\bbnu\left(t\right)\right) \right)^{-1}} - \left(N-1\right)\bbnu\left(t\right)\\
			=&\left(N-1\right)\diag{ \bD^{-1} - \Diag{\bbnu\left(t\right)} - \left( \bLambda\left(\bbnu\left(t\right)\right) - \frac{1}{\beta\left(\bbnu\left(t\right)\right)}\bLambda^2\left(\bbnu\left(t\right)\right) \right)^{-1}}\\
			\equaa&\left(N-1\right)	\diag{ \bLambda^{-1}\left(\bbnu\left(t\right)\right) - \bLambda^{-1}\left(\bbnu\left(t\right)\right)\left( \bI - \frac{1}{\beta\left(\bbnu\left(t\right)\right)}\bLambda\left(\bbnu\left(t\right)\right) \right)^{-1} } = -\left(N-1\right)\diag{\left(\beta\left(\bbnu\left(t\right)\right)\bI - \bLambda\left(\bbnu\left(t\right)\right)\right)^{-1}}
		\end{split}
	\end{equation}
	\hrule
\end{figure*}
We then show that when $t = 0$, the matrices that need to be inverted in \eqref{equ:aux1 in AppA} are intertible.
From \eqref{equ:Lambda} and $\bbnu\left(0\right) \le \bzero$, we can obtain that $\bLambda\left(\bbnu\left(0\right)\right)$ is positive definite and hence invertible.
From \eqref{equ:beta}, we have
\begin{equation*}
	\beta\left(\bbnu\left(0\right)\right) > \left[ \bLambda\left(\bbnu\left(0\right)\right) \right]_{i,i} > 0, i\in \setposi{M}.
\end{equation*}
This implies that
\begin{equation*}
	\beta\left(\bbnu\left(0\right)\right)\bI - \bLambda\left(\bbnu\left(0\right) \right)
\end{equation*}
is positive definite and hence invertible.
Moreover, combining \eqref{equ:aux0 in AppA} and \eqref{equ:aux1 in AppA}, we have $\bg\left(\bbnu\left(0\right)\right) < \bzero$, and 
\begin{equation*}
	\bbnu\left(1\right) = d\bg\left(\bbnu\left(0\right)\right) + \left(1-d\right)\bbnu\left(0\right) < \bzero
\end{equation*}
is finite.
Following by that, assuming that at the $t$-th iteration, where $t \ge 1$, we have
$\bbnu\left(t\right) < \bzero$ is finite, $\bLambda\left(\bbnu\left(t\right)\right)$ and
\begin{equation*}
	\beta\left(\bbnu\left(t\right)\right)\bI - \bLambda\left(\bbnu\left(t\right) \right)
\end{equation*} 
are positive definite and invertible.
In the same way, it can be readily checked that $\bbnu\left(t+1\right) < \bzero$ is finite.
Hence, we have $\bLambda\left(\bbnu\left(t+1\right)\right)$ and
\begin{equation*}
	\beta\left(\bbnu\left(t+1\right)\right)\bI - \bLambda\left(\bbnu\left(t\right) \right)
\end{equation*}
are positive definite and invertible.
By induction, we have shown that when $t \ge 1$, we have $\bbnu\left(t\right) < \bzero$ is finite,
and for $t \ge 0$, $\bLambda\left(\bbnu\left(t\right)\right)$ and 
\begin{equation*}
	\beta\left(\bbnu\left(t\right)\right)\bI - \bLambda\left(\bbnu\left(t\right) \right)
\end{equation*}
are positive definite and invertible.

We now show that $\btheta\left(t+1\right)$ in \eqref{equ:tilde theta_s} can be re-expressed as that in \eqref{equ:update of theta_0}.
From \eqref{equ:tilde theta_s}, we can obtain  \eqref{equ:aux2 in AppA}, where we omit some of the counter $t$ at the right hand side of the equation for the notational convenience, and
\begin{equation}\label{equ:T}
	\bT\left(\bbnu\right) = \left( \bI - \frac{1}{\beta\left(\bbnu\right)} \bLambda\left(\bbnu\right) \right)^{-1},
\end{equation} 
where the matrix invertibility comes from \eqref{equ:Lambda} and \eqref{equ:beta} directly. 
\begin{figure*}
	\begin{equation}\label{equ:aux2 in AppA}
		\begin{split}
			&\btheta\left(t+1\right) = 
			\frac{d\left(N-1\right)}{N}\bT\left(\bbnu\right)\left[ \frac{1}{\beta\left(\bbnu\right)}\bA^H\left(2\by\right) - \frac{1}{\beta\left(\bbnu\right)}\bA^H \bA\bLambda\left(\bbnu\right){\btheta}\left(t\right)  + N\btheta\left(t\right) \right] + \left(1-dN\right)\btheta\left(t\right)\\
			=&\underbrace{\frac{2d\left(N-1\right)}{\beta\left(\bbnu\right) N}\bT\left(\bbnu\right)\bA^H\by}_{\bb\left(\bbnu\right)} + \underbrace{\left[ \frac{d\left(N-1\right)}{N}\bT\left(\bbnu\right)\left(N\bI - \frac{1}{\beta\left(\bbnu\right)}\bA^H\bA\bLambda\left(\bbnu\right)  \right)\btheta\left(t\right) - d\left(N-1\right)\btheta\left(t\right) \right]}_{d\bB\left(\bbnu\right) \btheta\left(t\right) }  + \left(1-d\right)\btheta\left(t\right)
		\end{split}
	\end{equation}
	\hrule
\end{figure*}
Thus, we can obtain 
\begin{align}
	&\bB\left(\bbnu\right)  \\
	=&  \frac{\left(N-1\right)}{N}\bT\left(\bbnu\right)\left(N\bI - \frac{1}{\beta\left(\bbnu\right)}\bA^H\bA\bLambda\left(\bbnu\right)  \right) - \left(N-1\right)\bI \nonumber \\
	=& \left(N-1\right) \left[\frac{\bT\left(\bbnu\right)}{N}\left( N\bI - \frac{\bA^H\bA\bLambda\left(\bbnu\right)}{\beta\left(\bbnu\right)} \right)  - \bT\left(\bbnu\right)\bT^{-1}\left(\bbnu\right) \right]                     \nonumber \\
	=&  \frac{\left(N-1\right)}{\beta\left(\bbnu\right)} \bT\left(\bbnu\right)\left( \bI - \frac{1}{N}\bA^H\bA \right)\bLambda\left(\bbnu\right).               \nonumber
\end{align}
Also, it is not difficult to show that given a finite $\btheta\left(0\right)$, $\btheta\left(t\right)$ is finite at each iteration.

\section{Proof of Lemma \ref{lemma:function tilde{g}}}\label{proof:function tilde{g}}
From Appendix \ref{proof:calculation of tilde_bvartheta_s}, it can be checked that given  $\bbnu \le  \bzero$, $\tilde{\bg}\left(\bbnu\right)$ and $\bg\left(\bbnu\right)$ are well defined.
Denote $g_i\left(\bbnu\right)$, $\nu_{i}$, $d_i$ and ${\lambda}_i\left(\bbnu\right)$ as the $i$-th components of $\bg\left(\bbnu\right)$, $\bbnu$, the diagonals of  $\bD$ and $\bLambda\left(\bbnu\right)$, respectively, where $i\in \setposi{M}$.
Due to $\bbnu \le \bzero$, we have
\begin{subequations}
	\begin{equation}
		\beta\left(\bbnu\right) = \tilde{\sigma}_z^2 + \sum_{i=1}^{M}\lambda_{i}\left(\bbnu\right) {>}0,
	\end{equation}
	\begin{equation}\label{equ:lambda i}
		{\lambda}_{i}\left(\bbnu\right) = \frac{1}{d_i^{-1} - \nu_{i}}{>}0,
	\end{equation}
	\begin{equation}\label{g_i}
		\begin{split}
			g_{i}(\bbnu) &=- \frac{N-1}{\beta(\bbnu) - {\lambda}_i(\bbnu)} \\
			&= -\frac{N-1}{\tilde{\sigma}_z^2 + \sum_{i'\neq i}{\lambda}_{i'}\left(\bbnu\right)} < 0.	
		\end{split}
	\end{equation}
\end{subequations}
From \eqref{g_i} and \eqref{equ:update of nu_0}, the two properties of $\bg\left(\bbnu\right)$ and $\tilde{\bg\left(\bbnu\right)}$, i.e., the monotonicity and the scalability, are not difficult to see. 
We next show its boundedness. 

From the definitions, we have
\begin{subequations}
	\begin{equation}
		\lim_{\nu_{1},\nu_{ 2},\ldots,\nu_{ M}\rightarrow-\infty}\beta\left(\bbnu\right) = \tilde{\sigma}_z^2.
	\end{equation}
\end{subequations} 
From the monotonicity of $\bg\left(\bbnu\right)$, we can obtain
\begin{equation}
	\bg\left(\bbnu\right)>	\lim_{\nu_{1},\nu_{ 2},\ldots,\nu_{ M}\rightarrow-\infty}\bg\left(\bbnu\right) =  -\frac{N-1}{\tilde{\sigma}_z^2}\mathbf{1} = \tilde{\bg}_{\textrm{min}} .
\end{equation}
Thus,  $\tilde{\bg}_{\textrm{min}}<\bg\left(\bbnu\right)<\bzero$.
Then,  $\tilde{\bg}_{\textrm{min}}<\tilde{\bg}\left(\bbnu\right) < \bzero$ directly follows from the definition of $\tilde{\bg}\left(\bbnu\left(t\right)\right)$ in \eqref{equ:update of nu_0}. 
This completes the proof.


\section{Proof of Theorem \ref{The-theta_1 conv}}\label{proof:theta_1 conv}

Consider $\bbnu\left(1\right) = \tilde{\bg}\left(\bbnu\left(0\right)\right)$.
If $\bbnu\left(1\right) \le \bbnu\left(0\right)$, by Lemma \ref{lemma:function tilde{g}},
\begin{equation}
	\bbnu\left(2\right) = \tilde{\bg}\left(\bbnu\left(1\right)\right)\le \tilde{\bg}\left(\bbnu\left(0\right)\right) = \bbnu\left(1\right).
\end{equation}
Then, the sequence $\bbnu\left(t\right)$ is a decreasing sequence.
By Lemma \ref{lemma:function tilde{g}}, this sequence is also bounded.
Thus, it converges to a finite vector $\bbnu^\star$.
Also, by Lemma \ref{lemma:function tilde{g}}, we have the result of Theorem \ref{The-theta_1 conv}.
The case of $\bbnu\left(1\right) \ge \bbnu\left(0\right)$ can be similarly proved.
This completes the proof.

\section{ Calculation of \eqref{equ:relation of bbnu} }\label{proof:Calculation 2}
From \eqref{equ:update of nu_0} and ${\bbnu^\star} = \tilde{\bg}\left(\bbnu^\star\right)$, we have $\bbnu^\star = \bg\left(\bbnu^\star\right)$.
Substituting $\bbnu^\star = \bg\left(\bbnu^\star\right)$ and $\bbnu\left(t\right) = \bbnu^{\star}$ into the first equation of \eqref{equ:aux1 in AppA} in Appendix \ref{proof:calculation of tilde_bvartheta_s}, we can obtain \eqref{equ:relation of bbnu}.

\section{Proof of Lemma \ref{The-synchronous updates}}\label{Proof-The-syn}
Define $\btheta^\star$ as
\begin{equation}
	\btheta^\star \triangleq \left(\bI - \tilde{\bB}^\star\right)^{-1}\bb^\star.
\end{equation}
Since $\rho\left(\tilde{\bB}^\star\right)<1$, $1$ is not an eigenvalue of $\tilde{\bB}^\star$ and $\bI-\tilde{\bB}^\star$ is invertible. Thus,  the above $\btheta^\star$ exists

We next show that $\btheta\left(t\right)$ converges to $\btheta^\star$. 
Since $\rho\left(\tilde{\bB}^\star\right) <1$, there exists a matrix norm $\norm{\cdot}$ such that \cite[Lemma 5.6.10]{horn2012matrix}
\begin{equation}
	\norm{\tilde{\bB}^\star} < 1.
\end{equation}
Then, let $\norm{\cdot}$ be the vector norm that induces the matrix norm $\norm{\cdot}$ \cite[Definition 5.6.1]{horn2012matrix}.
Define the error between $\btheta\left(t\right)$ and $\btheta^{\star}$ as
\begin{equation}
	\varepsilon\left(t\right) \triangleq \lVert \btheta\left(t\right) - \btheta^{\star} \rVert.
\end{equation}
Then, we can obtain \eqref{ineq0}  on the next page.
\begin{figure*}
	\begin{equation}\label{ineq0}
		\begin{split}
			&\varepsilon\left(t+1\right) = \lVert \btheta\left(t+1\right) - \btheta^{\star}\rVert
			= \lVert \tilde{\bB}\left( \bbnu\left(t\right) \right)\btheta\left(t\right) - \tilde{\bB}^{\star}\btheta^{\star} + \bb\left(\bbnu\left(t\right)\right)-\bb^{\star} \rVert\\
			&=\lVert \tilde{\bB}\left(\bbnu\left(t\right)\right)\left( \btheta\left(t\right)-\btheta^{\star}\right) + \left( \tilde{\bB}\left(\bbnu\left(t\right)\right)-\tilde{\bB}^{\star} \right)\btheta^{\star}  + \bb\left(\bbnu\left(t\right)\right)-\bb^{\star}  \rVert \\
			&{\le}\lVert \tilde{\bB}\left(\bbnu\left(t\right)\right)\left( \btheta\left(t\right)-\btheta^{\star}\right) \rVert  + \lVert \left( \tilde{\bB}\left(\bbnu\left(t\right)\right)-\tilde{\bB}^{\star} \right)\btheta^{\star} + \bb\left(\bbnu\left(t\right)\right)-\bb^{\star} \rVert\\
			&{\le }\norm{\tilde{\bB}\left(\bbnu\left(t\right)\right)}\varepsilon\left(t\right) + \lVert \left( \tilde{\bB}\left(\bbnu\left(t\right)\right)-\tilde{\bB}^{\star} \right)\btheta^{\star} + \bb\left(\bbnu\left(t\right)\right)-\bb^{\star} \rVert
		\end{split}
	\end{equation}
	\hrule
\end{figure*}
Define a sequence $c\left(t\right)$ as
\begin{equation}\label{definition-ct}
	c\left(t\right)\triangleq\lVert \left( \tilde{\bB}\left(\bbnu\left(t\right)\right)-\tilde{\bB}^{\star} \right)\btheta_0^{\star} + \bb\left(\bbnu\left(t\right)\right)-\bb^{\star} \rVert.
\end{equation}
Since $\bbnu$ converges to $\bbnu^\star$, we have
\begin{equation*}
	\liminfty{t}\tilde{\bB}\left(\bbnu\left(t\right)\right) = \tilde{\bB}^{\star}, \ \liminfty{t}\bb\left(\bbnu\left(t\right)\right) = \bb^{\star},
\end{equation*}
and thus,
\begin{equation}\label{equ:aux1 in App2}
	\liminfty{t}c\left(t\right)= 0, 
\end{equation}
\begin{equation}\label{equ:aux2 in App2}
	\liminfty{t}\norm{\tilde{\bB}\left(\bbnu\left(t\right)\right)} = \norm{\tilde{\bB}^\star} <1.
\end{equation}
Let 
\begin{equation}
	\delta_1 \triangleq \frac{1 - \norm{\tilde{\bB}^\star}}{2} >0,
\end{equation}
\begin{equation}
	\delta_2 \triangleq \norm{\tilde{\bB}^\star} + \delta_1 <1.
\end{equation}
To show 
\begin{equation*}
	\liminfty{t}\varepsilon\left(t\right) = 0,
\end{equation*}
we only need to show that $\forall \epsilon > 0$, $\exists \  t_0$, when $t > t_0$, we have
\begin{equation*}
	\varepsilon\left(t\right) < \epsilon.
\end{equation*} 
From \eqref{equ:aux1 in App2} and \eqref{equ:aux2 in App2}, we can obtain that $\exists \ t_1 > 0$, when $t > t_1$, we have
\begin{equation*}
	c\left(t\right) < \frac{\epsilon\left(1-\delta_2\right)}{2},
\end{equation*}
\begin{equation*}
	\norm{\tilde{\bB}\left(\bbnu\left(t\right)\right)} \le \delta_2.
\end{equation*}
Then, for $t \ge t_1 $, we have
\begin{equation*}
	\varepsilon\left(t+1\right) \le \delta_2\varepsilon\left(t\right) + \frac{\epsilon\left(1-\delta_2\right)}{2},
\end{equation*} 
and hence for any positive integer $\Delta t$,
\begin{equation*}
	\begin{split}
		&\varepsilon\left(t+\Delta t\right)\\ < & \delta_2^{\Delta t}\varepsilon\left(t\right) + \left(\delta_2^{\Delta t-1} + \delta_2^{\Delta t-2} + \cdots + \delta_2^{0}\right)  \frac{\epsilon\left(1-\delta_2\right)}{2} \\
		<& \delta_2^{\Delta t}\varepsilon\left(t\right) + \frac{\epsilon}{2}.
	\end{split}
\end{equation*}
Since $0<\delta_2<1$, we have 
\begin{equation*}
	\liminfty{\Delta t}\delta_2^{\Delta t} = 0.
\end{equation*}
Let $\Delta t$ such that 
\begin{equation*}
	\delta_2^{\Delta t} < \frac{\epsilon}{2\varepsilon\left(t_1\right)}.
\end{equation*}
Let $t_0 = t_1 + \Delta t$.
Then, when $t> t_0$, we have
\begin{align*}
	\varepsilon\left(t\right) &= \varepsilon\left(t_0 + t - t_0\right) = \varepsilon\left(t_1 + \Delta t + t - t_0\right) \\
	&< \delta_2^{ \Delta t + t - t_0}\varepsilon\left(t_1\right) + \frac{\epsilon}{2} < \delta_2^{\Delta t}\varepsilon\left(t_1\right) + \frac{\epsilon}{2} \\ \nonumber
	&< \frac{\epsilon}{2\varepsilon\left(t_1\right)}\varepsilon\left(t_1\right) + \frac{\epsilon}{2} = \epsilon.
\end{align*}
This proves 
\begin{equation*}
	\lim\limits_{t\to \infty} \varepsilon\left(t\right) = 0.
\end{equation*}
Since all vector norms are equivalent, it implies that $\norm{\btheta\left(t\right) - \btheta^\star}_2$ with the Euclidean norm also goes to zero as $t \to \infty$. 
This completes the proof of Lemma \ref{The-synchronous updates}.

\section{Proof of Lemma \ref{Lemma-spectral radius}}\label{{Proof-Lemma-spec}}
From \eqref{equ:update of nu_0}, we have
\begin{align}
	&\bbnu^{\star} = \bg\left(\bbnu^\star\right) \nonumber \\
	=&-\left(N-1\right)\diag{\left( \beta^{\star}\bI - \bLambda^{\star}    \right)^{-1}}.
\end{align}
From the definition of $\beta^{\star}$ in \eqref{equ:beta0 star} and $\bLambda^{\star}$ in \eqref{equ:Lambda0 star}, we can readily show that $\beta^{\star}\bI - \bLambda^{\star} $ is invertible.
Since we have proven that $\bbnu^{\star} <\bzero$ in Theorem \ref{The-theta_1 conv},
from the definition of $\bLambda^\star$ in \eqref{equ:Lambda0 star}, we can obtain
\begin{equation}\label{Lambda *}
	\begin{split}
		\bLambda^\star &= \left( \bD^{-1} -\Diag{\bbnu^{\star}} \right)^{-1} \\
		&{\prec} \left( -\Diag{\bbnu^{\star}} \right)^{-1}
		{=}\frac{1}{N-1}\left( \beta^*\bI - \bLambda^* \right).
	\end{split}
\end{equation}
From the definition, we can obtain $\bLambda^\star$ is diagonal positive definite.
Let $\lambda_i^\star = \left[ \bLambda^\star \right]_{i,i}, i\in \setposi{M}$.
Hence, $\lambda_i^\star$ is an eigenvalue of $\bLambda^{\star}$ and  $\lambda_i^\star > 0, i\in \setposi{M}$. 
Then, from \eqref{Lambda *}, 
we have 
\begin{equation}
	\lambda_i^\star - \frac{\beta^\star - \lambda_i^\star}{N-1} < 0, i\in \mathcal{Z}_M^+,
\end{equation}
which implies that $\lambda_i^\star < \frac{\beta^*}{N}, i\in \mathcal{Z}_M^+$. Hence, we have $\rho\left( \bLambda^\star \right) < \frac{\beta^\star}{N}$.  This completes the proof.

\section{Proof of Lemma \ref{lemma:eigens of B star}}\label{proof:eigens of B star}
We first prove that the eigenvalues of $\bB^\star$ are all real.
Let
\begin{equation}
	\begin{split}
		\bQ &\triangleq \left( \bI - \frac{1}{N}\bD^{-1}\bLambda^\star \right)^{1/2}\left( \frac{\bLambda^\star}{\beta^\star} \right)^{1/2}\left( N\bI - \bA^H\bA \right)\\
		& \ \ \times\left( \bI - \frac{1}{N}\bD^{-1}\bLambda^\star \right)^{1/2}\left( \frac{\bLambda^\star}{\beta^\star} \right)^{1/2}\\
		&= \bK^{-1}\bB^\star\bK \sim \bB^\star,
	\end{split}
\end{equation}
where $\bK$ is the following diagonal positive definite matrix: 
\begin{equation}
	\bK = \left( \frac{\bLambda^\star}{\beta^\star}\right)^{-1/2}\left( \bI - \frac{1}{N}\bD^{-1}\bLambda^\star \right)^{1/2}. 
\end{equation}
Thus, $\bB^\star$ and $\bQ$ have the same eigenvalues.
From the definition, $\bQ$ is Hermitian. Therefore, the eigenvalues of $\bQ$ and $\bB^\star$ are all real. 
Then, from  \eqref{equ:Q in Appendix} on the next page,
\begin{figure*}
	\begin{equation} \label{equ:Q in Appendix}
		\bQ = \underbrace{N\left(\bI - \frac{1}{N}\bD^{-1}\bLambda^\star\right)\frac{\bLambda^\star}{\beta^\star}}_{\bQ_1} +  \underbrace{\left( \bI - \frac{1}{N}\bD^{-1}\bLambda^\star \right)^{1/2}\left( \frac{\bLambda^\star}{\beta^\star} \right)^{1/2}\left( - \bA^H\bA \right) \left( \bI - \frac{1}{N}\bD^{-1}\bLambda^\star \right)^{1/2}\left( \frac{\bLambda^\star}{\beta^\star} \right)^{1/2}}_{\bQ_2}
	\end{equation}
	\hrule
\end{figure*}
we have $\bQ_1, \bQ_2$ are Hermitian, and hence \cite[6.70 (a), pp116]{seber2008matrix}
\begin{equation}
	\lambda_{max}\left(\bQ\right) \le \lambda_{max}\left(\bQ_1\right) + \lambda_{max}\left(\bQ_2\right).
\end{equation}
Then, for $\bQ_1$, we can readily check that it is positive definite, and thus
\begin{equation}
	\lambda_{max}\left(\bQ_1\right) = \rho\left(\bQ_1\right) \overset{\left(\mathrm{a}\right)}{\le}N\rho\left(\bI - \frac{1}{N}\bD^{-1}\bLambda^\star\right)\rho\left(\frac{\bLambda^\star}{\beta^\star}\right)
	\overset{\left(\mathrm{b}\right)}{<}1,
\end{equation}
where $\left(\mathrm{a}\right)$ comes from \cite[Exercise below Theorem 5.6.9]{horn2012matrix} and $\left(\mathrm{b}\right)$ comes from \eqref{equ:aux1 in text} and \eqref{equ:aux in lemma5}. 
Define $\bK_1$ as
\begin{equation}
	\bK_1 = \left( \bI - \frac{1}{N}\bD^{-1}\bLambda^\star \right)^{1/2}\left( \frac{\bLambda^\star}{\beta^\star} \right)^{1/2}.
\end{equation}
Then, we can obtain that $-\bQ_2 = \left(\bA\bK_1\right)^H\bA\bK_1$.
Hence, $\bQ_2$ is negative semidefinite,  and we have $\lambda_{max}\left(\bQ_2\right) \le  0$. Thus, we have
\begin{equation}
	\lambda_{max}\left(\bQ\right) < 1.
\end{equation}
Since $\bB^\star \sim \bQ$, we have $\lambda_{max}\left(\bB^\star\right) < 1$.
Then, from \eqref{equ:range of eigen of B}, we have 
\begin{equation}
	\rho\left(\bB^{\star}\right)< 1\times \rho\left(N\bI - \bA^H\bA\right)\times \frac{1}{N} = \frac{\rho\left(N\bI - \bA^H\bA\right)}{N}.
\end{equation}
Thus, we can obtain that 
\begin{equation*}
	-\frac{\rho\left(N\bI - \bA^H\bA\right)}{N} < \lambda_{B,i} < 1.
\end{equation*}
This completes the proof.

\section{Proof of Theorem \ref{lemma:radius of NI-A^HA}}\label{proof:radius of NI-A^HA}
We first give the range of $\rho\left(\bA^H\bA\right)$.
To do so, we begin by giving the detailed expression for $\bA$.
After vectorizing (\ref{equ:maMIMO rece signal 2}), we have
\begin{equation}\label{equ:maMIMO rece signal 3}
	\by = \tilde{\bA}\tilde{\bh} + \bz,
\end{equation}
where $\by, \bz\in\bbC^{N\times 1}$ and  $\tilde{\bh}\in\bbC^{\tilde{M}\times 1}$ are the vectorizations of $\bY$, $\bZ$ and $\bH$, respectively, 
\begin{equation}\label{equ:def of tilde A}
	\tilde{\bA}\triangleq \bM^T\otimes\bV \in\bbC^{N\times \tilde{M}},
\end{equation}
$N =N_rN_p$ and $\tilde{M} = KF_aF_\tau N_rN_f$. 
Define the number of non-zero components in $\bomega \triangleq \mtxvec{\left[ \bOmega_1, \ \bOmega_2, \ \cdots, \ \bOmega_K \right]}$ as $M \triangleq \lVert \bomega \rVert_0$.
Then, $M$ is the actual number of variables to be estimated, i.e., the number of components in $\tilde{\bh}$ with non-zero variance. 
Denote the indexes of non-zero components in $\bomega$ as $\mathcal{P} \triangleq \braces{p_1,p_2,\ldots,p_M}$,
where $1\le p_1 < p_2 <\ldots <p_M \le \tilde{M}$. 
We define an extraction matrix as $\bE  \triangleq \left[ \be_{p_1}, \ \be_{p_2}, \ \cdots, \ \be_{p_M} \right]\in \bbC^{\tilde{M}\times M}$,
where $\be_i\in\bbC^{\tilde{M}\times 1}, i\in\mathcal{P}$ is the $i$-th column of the $\tilde{M}$ dimensional identity matrix.
Then, \eqref{equ:maMIMO rece signal 3} can be rewritten as $\by = \bA\bh + \bz$,
where 
\begin{equation*}\label{equ:A in general}
	\bA = \tilde{\bA}\bE \in\bbC^{N\times M}
\end{equation*}
is the matrix of $\tilde{\bA}$ after column extraction, $\bh = \bE^T\tilde{\bh} \in\bbC^{M\times 1}$ is the vector of $\tilde{\bh}$ after variable extraction and $\bD \triangleq \Diag{\bE^T\bomega}$.
From the definition, $\bA^H\bA$ is positive semidefinite. The eigenvalues $v_1\le v_2\le  \cdots\le  v_M$ of $\bA^H\bA$ are real and nonnegative. 
Thus, we can obtain $v_M = \rho\left(\bA^H\bA\right)$.
Then, we have
\begin{align}
	\rho\left(\bA^H\bA\right) &= v_M  \le \sum_{m=1}^{M}v_m 
	= \tr{\bA\bA^H} \nonumber\\
	&\le \sum_{m=1}^{M}v_M = M\rho\left( \bA^H\bA \right).
\end{align} 
When $\left|a_{i,j}\right|=1, \forall i,j$, we can obtain
\begin{equation}
	\tr{\bA\bA^H} = \lVert \bA\rVert_F^2 = NM.
\end{equation}
Thus, we have $\rho\left(\bA^H\bA\right) \le NM \le M\rho\left(\bA^H\bA\right)$, which implies that
\begin{equation}
	N\le \rho\left( \bA^H\bA \right)\le NM.
\end{equation}
Hence, we have $0\le v_1\le \cdots \le v_M \le NM$.
The eigenvalues of $N\bI - \bA^H\bA$ are $v_m' =N - v_m, m\in \setposi{M}$.
Thus, we have $N-NM \le v_m'\le N$, and $\left|  v_m'\right| \le \max\braces{N,NM-N}$.
Since in this paper $M >1$, we have $\rho\left(N\bI - \bA^H\bA\right) \le NM-N$.
If $\rank{\bA} = 1$, then $\bA$ can be decomposed as $\bA = \ba\bb^H$, where $\ba \in \bbC^{N\times 1}$, $\bb \in \bbC^{M\times 1}$, and $\ba$ and $\bb$ are non-zero.
Combining \cite[6.54 (c)]{seber2008matrix}, $\bb \bb^H$ and $\bb^H \bb$ are positive semi-definite, we can obtain that 
\begin{align}
	\rho\left(\bA^H\bA\right) &= \rho\left( \bb \ba^H\ba \bb^H \right) = \ba^H\ba\rho\left( \bb \bb^H \right) =\ba^H\ba\rho\left( \bb^H \bb \right)\nonumber  \\ 
	&= \tr{\bA^H\bA} = NM.
\end{align}
Then, we have $\rho\left(N\bI - \bA^H\bA\right) = NM -N$.
This completes the proof.

\section{Proof of Theorem \ref{lemma:SR in mMIMO-CE 1}}\label{proof:SR in mMIMO-CE 1}

From the definition of $\bA$ in \eqref{equ:maMIMO rece signal 4}, we have 
\begin{equation*}
	\bA^H\bA = \bE^T\left( \tilde{\bA}^H\tilde{\bA} \right)\bE,
\end{equation*}
which implies that $\bA^H\bA$ is a principal submatrix of $\tilde{\bA}^H\tilde{\bA}$.
Combining \cite[Theorem 4.3.28]{horn2012matrix} and the componentary transformation, we have
\begin{equation}
	\lambda_{max}\left(\bA^H\bA\right) \le\lambda_{max}\left(\tilde{\bA}^H\tilde{\bA}  \right),  
\end{equation}
which implies that $\rho\left(\bA^H\bA\right) \le \rho\left(\tilde{\bA}^H\tilde{\bA}\right)$.
From \cite[6.54 (c), pp 107]{seber2008matrix}, we can obtain $\rho\left(\tilde{\bA}^H\tilde{\bA}\right) = \rho\left(\tilde{\bA}\tilde{\bA}^H\right)$.
From the definition \eqref{equ:def of tilde A}, we have
\begin{equation}
	\begin{split}
		\tilde{\bA}\tilde{\bA}^H &= \left(\bM^T \otimes \bV\right)\left( \bM^T \otimes \bV \right)^H \\
		&= \left( \bM^T\bM^* \right) \otimes \left( \bV_v\otimes \bV_h \right)\left( \bV_v^H\otimes \bV_h^H \right)\\
		&= F_vF_hN_r\bK \otimes \bI,
	\end{split}     
\end{equation}
where $\bK = \sum_{k=1}^{K}\bX_k\bF\bF^H\bX_k^H$. 
Since $\bK$ is Hermitian, we can decompose $\bK$ as $\bK = \bU\bLambda_K\bU^H$, where $\bU$ is unitary. 
Then, we can obtain
\begin{equation}
	\bK \otimes \bI = \left( \bU\otimes \bI \right)\left(\bLambda_K\otimes \bI\right)\left(\bU^H\otimes \bI\right) = \bU'\bLambda_K'\left(\bU'\right)^H,
\end{equation} 
where $\bU'$ is unitary and $\bLambda_K'$ is diagonal. 
Since $\bK \otimes \bI$ is also Hermitian, we can obtain that 
\begin{equation*}
	\rho\left(\bK\otimes \bI\right) = \rho\left(\bK\right).
\end{equation*}
Since $\bX_k$ is unitary, we also have $\rho\left(\bX_k\bF\bF^H\bX_k^H\right) = \rho\left( \bF\bF^H \right), \forall k$.
Finally, we have
\begin{equation}
	\begin{split}
		&\rho\left( \tilde{\bA}^H\tilde{\bA} \right) = \rho\left(\tilde{\bA}\tilde{\bA}^H\right) 
		= F_vF_hN_r\rho\left(\bK\right)\\
		\overset{\left(\mathrm{a}\right)}{\le}& F_vF_hN_r\sum_{k=1}^{K}\rho\left(\bF\bF^H\right) = KF_vF_hN_r\rho\left(\bF\bF^H\right),
	\end{split}
\end{equation}
where $\left(\textrm{a}\right)$ comes from \cite[6.70 (a), pp 116]{seber2008matrix} and $\bX_k\bF\bF^H\bX_k^H$ is positive semi-definite.
Similarly, we can obtain 
\begin{align}
	\rho\left( \bF\bF^H \right) &= \rho\left(\bF^H\bF\right) = \rho\left( 		 \tilde{\bI}_{F_\tau N_p\times F_\tau N_f}^T \bF_d^H\bF_d\tilde{\bI}_{F_\tau N_p\times F_\tau N_f} \right) \nonumber \\
	&\le \rho\left(  \bF_d^H\bF_d \right) = F_\tau N_P ,
\end{align}
and hence,
\begin{equation}
	\rho\left(\bA^H\bA\right)\le KF_vF_hF_\tau N_r N_p = KF_vF_hF_\tau N.
\end{equation}
From a similar process in Appendix \ref{proof:radius of NI-A^HA}, we can obtain 
\begin{equation}\label{equ:aux1 in appH}
	\rho\left(N\bI-\bA^H\bA\right) \le KF_vF_hF_\tau N - N.
\end{equation}
Substituting \eqref{equ:aux1 in appH} into the right hand side of \eqref{equ:range of d 0}, we have 
\begin{equation}
	\frac{2}{1 + \frac{\rho\left( N\bI - \bA^H\bA \right)}{N}} \ge \frac{2}{KF_vF_hF_\tau}.
\end{equation}
In this case, if $d < \frac{2}{KF_vF_hF_\tau}$, then EIGA converges.
This completes the proof.

\section{Proof of Theorem \ref{lemma:SR in mMIMO-CE 2}}\label{proof:SR in mMIMO-CE 2}
From the definitions \eqref{equ:A in APSP}, it is not difficult to obtain that 
\begin{align}
	\rho\left(\bA^H\bA\right) \le \rho\left(\tilde{\bA}_p^H\tilde{\bA}_p\right) = \rho\left(  \tilde{\bA}_p\tilde{\bA}_p^H\right)= F_vF_hF_\tau N. 
\end{align}
Hence, we can obtain that 
\begin{equation*}
	\rho\left( N\bI - \bA^H\bA \right) \le \left( F_vF_hF_\tau -1 \right)N. 
\end{equation*}
Similarly, substituting the above range into the right hand side of \eqref{equ:range of d 0}, we have 
\begin{equation}
	\frac{2}{1 + \frac{\rho\left( N\bI - \bA^H\bA \right)}{N}} \ge \frac{2}{F_vF_hF_\tau}.
\end{equation}
In this case, if $d < \frac{2}{F_vF_hF_\tau}$, then EIGA converges.
This completes the proof.

\section{Proof of Lemma \ref{lem:m and e-condition of RIGA}}\label{proof:m and e-condition of RIGA}
	Given the fixed points $\bvartheta_{0}^{\star} = \funf[\btheta_{0}^{\star},\bbnu_0^{\star}]$, $\overlinebvartheta^{\star} = \funf[\overlinebtheta^{\star}\overlinebnu^{\star}]$ and $\tildebvartheta_{0n}^{\star} = \funf[\tildebtheta_{0n}^{\star},\tildebnu_{0n}^{\star}], n\in\setposi{N}$.
	From the definitions in \eqref{equ:definition of fixed point of RIGA1}, \eqref{equ:mu_0 and Sigma_0} and \eqref{equ:mu_n and Sigma_n}, we have
	\begin{subequations}
		\begin{equation}\label{equ:mu_0^star}
			\bmu_0^{\star} = \frac{1}{2}\bSigma_{0}^{\star}\btheta_{0}^{\star},
		\end{equation}
		\begin{equation}\label{equ:Sigma_0^star}
			\bSigma_{0}^{\star} = \left(  \bD^{-1} - \Diag{\bbnu_0^{\star}} \right)^{-1},
		\end{equation}
		\begin{equation}\label{equ:mu_0n^star}
			\bmu_{0n}^{\star} = \frac{1}{2}\bSigma_{0n}^{\star}{\btheta}_{0n}^{\star}, n\in\setposi{N},
		\end{equation}
		\begin{equation}\label{equ:Sigam_0n^star}
			\bSigma_{0n}^{\star} = \left(  \bD^{-1} - \Diag{{\bbnu}_{0n}^{\star}} \right)^{-1}, n\in\setposi{N},
		\end{equation}
		\begin{equation}\label{equ:mu_n^star}
			\bmu_{n}^{\star} = \bSigma_{n}^{\star}\left( \frac{y_n}{\tilde{\sigma}_z^2}\bgamma_n + \frac{1}{2}\overlinebtheta^{\star} \right), n\in\setposi{N},
		\end{equation}
		\begin{equation}\label{equ:Sigma_n^star}
			\bSigma_n^{\star} \equaa \bLambda^{\star} - \frac{1}{\beta^{\star}}\bLambda^{\star}\bgamma_n\bgamma_n^H\bLambda^{\star}, n\in\setposi{N},
		\end{equation}
	\end{subequations}
	where $\left(\textrm{a}\right)$ comes from that the magnitudes of the components in $\bA$ are $1$, and $\bLambda^{\star}$ and $\beta^{\star}$ are given by \eqref{equ:Lambda^star first} and \eqref{equ:beta^star first}, respectively.
	Note that the noise variance $\sigma_z^2$ is replaced with  $\tilde{\sigma}_z^2$ in \eqref{equ:mu_n^star} since its input noise variance is $\tilde{\sigma}_z^2$. 
	Then, we can obtain
	\begin{equation}\label{equ:m-condition of variance of RIGA1}
		\begin{split}
			\diag{\bSigma_{0}^{\star}} &= \diag{\left( \bD^{-1} - \Diag{\bbnu_0^{\star}} \right)^{-1}}\\
			&\equaa \diag{\left( \bD^{-1} - \Diag{{\bbnu}_{0n}^{\star}} \right)^{-1}}\\
			&\equab \diag{\bSigma_{0n}^{\star}} \equac \diag{\bSigma_{n}^{\star}}, n\in\setposi{N},
		\end{split}
	\end{equation}
	where $\left(\textrm{a}\right)$ comes from  \eqref{equ:m-condition of RIGA2},
	$\left(\textrm{b}\right)$ comes from \eqref{equ:Sigam_0n^star}, $\left(\textrm{c}\right)$ comes from that $p_0\left(\bh;\tildebvartheta_{0n}^{\star}\right)$ is the $m$-projection of $p_n\left(\bh;\overlinebvartheta^{\star}\right)$ and thus \eqref{equ:mp invariant} holds.
	Then, we can obtain 
	\begin{equation}
		\bmu_0^{\star} \equaa \frac{1}{2}\bSigma_{0}^{\star}\btheta_{0}^{\star} \equab \frac{1}{2N}\sum_{n=1}^{N}\bSigma_{0n}^{\star}\tildebtheta_{0n}^{\star} \overset{\left(\textrm{c}\right)}{=}\!\frac{1}{N}\sum_{n=1}^{N}\bmu_{0n}^{\star} \overset{\left(\textrm{d}\right)}{=}\!\frac{1}{N}\sum_{n=1}^{N}\bmu_{n}^{\star},
	\end{equation}
	where $\left(\textrm{a}\right)$ comes from \eqref{equ:mu_0^star}, $\left(\textrm{b}\right)$ comes from  \eqref{equ:m-condition of RIGA1 new} and \eqref{equ:m-condition of variance of RIGA1}, $\left(\textrm{c}\right)$ comes from \eqref{equ:mu_0n^star},
	and $\left(\textrm{d}\right)$ comes from that $p_0\left(\bh;\tildebvartheta_{0n}^{\star}\right)$ is the $m$-projection of $p_n\left(\bh;\overlinebvartheta^{\star}\right)$ and thus \eqref{equ:mp invariant} holds.
	This completes the proof.

\section{Proof of Theorem \ref{the:the expression of mu_0^star in S-IGA}}\label{proof:the expression of mu_0^star in S-IGA}

From Theorem \ref{The-theta_1 conv} and \eqref{equ:m-condition of RIGA2}, we can obtain $\bbnu_0^{\star} < 0$.
From Lemma \ref{lem:m and e-condition of RIGA}, we have
\begin{align}
	\bmu_0^{\star} &= \frac{1}{N}\bmu_n^{\star}
	\equaa \frac{1}{2N}\sum_{n=1}^{N}\bSigma_{n}^{\star}\left(\overlinebtheta^{\star} + \frac{2y_n}{\tilde{\sigma}_z^2}\bgamma_n\right)\nonumber \\
	&\equab \frac{1}{2N}\sum_{n=1}^{N}\bSigma_{n}^{\star}\left(\frac{N-1}{N}\btheta_0^{\star} + \frac{2y_n}{\tilde{\sigma}_z^2}\bgamma_n\right)\\
	&\equac \underbrace{		\frac{N-1}{N^2}\sum_{n=1}^{N}\bSigma_{n}^{\star}\left(\bSigma_{0}^{\star}\right)^{-1}}_{\bQ}\bmu_0^{\star} + \underbrace{\frac{1}{N\tilde{\sigma}_z^2}\sum_{n=1}^{N}\bSigma_{n}^{\star}\bgamma_n y_n}_{\bq}\nonumber\\
	&=\bQ\bmu_0^{\star} + \bq,\nonumber
\end{align}
where $\left(\textrm{a}\right)$ comes from \eqref{equ:mu_n^star}, $\left(\textrm{b}\right)$ comes from the $e$-condition in \eqref{equ:m and e-conditions of RIGA}, and $\left(\textrm{c}\right)$ comes from \eqref{equ:mu_0^star}. Combining \eqref{equ:Sigma_n^star}, $\bq$ can be expressed as 
\begin{align}
	\bq &= \frac{1}{N\tilde{\sigma}_z^2}\bLambda^{\star}\sum_{n=1}^{N}\left(\bI - \frac{1}{\beta^{\star}}\bgamma_n\bgamma_n^H\bLambda^{\star}\right)\bgamma_ny_n \nonumber \\
	&\equad \frac{1}{N\tilde{\sigma}_z^2}\bLambda^{\star}\left( \bA^H\by - \sum_{n=1}^{N}\frac{\bgamma_n^H\bLambda^{\star}\bgamma_n}{\beta^{\star}}\bgamma_ny_n \right)\\
	&\equae \frac{1}{N\tilde{\sigma}_z^2}\bLambda^{\star}\left(1-\frac{\tr{\bLambda^{\star}}}{\beta^{\star}}\right)\bA^H\by
	\equaf \frac{1}{N\beta^{\star}}\bLambda^{\star}\bA^H\by,\nonumber
\end{align}
where $\left(\textrm{d}\right)$ comes from the definition of $\bgamma_n$ in \eqref{equ:gamma_n}, i.e., $\bA^H = \left[ \bgamma_1, \ \bgamma_2, \ \cdots, \ \bgamma_N \right]$, $\left(\textrm{e}\right)$ comes from that the magnitudes of the components in $\bA$ are $1$ and $\bA^H = \left[ \bgamma_1, \ \bgamma_2, \ \cdots, \ \bgamma_N \right]$, and $\left(\textrm{f}\right)$ comes from $\eqref{equ:beta^star first}$.
Meanwhile, $\bQ$ can be expressed as 
\begin{align}
	\bQ &= \frac{N-1}{N^2}\sum_{n=1}^{N}\left(\bI - \frac{1}{\beta^{\star}} \bLambda^{\star}\bgamma_n\bgamma_n^H \right)\bLambda^{\star}\left( \bSigma_{0}^{\star} \right)^{-1}\nonumber\\
	&\overset{\left(\textrm{g}\right)}{=}\frac{N-1}{N^2}\left( N\bI - \frac{1}{\beta^{\star}}\bLambda^{\star}\bA^H\bA \right)\bLambda^{\star}\left( \bSigma_{0}^{\star} \right)^{-1}\nonumber\\
	&\overset{\left(\textrm{h}\right)}{=} \frac{N-1}{N^2}\left( N\bI - \frac{1}{\beta^{\star}}\bLambda^{\star}\bA^H\bA \right)\bLambda^{\star} \\
	&\ \ \ \ \times \left\lbrace \frac{1}{N-1}\left[N\left( \bD^{-1} - \Diag{\overlinebnu^{\star}}\right) - \bD^{-1} \right] \right\rbrace \nonumber \\
	&\overset{\left(\textrm{i}\right)}{=}\frac{1}{N^2}\left( N\bI - \frac{1}{\beta^{\star}}\bLambda^{\star}\bA^H\bA \right)\left( N\bI - \bLambda^{\star}\bD^{-1} \right)\nonumber\\
	&= \bI - \frac{1}{N}\bLambda^{\star}\bD^{-1} -\frac{1}{N\beta^{\star}}\bLambda^{\star}\bA^H\bA\left( \bI - \frac{1}{N}\bLambda^{\star}\bD^{-1} \right),\nonumber
\end{align}
where $\left(\textrm{g}\right)$ comes from $\bA^H = \left[ \bgamma_1 \ \bgamma_2 \ \ldots \ \bgamma_N \right]$, $\left(\textrm{h}\right)$ comes from \eqref{equ:Sigma_0^star} and $e$-condition in \eqref{equ:m and e-conditions of RIGA}, and $\left(\textrm{i}\right)$ comes from \eqref{equ:Lambda^star first}. Thus, we have
\begin{align}\label{equ:mu_0^star final}
	\bmu_0^{\star} &= \left(\bI-\bQ\right)^{-1}\bq\nonumber\\
	&= \left(  \frac{1}{N}\bLambda^{\star}\bD^{-1} +\frac{1}{N\beta^{\star}}\bLambda^{\star}\bA^H\bA\left( \bI - \frac{1}{N}\bLambda^{\star}\bD^{-1} \right) \right)^{-1}\nonumber\\
	&\ \ \ \ \times \frac{1}{N\beta^{\star}}\bLambda^{\star}\bA^H\by\nonumber\\
	&=\bD\left[ \bA^H\bA\left( \bD- \frac{1}{N}\bLambda^{\star} \right) + \beta^{\star}\bI  \right]^{-1}\bA^H\by.
\end{align}
We then show that the matrix inversion above is valid.
From \eqref{equ:Lambda^star first}, we can obtain $0<\left[ \bLambda^{\star} \right]_{i,i} < \left[\bD\right]_{i,i}$,
since we have $\overlinebnu^{\star}<0$. Then, we have $\bK\triangleq\bD - \frac{1}{N}\bLambda^{\star} \succ 0$.
From \eqref{equ:beta^star first}, we have 
\begin{equation}\label{equ:bound of beta^star}
	0<\tilde{\sigma}_z^2<\beta^{\star}<\tilde{\sigma}_z^2 + \tr{\bD}.
\end{equation}
Thus, we can obtain
\begin{equation}
	\left[ \bA^H\bA\left( \bD- \frac{1}{N}\bLambda^{\star} \right) + \beta^{\star}\bI  \right]^{-1} \!\!\!\!\!\!\!= \bK^{-1}\left( \bA^H\bA + \beta^{\star}\bK^{-1} \right)^{-1}.
\end{equation}
Hence, the matrix above is invertible and \eqref{equ:mu_0^star final} is valid. This completes the proof.

\section{Proof of Lemma \ref{lemma:bijection of f}}\label{proof:bijection of f}
$f$ can be expressed as 
\begin{equation}
	f = x - \sum_{i=1}^{M}\frac{d_ix}{x + d_i\left(N-1\right)},
\end{equation}
where $d_i = \left[ \bD \right]_{i,i} > 0, i\in\setposi{M}$.
Then, the derivative  of $f$ satisfies 
\begin{equation}
	\frac{\intd f}{\intd x} = 1-\frac{1}{N-1}\sum_{i=1}^{M}\left( \frac{d_i}{d_i + \frac{x}{N-1}} \right)^2 \overset{\left(\textrm{a}\right)}{>}1-\frac{M}{N-1},
\end{equation}
where $\left(\textrm{a}\right)$ comes from $x >0$ and $d_i>0, i\in\setposi{M}$. If $M<N$ ($M$ and $N$ are both integers), then $f$ is a monotonically increasing function.
From $f\left(0\right) = 0$, we can obtain $f > 0$ when $x >0$.
This completes the proof.

\section{Proof of Theorem \ref{the:asymptotic value of mu_0 under case 2}}\label{proof:asymptotic value of mu_0 under case 2}
We first derive the asymptotic value of $\left[\bLambda^{\star}\right]_{i,i}, i\in\setposi{M}$, and $\beta^{\star}$ when $N$ tends to infinity and $M < N$. 
\begin{align}\label{equ:nu0 divide N}
	&\frac{1}{N}\Diag{\bbnu_0^{\star}} \equaa \Diag{\bbnu_0^{\star} - \overlinebnu^{\star}} 
	\equab \left( \bLambda^{\star} \right)^{-1} - \left( \bSigma_{0}^{\star} \right)^{-1} \nonumber \\
	\equac& \left( \bLambda^{\star} \right)^{-1} - \left( \bI \odot \bSigma_{n}^{\star} \right)^{-1}
	\equad \left( \bLambda^{\star} \right)^{-1} - \left( \bLambda^{\star} - \frac{1}{\beta^{\star}}\left( \bLambda^{\star}\right)^2  \right)^{-1}\nonumber\\
	=&-\frac{1}{\beta^{\star}}\left(\bI - \frac{1}{\beta^{\star}}\bLambda^{\star} \right)^{-1},
\end{align}
where $\left(\textrm{a}\right)$ comes from the $e$-condition in \eqref{equ:m and e-conditions of RIGA}, $\left(\textrm{b}\right)$ comes from \eqref{equ:Sigma_0^star} and \eqref{equ:Lambda^star first},  $\left(\textrm{c}\right)$ comes from
$\diag{\bSigma_{0}^{\star}} = \diag{\bSigma_{n}^{\star}}, n\in\setposi{N}$, in Lemma \ref{lem:m and e-condition of RIGA},
and $\left(\textrm{d}\right)$ comes from \eqref{equ:Sigma_n^star}, the magnitudes of the components in $\bA$ are $1$ and \cite[Equation 11.42, pp 252]{seber2008matrix}.
We then show that when $N$ tends to infinity, each diagonal component in $\bLambda^{\star}$ tends to $0$. Since we have $\overlinebnu^{\star} <0$, then, according to \eqref{equ:Lambda^star first} and \eqref{equ:beta^star first}, we have $\diag{\bLambda^{\star}} > 0$ and $\diag{\bLambda^{\star}}<\beta^{\star}$. Then, we can obtain
$0<\left[ \bI - \frac{1}{\beta^{\star}}\bLambda^{\star} \right]_{i,i}<1$ and $	1<\left[  \left( \bI - \frac{1}{\beta^{\star}}\bLambda^{\star} \right)^{-1} \right]_{i,i}, i\in\setposi{M}$. Combining \eqref{equ:nu0 divide N}, we have 
\begin{equation}
	\frac{1}{N}\bbnu_0^{\star} < -\frac{1}{\beta^{\star}}.
\end{equation}
Then, from the $e$-condition in \eqref{equ:m and e-conditions of RIGA}, we have
\begin{equation}
	\overlinebnu^{\star} = \frac{N-1}{N}\bbnu_0^{\star} < -\frac{N-1}{\beta^{\star}}<0.
\end{equation}
Since $\bD$ is positive definite diagonal, from \eqref{equ:Lambda^star first}, we can obtain
\begin{equation}\label{equ:inequality of Lambada}
	0<	\left[ \bLambda^{\star} \right]_{i,i} < -\frac{1}{\left[\overlinebnu^{\star}\right]_i}<\frac{\beta^{\star}}{N-1}\overset{\left(\textrm{a}\right)}{<} \frac{\tilde{\sigma}_z^2 + \tr{\bD}}{N-1}, 
\end{equation}
where $i\in\setposi{M}$ and
$\left(\textrm{a}\right)$ comes from \eqref{equ:bound of beta^star}.
Thus, we obtain $\lim\limits_{N\to \infty}\left[ \bLambda \right]_{i,i} = 0$.

We then show the asymptotic value of $f\left(\beta^{\star}\right)$.
From \eqref{equ:bound of beta^star}, it is readily obtained that $\beta^{\star} > 0$, and thus, $f\left(\beta^{\star}\right)$ is valid.
Then, we can obtain the following relationship
\begin{equation}\label{equ:range of nu_0 divide N new}
	-\frac{N-1}{\left(N-2\right)\beta^{\star}}\mathbf{1}\overset{\left(\textrm{a}\right)}{<}  \frac{\bbnu_0^{\star}}{N} \equab \frac{\overlinebnu^{\star}}{N-1} \overset{\left(\textrm{c}\right)}{<} -\frac{1}{\beta^{\star}}\mathbf{1},
\end{equation}
where $\mathbf{1}$ is the all-one vector,
$\left(\textrm{a}\right)$ comes from $0<\beta^{\star}$ in \eqref{equ:bound of beta^star}, the last equation in \eqref{equ:nu0 divide N} and $\left[\bLambda^{\star}\right]_{i,i} < \beta^{\star}/\left(N-1\right)$ in \eqref{equ:inequality of Lambada}, 
$\left(\textrm{b}\right)$ comes from the $e$-condition in \eqref{equ:m and e-conditions of RIGA},
and $\left(\textrm{c}\right)$ comes from $-1/\left[ \bbnu^{\star} \right]_{i} < \beta^{\star}/\left(N-1\right)$ in \eqref{equ:inequality of Lambada}, 
$0<\beta^{\star}$ and $\bbnu^{\star} < 0$ in Theorem \ref{the:the expression of mu_0^star in S-IGA}. 
Combining $\beta^{\star}$ in \eqref{equ:Lambda^star and beta^star} and the relationship in  \eqref{equ:range of nu_0 divide N new},
we define three functions as
\begin{subequations}\label{equ:f1 and f2}
	\begin{equation}\label{equ:f0}
		f_0 \triangleq \tr{ \left( \bD^{-1} + \frac{\left(N-1\right)^2}{\left(N-2\right)\beta^{\star}} \bI\right)^{-1} },
	\end{equation}
	\begin{equation}
		f_1 \triangleq \tr{\bLambda^{\star}} = \tr{\left( \bD^{-1}-\Diag{\overlinebnu^{\star}} \right)^{-1}},
	\end{equation}
	\begin{equation}\label{equ:function f2}
		f_2\left(\beta^{\star}\right) \triangleq
		\tr{\left( \bD^{-1} + \frac{N-1}{\beta^{\star}}\bI \right)^{-1}}.
	\end{equation}
\end{subequations}
Write $f_2$ as the function of $\beta^{\star}$ since we will use this form in the following.
From  \eqref{equ:range of nu_0 divide N new}, it is not difficult to show that  $f_0 < f_1 < f_2$. 
Thus, we have $0<f_2-f_1$ and $f_0-f_1<0$.
From \eqref{equ:f0}, we can obtain 
\begin{align}\label{equ:f_0 - f_2}
	f_0 &= \tr{\left(  \bD^{-1} + \frac{N-1}{\beta^{\star}}\bI + \frac{N-1}{\left(N-2\right)\beta^{\star}}\bI \right)^{-1}} \\\nonumber
	&\overset{\left(\textrm{a}\right)}{>}\tr{ \bL^{-1} - \frac{N-1}{\left(N-2\right)\beta^{\star}}\bL^{-2} }\\
	&\overset{\left(\textrm{b}\right)}{>}f_2 - \frac{N-1}{\left(N-2\right)\beta^{\star}}\tr{\left( \frac{N-1}{\beta^{\star}}\bI \right)^{-2}} \nonumber\\
	&= f_2 - \frac{M\beta^{\star}}{\left(N-1\right)\left(N-2\right)},\nonumber
\end{align}
where $	\bL = \left( \bD^{-1} + \frac{N-1}{\beta^{\star}}\bI \right)$,
$\left(\textrm{a}\right)$ comes from $\left(a+b\right)^{-1}>a^{-1} - a^{-2}b $ with $a,b >0$ and $\left(\textrm{b}\right)$ comes from that $\bD$ is positive definite. 
Then, we can obtain
\begin{equation}
	f_2 -  f_1  \overset{\left(\textrm{a}\right)}{<}  f_0 + \frac{M\beta^{\star}}{\left(N-1\right)\left(N-2\right)} - 
	f_1
	\overset{\left(\textrm{b}\right)}{<}\!\frac{M\left( \tilde{\sigma}_z^2 + \tr{\bD} \right)}{\left(N-1\right)\left(N-2\right)},
\end{equation}
where $\left(\textrm{a}\right)$ comes from \eqref{equ:f_0 - f_2}, and $\left(\textrm{b}\right)$ comes from $f_0 - f_1<0$ and \eqref{equ:bound of beta^star}.
From $f_2-f_1>0$ and
$M < N$, we have $\lim\limits_{N\to \infty}\left(f_2\left(\beta^{\star}\right) - f_1\right) = 0$.
From \eqref{equ:beta^star first} and $\tilde{\sigma}_z^2 = f\left(\sigma_z^2\right)$, we can immediately obtain 
\begin{equation}
	\lim\limits_{N,M\to \infty}\beta^{\star} = \lim\limits_{N,M\to \infty}f\left(\sigma_z^2\right) + \lim\limits_{N,M\to \infty}f_2,
\end{equation} 
and thus,  
\begin{equation}\label{equ:case 2}
	\lim\limits_{N\to \infty}f\left(\beta^{\star}\right)   = \lim\limits_{N\to \infty}\left(\beta^{\star} - f_2\right) 
	= \lim\limits_{N\to \infty}f\left(\sigma_z^2\right).
\end{equation}
Combining Lemma \ref{lemma:bijection of f}, when $M < N$ we can obtain 
\begin{equation}\label{equ:asy value of beta^star under case 2}
	\lim\limits_{N\to \infty}\beta^{\star} = \sigma_z^2.
\end{equation}
Combining  \eqref{equ:asy value of beta^star under case 2}, $\lim\limits_{N\to \infty}\left[\bLambda^{\star}\right]_{i,i} = 0, i\in \setposi{M}$, and $\bmu_{0}^\star$ in Theorem \ref{the:the expression of mu_0^star in S-IGA}, we have 
\begin{equation}
	\lim\limits_{N\to \infty}\bmu_{0}^{\star}  = \bD\left( \bA^H\bA\bD + \sigma_z^2\bI  \right)^{-1}\bA^H\by = 		\tilde{\bmu}.
\end{equation}

From \eqref{equ:function f beta} and \eqref{equ:function f2}, we have $f\left(x\right) = x- f_2\left(x\right)$, where $x>0$.     
We show that when $M$ is fixed, 
\begin{equation*}
	\lim\limits_{N\to \infty}f_2\left(x\right) = 0, x>0.
\end{equation*}
Denote $d_i \triangleq \left[\bD\right]_{i,i} >0, i\in \setposi{M}$, $d_{\textrm{min}}$ and $d_{\textrm{max}}$ as the minimum and the maximum of $\braces{d_i}_{i=1}^M$.
We have 
\begin{equation*}
	f_2\left(x\right) = \sum_{i=1}^{M}\frac{d_i x}{\left(N-1\right)d_i + x }.
\end{equation*}
Treating $f_2$ as a function of $\braces{d_i}_{i=1}^M$, we have
\begin{equation}
	\frac{\partial f_2}{\partial d_i} = \frac{\left(N-1\right)d_i^2}{\left[ x + \left(N-1\right)d_i \right]^2} > 0, i\in \setposi{M}.
\end{equation}
Thus, we can obtain $f_2^{\min} < f_2\left(x\right) < f_2^{\max}$, where
where 
\begin{subequations}
	\begin{equation}
			f_2^{\min} = \frac{Md_{\textrm{min}}x}{\left(N-1\right)d_\textrm{min} + x},
		\end{equation}
\begin{equation}
	f_2^{\max} = \frac{Md_{\textrm{max}}x}{\left(N-1\right)d_\textrm{max} + x}.
\end{equation}
\end{subequations}
When $M$ is fixed, $d_{\min}$ and $d_{\max}$ are also fixed, and thus we have $\lim\limits_{N\to \infty}f_2\left(x\right) = 0$.
Hence, we can obtain $\lim\limits_{N\to \infty}f\left(\sigma_z^2\right) = \sigma_z^2$.
When $\tilde{\sigma}_z^2 = \sigma_z^2$, from $\lim\limits_{N\to \infty}\left(f_2\left(\beta^{\star}\right) - f_1\right) = 0$ and \eqref{equ:beta^star first}, we can readily obtain 
$\lim\limits_{N\to \infty}\beta^{\star} = \sigma_z^2 + \lim\limits_{N\to \infty}f_2\left(\beta^{\star}\right) = \sigma_z^2$.
Similar to the previous process, we can obtain $\lim\limits_{N\to \infty}\bmu_{0}^{\star} = \tilde{\bmu}$.
This completes the proof.
\end{document}